%% file: Solution-GQSL-arXiv-v3.tex
\documentclass[11pt,notitlepage,tightenlines,nofootinbib,superscriptaddress]{revtex4-2}

\input{macros}

\usepackage[most,breakable]{tcolorbox}
\renewenvironment{boxed}[1]{\expandafter\ifstrequal\expandafter{#1}{orange}{\begin{tcolorbox}[colback=red!15,colframe=orange!15,breakable,enhanced]}{\begin{tcolorbox}[colback=Blues5seq1,colframe=Blues5seq5,breakable,enhanced]}}{\end{tcolorbox}}

\newcommand{\rel}[3]{#1\big(#2\,\big\|\,#3\big)}
\newcommand{\hyp}[3]{H_{#1\!,\,#2;\, #3}}
\newcommand{\LL}{\pazocal{L}}
\newcommand{\OO}{\pazocal{O}}
\newcommand{\ane}{\mathrm{ANE}}
\newcommand{\lowarrow}[1]{\text{\raisebox{-2.3pt}{$\overset{\text{\raisebox{-1.6pt}{$#1$}}}{\longrightarrow}$}}\,}
\DeclareMathOperator{\stein}{Stein}
\newcommand{\VV}{\pazocal{V}}
\newcommand{\raisemath}[1]{\mathpalette{\raisem@th{#1}}}
\newcommand{\raisem@th}[3]{\raisebox{#1}{$#2#3$}}
\makeatother
\newcommand{\aiid}[1]{\pazocal{A}^{\raisemath{0.5pt}{#1}}}
\newcommand{\VVsym}{\pazocal{V}_{\mathrm{Sym}}}

\begin{document}

\title{A solution of the generalised quantum Stein's lemma}

\author{Ludovico Lami}
\email{ludovico.lami@gmail.com}
\affiliation{Scuola Normale Superiore, Piazza dei Cavalieri 7, 56126 Pisa, Italy}
\affiliation{QuSoft, Science Park 123, 1098 XG Amsterdam, the Netherlands}
\affiliation{Korteweg--de Vries Institute for Mathematics, University of Amsterdam, Science Park 105-107, 1098 XG Amsterdam, the Netherlands}
\affiliation{Institute for Theoretical Physics, University of Amsterdam, Science Park 904, 1098 XH Amsterdam, the Netherlands}

\begin{abstract}
We solve the generalised quantum Stein's lemma, proving that the Stein exponent associated with entanglement testing, namely, the quantum hypothesis testing task of distinguishing between $n$ copies of an entangled state $\rho_{AB}$ and a generic separable state $\sigma_{A^n:B^n}$, equals the regularised relative entropy of entanglement. Not only does this determine the ultimate performance of entanglement testing, but it also establishes the reversibility of all quantum resource theories under asymptotically resource non-generating operations, with the regularised relative entropy of resource governing the asymptotic transformation rate between any two quantum states. As a by-product, we prove that the same Stein exponent can also be achieved when the null hypothesis is only approximately i.i.d., in the sense that it can be modelled by an `almost power state'. To solve the problem we introduce two techniques. The first is a procedure that we call `blurring', which, informally, transforms a permutationally symmetric state by making it more evenly spread across nearby type classes. Blurring alone suffices to prove the generalised Stein's lemma in the fully classical case, but not in the quantum case. Our second technical innovation, therefore, is to perform a second quantisation step to lift the problem to an infinite-dimensional bosonic quantum system; we then solve it there by using techniques from continuous-variable quantum information. Rather remarkably, the second-quantised action of the blurring map corresponds to a pure loss channel. A careful examination of this second quantisation step is the core of our quantum solution. 
\end{abstract}

\maketitle

\section{Introduction} \label{sec_intro}

Quantum entanglement, initially regarded as a mathematical oddity of quantum theory~\cite{EPR, schr}, is increasingly recognised as a fundamental concept of quantum mechanics~\cite{Horodecki-review}, and as one of the sources of its technological power. It fuels quantum teleportation~\cite{teleportation} and superdense coding~\cite{dense-coding}, it enables the violation of Bell inequalities~\cite{Brunner-review}, it plays a critical role in quantum key distribution~\cite{Ekert91, RennerPhD}, and it enhances quantum communication on noisy channels~\cite{superactivation, Hastings2008}. Recently, it has been proposed that its detection in a gravitationally interacting system would provide evidence in favour of the quantisation of gravity~\cite{tabletop, testing}.

\subsection{Entanglement testing}

All of the above applications require the detection of entanglement in unknown states, as this is a necessary step, e.g., in certifying the devices that produce it. This leads us to the key operational task that will be the central focus of this paper: \deff{entanglement testing}. Imagine that you bought a device that --- so is promised to you --- produces many copies of a possibly mixed bipartite entangled state $\rho_{AB}$. If that were actually the case, after using it $n$ times the global state would of the form $\rho_{AB}^{\otimes n}$. You however suspect that the device might be faulty or even maliciously engineered, so that, in reality, it produces a global \emph{separable} (that is, unentangled~\cite{Werner}) state $\sigma_{A^n:B^n}$ --- here we stress that separability is between $n$ copies of $A$ on one side and $n$ copies of $B$ on the other. How do you decide whether to trust the device or not? See Figure~\ref{entanglement_testing_fig} for a pictorial representation of the problem. We can think of this question --- one of the most fundamental of entanglement theory --- as an important special case of quantum hypothesis testing. We call it `entanglement testing'. As in traditional quantum hypothesis testing, there are two hypotheses to decide between:
\begin{itemize}
\item Null hypothesis: the unknown state is $\rho_{AB}^{\otimes n}$.
\item Alternative hypothesis: the unknown state is some $\sigma_{A^n:B^n}$, of which nothing is known except that it is separable across the cut $A^n:B^n$.
\end{itemize}
Unlike in traditional quantum hypothesis testing, which features two \emph{independent and identically distributed} (i.i.d.) hypotheses~\cite{Hiai1991, Ogawa2000, Nussbaum2009, qChernoff, Audenaert2008}, in entanglement testing the alternative hypothesis is \emph{composite}, that is, non-i.i.d.; this feature complicates the analysis considerably and has hindered progress on the problem for many years. To this day, only a handful of results on composite quantum hypothesis testing are known~\cite{Bjelakovic2005, Noetzel2014, brandao_adversarial, berta_composite}.

As in every hypothesis testing problem, in entanglement testing there are two types of errors:
\begin{itemize}
\item Type-I error: the state was $\rho_{AB}^{\otimes n}$, but we incorrectly guessed that it was separable.
\item Type-II error: the state was separable, but we incorrectly guessed that it was $\rho_{AB}^{\otimes n}$.
\end{itemize}
There is clearly a trade-off between the probabilities of these two errors occurring: since the above events are often not equally consequential in practical applications, in asymmetric hypothesis testing one tries to minimise one of the two while keeping the other at most equal to $\e\in (0,1)$. When this is done, the minimised probability usually decays \emph{exponentially fast} to zero as a function of $n$, the number of available copies. The \emph{coefficients} governing these exponential decays are known as error exponents. The rate of decay of the type-II error probability is usually called the \deff{Stein exponent}, while the rate of decay of the type-I error probability is known as the \deff{Sanov exponent}. Formally,
\begin{align}
\stein(\rho\,\|\,\sep) &\coloneqq \lim_{\e\to 0^+} \liminf_{n\to\infty} -\frac1n \log \min\left\{ \pr\{\text{$n$-copy type-II error}\}\!:\ \pr\{\text{type-I error}\} \leq \e\right\} , \label{Stein_informal} \\
\mathrm{Sanov}(\rho\,\|\,\sep) &\coloneqq \lim_{\e\to 0^+} \liminf_{n\to\infty} -\frac1n \log \min\left\{ \pr\{\text{$n$-copy type-I error}\}\!:\ \pr\{\text{type-II error}\} \leq \e\right\} . \label{Sanov_informal} 
\end{align}
where $(\alpha_n,\beta_n)$ represent achievable pairs of type-I and type-II error probabilities in $n$-copy entanglement testing. In this paper we will establish an expression for~\eqref{Stein_informal}, while in~\cite{generalised-Sanov} we solve~\eqref{Sanov_informal}.

In the simple i.i.d.\ setting, where the alternative hypothesis is represented by some i.i.d.\ state $\sigma^{\otimes n}$, the Stein exponent is given by the quantum \deff{relative entropy}, defined as~\cite{Umegaki1962}
\bb
D(\rho\|\sigma) \coloneqq \Tr \left[\rho \left(\log \rho - \log \sigma\right)\right] .
\label{Umegaki}
\ee
In entanglement testing, however, the second state is not fixed; in this case, a natural guess for the Stein exponent is obtained by minimising the relative entropy with respect to the second argument --- which corresponds to considering a worst-case scenario. By doing so one defines the \deff{relative entropy of entanglement}, given by~\cite{Vedral1997, Vedral1998}
\bb
E_R(\rho_{AB}) = D\big(\rho_{AB} \big\| \SEP_{A:B}\big) \coloneqq \min_{\sigma_{AB} \in \SEP_{A:B}} D(\rho_{AB} \|\sigma_{AB})\, ,
\label{relative_entropy_entanglement}
\ee
where $\SEP_{A:B}$ denotes the set of separable or unentangled states on $AB$. To obtain an asymptotically meaningful expression, however, we need to consider the relative entropy of entanglement not of a single copy but of many copies of $\rho_{AB}$. The resulting measure, called the \deff{regularised relative entropy of entanglement}, is defined as
\bb
E_R^\infty(\rho_{AB}) = D^\infty\big(\rho_{AB} \big\| \SEP_{A:B}\big) \coloneqq \lim_{n\to\infty} \frac1n\, E_R\big(\rho_{AB}^{\otimes n}\big)\, ,
\label{regularised_relative_entropy_entanglement}
\ee
where it is understood that on the right-hand side the relevant bipartition is $A^n\!:\!B^n$. The regularisation is needed, as additivity violations, i.e.\ examples of states $\rho_{AB}$ where $E_R^\infty(\rho_{AB}) < E_R(\rho_{AB})$, are known~\cite{Werner-symmetry}.

\begin{center}
\begin{figure}[h!t] \centering
\includegraphics[scale=.193]{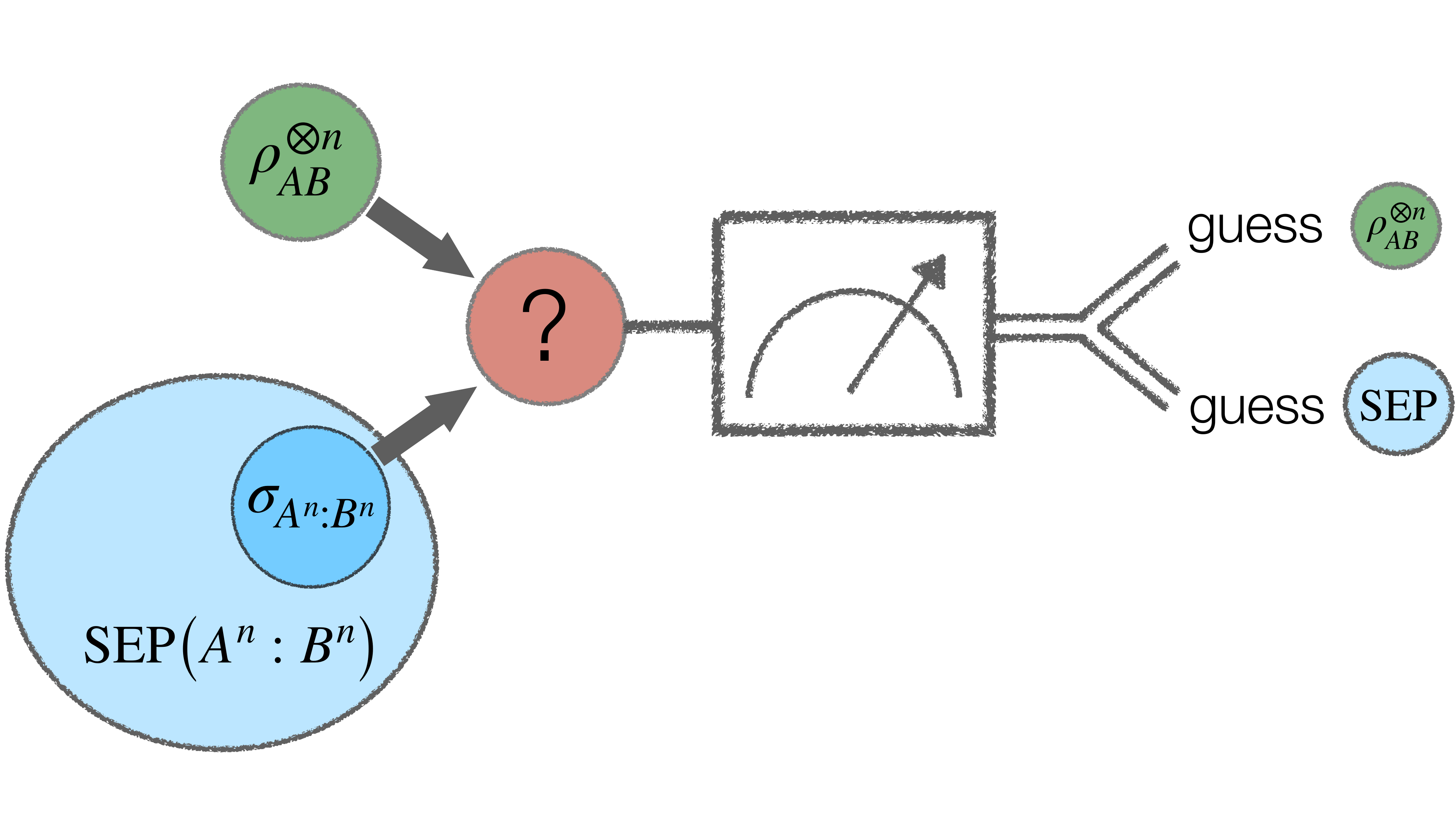}
\caption{\justifying  
The task of entanglement testing is a special case of quantum hypothesis testing. The null hypothesis is that the unknown quantum state is of the form $\rho_{AB}^{\otimes n}$, for some bipartite state $\rho_{AB}$. The alternative hypothesis is that it is a generic separable, i.e.\ unentangled, state $\sigma_{A^n:B^n}$ of the global system $A^n:B^n$. In order to guess which of these two hypotheses is the correct one, we can perform an arbitrary measurement on the unknown state. The fundamental difficulty in analysing this problem is that the set of states that constitutes the alternative hypothesis is non-i.i.d.\ and in fact potentially quantumly correlated, because although $\sigma_{A^n:B^n}$ is separable across the cut $A^n:B^n$, it might be entangled across the several $A$ systems, and similarly across the several $B$ systems.}
\label{entanglement_testing_fig}
\end{figure}
\end{center}

\subsection{Connection with asymptotic entanglement manipulation}

Entanglement testing is connected at a fundamental level with the task of entanglement distillation, a critical primitive of quantum technology platforms. Entanglement distillation consists in preparing as much \emph{pure} entanglement as possible starting from many copies of a possibly \emph{mixed} state $\rho_{AB}$ and using only `free operations' belonging to some restricted set of quantum operations $\OO$, which encapsulates the operational and technological constraints to which we are subjected~\cite{Bennett-distillation, Bennett-distillation-mixed, Bennett-error-correction, Horodecki-review}. The amount of pure entanglement produced by the process is measured by the number of `ebits' --- an ebit, represented by the two-qubit maximally entangled state $\ket{\Phi_2} \coloneqq \frac{1}{\sqrt2} \left(\ket{00} + \ket{11}\right)$, is the fundamental unit of entanglement.

The ultimate efficiency of distillation can be defined as the maximum number of ebits that can be extracted with operations in $\OO$ per copy of $\rho_{AB}$, in the asymptotic limit in which many copies of $\rho_{AB}$ are available. This figure of merit is called the \deff{$\OO$-distillable entanglement}, and it is denoted by $E_{d,\OO}(\rho)$. The reverse task to entanglement distillation, entanglement dilution~\cite{Bennett-distillation, Bennett-error-correction, Hayden-EC, Buscemi2011, EC-infinite}, is the process of consuming ebits to prepare mixed states, and it is also important~\cite{Acin2005, Miller2022, Berta2013, Wilde2018}. Its ultimate efficiency, the \deff{$\OO$-entanglement cost}, is denoted by $E_{c,\OO}(\rho)$. Historically, the first set of free operations to be studied has been that of LOCC, which models what distant parties can do when they have access to local quantum computers but can communicate only classical messages~\cite{Bennett-distillation, LOCC}. Although this setting is operationally very relevant, there are compelling reasons to go beyond it. First, both the LOCC-entanglement cost and especially LOCC-distillable entanglement are very difficult to compute or even approximate. This is due in part to the famous disproof~\cite{Hastings2008} of the additivity conjectures~\cite{Shor2004} in quantum information, and in part to the intricate mathematical structure of the LOCC set~\cite{LOCC}. To improve this state of affairs, we can study simpler classes of free operations than LOCC. In particular, larger classes $\OO \supseteq \locc$ lead to lower and upper bounds on the distillable entanglement and on the entanglement cost, because $E_{c,\,\OO}(\rho)\leq E_{c,\,\locc}(\rho)$ and $E_{d,\,\OO}(\rho)\geq E_{d,\,\locc}(\rho)$~\cite{negativity, plenioprl, Rains2001}. 

But there is a more fundamental reason to go beyond the LOCC paradigm. In~2005, M.\ B.\ Plenio asked whether it would be possible to formulate a theory of entanglement manipulation in which entanglement can be converted \emph{reversibly} between mixed and pure states~\cite{OpenProblemArxiv, OpenProblem}. Plenio's question was motivated by the analogy with thermodynamics~\cite{Popescu1997, vedral_2002, Horodecki2002}, in which work and heat can be reversibly interconverted by means of Carnot cycles~\cite{CARNOT}. A thermodynamical, reversible theory of entanglement would correspond to a set of free operations $\OO$ such that the $\OO$-distillable entanglement and the $\OO$-entanglement cost coincide for all states, i.e.
\bb
E_{d,\,\OO}(\rho_{AB}) \overset{?}{=} E_{c,\,\OO}(\rho_{AB})\qquad \forall\ \rho=\rho_{AB}\, .
\label{Plenio_problem}
\ee
\emph{The resulting entanglement measure would therefore be the unique asymptotic entanglement quantifier} within the theory, mimicking the role of entropy in thermodynamics and simplifying the overall theoretical landscape. Reversibility implies that the transformations $\Phi_2^{\otimes n} \lowarrow{\OO} \rho_{AB}^{\otimes m} \lowarrow{\OO} \Phi_2^{\otimes (n-o(n))}$ become possible with asymptotically sub-linear ($o(n)$) waste of entanglement (cf.~\cite{Kumagai2013}). Importantly, a reversible set $\OO$ would need to be larger than LOCC~\cite{Vidal-irreversibility} and even larger than the set of all non-entangling, or separability-preserving, operations~\cite{irreversibility}.

Looking beyond LOCC has the added advantage that, following the paradigm of quantum resource theories~\cite{RT-review}, the results derived for general classes of free operations $\OO$ can be more easily extended to other quantum resources, where the definition of LOCC, which is specific to entanglement theory, would not carry over. Among these, a prominent role is played by the resource theory of non-stabiliser quantum computation, informally known as `quantum magic', that is key to universal quantum computation~\cite{Bravyi-Kitaev, Veitch2014}.

In their 2008 ground-breaking paper~\cite{BrandaoPlenio1}, Brand\~{a}o and Plenio have addressed the above problem, proposing a theory of entanglement manipulation that would achieve reversibility by employing so-called `asymptotically non-entangling' (ANE) operations. These are defined axiomatically as those that inject at most a `vanishingly small amount' of entanglement in the system, as measured by a prescribed entanglement measure. Brand\~{a}o and Plenio computed the corresponding distillable entanglement and entanglement cost~\cite{BrandaoPlenio2}:
\bb
E_{d,\,\ane} (\rho_{AB}) &= \stein(\rho_{AB} \,\|\, \SEP)\, , \\
E_{c,\,\ane} (\rho_{AB}) &= E_R^\infty(\rho_{AB})\, .
\label{connection_entanglement_testing_distillation}
\ee
What is remarkable here is that the ANE distillable entanglement is given precisely by the Stein exponent associated with entanglement testing. This provides a \emph{surprising connection between entanglement manipulation and entanglement testing}, that has recently been shown to be even more far-reaching than previously envisaged~\cite{DNE-distillable}.

\subsection{Generalised quantum Stein's lemma}

The strength of the Brand\~{a}o--Plenio connection~\eqref{connection_entanglement_testing_distillation} is that it allows us to rephrase an entanglement distillation problem as a hypothesis testing problem. To put it to good use, we need to solve the latter. In a pioneering work~\cite{Brandao2010}, Brand\~{a}o and Plenio themselves proposed a solution to this question, in the form of a beautiful equality known as the \deff{generalised quantum Stein's lemma}:

\begin{center}
\begin{minipage}{.75\textwidth}
\emph{The Stein exponent for entanglement testing equals the regularised relative entropy of entanglement}:
\end{minipage}
\end{center}
\vspace{-4.5ex}
\bb
\stein(\rho\,\|\,\sep) = E_R^\infty(\rho_{AB})\, .
\label{GQSL}
\ee

However, recently a serious gap in the original proof by Brand\~ao and Plenio was found~\cite{gap, gap-comment}, casting doubts on the validity of the above identity. In this work, we provide an alternative proof of the generalised quantum Stein's lemma~\eqref{GQSL}, shedding light on the asymptotic theory of entanglement manipulation and on the framework of quantum resource theories as a whole. In fact, our approach extends to general quantum resources beyond entanglement, provided that they obey the five axioms proposed by Brand\~{a}o and Plenio~\cite{Brandao2010}.

The generalised quantum Stein's lemma has been recently proved also by Hayashi and Yamasaki~\cite{Hayashi-Stein}. The independent proof presented here relies, however, on very different techniques that we believe to be of independent interest in (quantum) information theory, mainly the \emph{blurring lemma} and the \emph{bosonic lifting}. In its basic classical version, the blurring lemma outlines a procedure to make an unknown probability distribution more `regular', by adding some noise to it in such a way as to `smear it' over nearby type classes. Proving the quantum blurring lemma necessitates our bosonic lifting technique, which translates a sequence of problems over many copies of a finite-dimensional quantum system into a \emph{single} problem in an infinite-dimensional Hilbert space. The key advantage of this translation, which effectively amounts to a second quantisation, is that it enables the use of techniques from continuous-variable quantum information. We believe that this technique may find applicability elsewhere in information theory, where the very definition of rate requires the study of asymptotic sequences of problems.

We describe these techniques below in greater detail; here we remark that they present some notable advantages over the approach by Hayashi and Yamasaki~\cite{Hayashi-Stein}. First, they unify the calculations of the generalised Stein exponent~\eqref{Stein_informal} (i.e.\ the proof of~\eqref{GQSL}) and of the generalised Sanov exponent~\eqref{Sanov_informal}. 
The latter is derived in~\cite{generalised-Sanov}, where a complementary `generalised Sanov theorem' is established. Notably, the argument presented there relies crucially on our (classical) blurring lemma. 
Secondly, we obtain almost for free an extension of optimal resource testing to `almost power states'~\cite[p.~803]{Brandao2010}, which allows us to explore quantum hypothesis testing beyond the i.i.d.\ regime in \emph{both} arguments, i.e.\ not only on the alternative hypothesis but also on the null hypothesis. (And it brings us closer to restoring the original Brand\~{a}o--Plenio argument, albeit with a mathematically different approach.) 

On the one hand, for general quantum resources, our approach uses of all the original five Brand\~{a}o--Plenio axioms, while the one by Hayashi and Yamasaki~\cite{Hayashi-Stein} removes two of them. We did not attempt to extend our result in this direction because said axioms are anyway satisfied for most resources theories of practical interest, including entanglement and magic as well as many others~\cite{RT-review}. 

The rest of the paper is organised as follows. In Section~\ref{sec_main_result} we present our main result in full generality, discussing some of its manifold implications and giving a more detailed overview of the techniques used to prove it. In Section~\ref{sec_notation} we introduce the notation. In Section~\ref{sec_classical_solution} we present our solution to the \emph{classical} generalised Stein's lemma, which is preparatory to the quantum solution and gives a slightly stronger result. Section~\ref{subsec_intuitive_description_classical} features an informal description of \emph{blurring}, one of our two key techniques, at the classical level. The classical blurring lemma can be found in Section~\ref{subsec_classical_blurring_lemma}; the proof of the classical generalised Stein's lemma follows immediately (Section~\ref{subsec_classical_proof}). 
The quantum solution is presented in Section~\ref{sec_quantum_solution}. First, in Section~\ref{subsec_informal_quantum} we provide an informal description of it; the formal statement of the quantum blurring lemma is in Section~\ref{subsec_quantum_blurring_lemma}; we postpone its proof and present first the solution of the generalised quantum Stein's lemma \emph{assuming} the quantum blurring lemma (Section~\ref{subsec_proof_GQSL}). The technical core of the paper is the proof of the quantum blurring lemma, which occupies Sections~\ref{subsec_technical_preliminaries}--\ref{subsec_proof_quantum_blurring_lemma}. In Section~\ref{sec_extension_GQSL_almost_iid} we show how to further extend our main result to almost power states.

\section{Main result} \label{sec_main_result}

Our main result holds not only for entanglement theory, but for the broad class of quantum resource theories satisfying the list of assumptions proposed by Brand\~{a}o and Plenio~\cite{Brandao2010}. Let $\HH$ be a Hilbert space of finite dimension $d$. We consider a family $(\FF_n)_n$ of sets of `free states'. In practice, each $\FF_n$ is a subset of the set of density operators on $n$ copies of the system, in formula $\FF_n \subseteq \D\big(\HH^{\otimes n}\big)$. The sets of free states should satisfy the following \emph{Brand\~{a}o--Plenio axioms}~\cite[p.~795]{Brandao2010}:
\begin{enumerate}
    \item Each $\FF_n$ is a convex and closed subset of $\D(\HH^{\otimes n})$.
    \item $\FF_1$ contains a full-rank state $\sigma_0$, i.e.\ $\FF_1 \ni \sigma_0 \geq c\id > 0$.
    \item The family $(\FF_{n})_n$ is closed under partial traces, i.e.\ if $\sigma\in \FF_{n+1}$ then $\Tr_{n+1}\sigma \in \FF_n$, where $\Tr_{n+1}$ denotes the partial trace over the last subsystem.
    \item The family $(\FF_{n})_n$ is closed under tensor products, i.e.\ if $\sigma\in \FF_n$ and $\sigma'\in \FF_m$ then $\sigma\otimes \sigma'\in \FF_{n+m}$.
    \item Each $\FF_n$ is closed under permutations, i.e.\ if $\sigma\in \FF_n$ and $\pi\in S_n$ denotes an arbitrary permutation of a set of $n$ elements, then also $U_\pi^{\vphantom{\dag}} \sigma U_\pi^\dag\in \FF_n$, where $U_\pi$ is the unitary implementing $\pi$ over $\HH^{\otimes n}$. 
\end{enumerate}

These axioms are conceived as an abstraction of the properties of separable states~\cite{Werner}, and indeed the family of sets of separable states $\FF_n = \SEP_{A^n:B^n}$, where
\bb
\SEP_{A:B} \coloneqq \co \left\{ \alpha_A \!\otimes\! \beta_B\!:\, \alpha_A\!\in\! \D(\HH_{A}),\, \beta_B\!\in\! \D(\HH_{B}) \right\}\! ,
\label{separable_states}
\ee
where $\alpha_A$ and $\beta_B$ are arbitrary density operators on $\HH_A$ and $\HH_B$, respectively, satisfies Axioms~1--5. The task of entanglement testing described in Section~\ref{sec_intro} can be rather easily generalised to that of general \deff{quantum resource testing}, where the alternative hypothesis of the $n$-copy problem includes all free states $\FF_n$ (cf.\ Figure~\ref{entanglement_testing_fig}). Our main result can then be stated as follows.

\vspace{1ex}
\begin{boxed}{}
\begin{thm}[(Generalised quantum Stein's lemma)] \label{GQSL_thm}
Let $\HH$ be a finite-dimensional Hilbert space, and let $(\FF_n)_n$ be a sequence of sets of states $\FF_n\subseteq \D\big(\HH^{\otimes n}\big)$ that obeys the Brand\~{a}o--Plenio axioms (Axioms~1--5 in Section~\ref{sec_main_result}). Then, for all $\rho\in \D(\HH)$, the hypothesis testing relative entropy of resource testing, defined by~\eqref{hypothesis_testing_relative_entropy} with the convention~\eqref{relent_resource}, satisfies that
\bb
\lim_{n\to\infty} \frac1n\, \rel{D_H^\e}{\rho^{\otimes n}}{\FF_n} = D^\infty(\rho\|\FF) \quad\ \forall\, \e\!\in\! (0,1)\, ,
\label{GQSL_resources_D_H}
\ee
implying that
\bb
\stein(\rho\|\FF) = D^\infty(\rho\|\FF)\, .
\label{GQSL_resources}
\ee
In particular, the Stein exponent associated to entanglement testing equals the regularised relative entropy of entanglement, i.e.~\eqref{GQSL} holds for all finite-dimensional bipartite states $\rho_{AB}$.
\end{thm}
\end{boxed}

Theorem~\ref{GQSL_thm} characterises the ultimate efficiency of entanglement testing, providing us with the formula~\eqref{GQSL} to compute it. Furthermore, it also entails that all four quantities appearing in~\eqref{connection_entanglement_testing_distillation} actually coincide:
\bb
E_{d,\,\ane} (\rho) = \stein(\rho \,\|\, \SEP) = E_R^\infty(\rho) = E_{c,\,\ane} (\rho)
\ee
for all $\rho = \rho_{AB}$. This solves immediately the Plenio problem~\eqref{Plenio_problem}, establishing a thermodynamical, fully reversible theory of entanglement manipulation under ANE operations; furthermore, it endows the regularised relative entropy of entanglement with a double operational meaning in entanglement testing and entanglement manipulation.

Theorem~\ref{GQSL_thm} encompasses a broad class of quantum resource theories~\cite{RT-review}, including quantum coherence~\cite{coherence-review}, athermality~\cite{thermo-review}, and the aforementioned `quantum magic'~\cite{Bravyi-Kitaev, Veitch2014}. It thus provides a much clearer understanding not only of entanglement theory, but indeed of all quantum resource theories. 

To prove Theorem~\ref{GQSL_thm} we introduce two new techniques. The first is a physical procedure that we call `blurring', which allows us to deal with the composite nature of the alternative hypothesis --- i.e.\ the fact that it might involve non-i.i.d.\ states. When applied to a permutationally invariant classical probability distribution over an $n$-copy alphabet, blurring makes it more evenly spread across nearby type classes (see Figure~\ref{blurring_fig} below for a pictorial representation). As it turns out, this technique alone suffices to prove a classical version of the generalised Stein's lemma, but not the fully quantum one. We then introduce our second technical innovation, which is to lift the problem to an infinite-dimensional bosonic quantum system by means of a second quantisation step; the problem will then be solved there with techniques from continuous-variable quantum information. We found rather remarkable, albeit not entirely unexpected, that the second-quantised action of the blurring map is represented by a bosonic quantum channel that is closely related to the pure loss channel, which has been widely studied in continuous-variable quantum information~\cite{holwer, Caruso2006, Wolf2006, Wolf2007}. This second quantisation step is the most technically delicate part of our proof. 

In the separate paper~\cite{generalised-Sanov} we use the same technique of blurring, and in particular the classical blurring lemma below (Lemma~\ref{blurring_lemma}), to solve also the complementary problem of calculating the Sanov exponent~\eqref{Sanov_informal} for entanglement testing and general resource testing. Interestingly, while the Stein exponent is given by a regularised expression (the regularised relative entropy of entanglement), the Sanov exponent can be expressed as a \emph{single-letter} formula, which makes it more computationally tractable.

In Section~\ref{sec_extension_GQSL_almost_iid}, we show that Theorem~\ref{GQSL_thm} can be extended rather easily to encompass also a modified version of entanglement testing in which one is no longer promised that the null hypothesis is perfectly i.i.d., but merely \emph{approximately} so, in the sense that a constant number of output copies can end up corrupted, i.e.\ not equal to $\rho$ and even possibly correlated. The intuition is that since the number of corrupted sites is anyway asymptotically constant, it should not affect an extensive quantity such as the Stein exponent. When one accounts for quantum correlations, this notion of approximate i.i.d.-ness can be formalised by means of \emph{almost power states}, a class of states introduced by Brand\~{a}o and Plenio themselves in the course of their proof attempt~\cite[p.~803]{Brandao2010}. Interestingly, with our techniques we are able to recover a weaker case of their Lemma~III.7, whose failure was precisely the breaking point of their argument. We fall short of restoring their proof completely, but we find this a strong indication that their overall proof strategy may be salvageable.

\section{Notation} \label{sec_notation}

\subsection{Classical notation}

In what follows, $\XX$ will denote a finite classical alphabet. The set of probability distributions over $\XX$, i.e.\ the set of functions $p:\XX\to [0,1]$ such that $\sum_x p(x) =1$, will be indicated as $\PP(\XX)$. To unify the notation to the quantum one, we will denote the \deff{total variation distance} between two probability distributions $p,q\in \PP(\XX)$ as
\bb
\frac12 \|p-q\|_1 \coloneqq \frac12 \sum_x |p(x) - q(x)|\, .
\label{TV_distance}
\ee
We will also sometimes use the distance
\bb
\|p-q\|_\infty \coloneqq \max_x |p(x) - q(x)|\, .
\label{infinity_distance_classical}
\ee
Note that $\frac12 \|p-q\|_1 \geq \|p-q\|_\infty$ for all pairs of normalised probability distributions $p,q\in \PP(\XX)$. For some positive integer $n\in \N^+$, $x^n$ will represent a generic sequence of length $n$ composed of symbols belonging to $\XX$. The set of all such sequences will be denoted as $\XX^n$.

Let $\XX$ be a finite alphabet, and let $n\in \N^+$ be a positive integer. An \deff{$\boldsymbol{n}$-type} on $\XX$ (or, simply, a type) is a probability distribution $t_n$ on $\XX$ such that $nt_n(x)\in \N$ is an integer for all $n\in \N$~\cite{CSISZAR-KOERNER}. Therefore, the set of all $n$-types on $\XX$ can be defined as
\bb
\mathcal{T}_n \coloneqq \left\{ \left(\frac{k_x}{n}\right)_{\!\!x\in \XX}\!\!:\ k_x\!\in\! \N\ \ \forall x\!\in\! \XX,\ \sum_{x\in\XX} k_x = n\right\} .
\ee
If $n$ is fixed or clear from the context, or simply for the sake of compactness, we will occasionally denote an $n$-type with $t$ instead of $t_n$. A well-known counting argument shows that
\bb
\left| \mathcal{T}_n \right| = \binom{n+|\XX|-1}{|\XX|-1}\leq (n+1)^{|\XX|-1}\, ,
\label{counting_types}
\ee
where $|\XX|$ is the cardinality of $\XX$. For a given $t_n \in \mathcal{T}_n$, the associated \deff{type class} $T_{n,t_n}$ is the set of sequences of length $n$ made from elements in $\XX$ that have type $t_n$. In formula,
\bb
T_{n,t_n} \coloneqq \left\{ x^n\in \XX^n\!:\ N(x|x^n) = n t_n(x)\ \ \, \forall\ x\!\in\! \XX \right\} ,
\label{type_class}
\ee
where $N(x|x^n)$ denotes the number of times the symbol $x$ appears in the sequence $x^n$. Clearly, any sequence in $T_{n,t_n}$ can be obtained from any other such sequence by applying a suitable permutation. The cardinality of $T_{n,t_n}$ can be calculated as
\bb
\big| T_{n,t_n} \big| = \frac{n!}{\prod_{x\in \XX} \big(nt_n(x)\big)!} \eqqcolon \binom{n}{nt_n}\, .
\label{size_type_class}
\ee

\subsection{Quantum notation}

\subsubsection{States and channels}

The set of states on a quantum system represented by a separable Hilbert space $\HH$ is identified with the set of \deff{density operators} on $\HH$. A trace class operator $\rho$ on $\HH$ is a density operator if it is positive semi-definite, denoted $\rho\geq 0$, i.e.\ such that $\braket{\psi|\rho|\psi} \geq 0$ for all $\ket{\psi}\in \C^d$, and of trace one, i.e.\ such that $\Tr \rho = \sum_x \braket{x|\rho|x} = 1$. Here, $\{\ket{x}\}_x$ represents an orthonormal basis of $\HH$, either finite or numerable (because $\HH$ is separable), and the sum is well defined because each term is non-negative. 

In what follows, we will denote with $\T(\HH)$ and $\D(\HH)$ the sets of trace class operators and of density operators on $\HH$, respectively. Note that $\T(\HH)$ can be thought of as a Banach space once it is equipped with the \deff{trace norm} $\|\cdot\|_1$ defined by $\|X\|_1\coloneqq \Tr |X|$, where $|X|\coloneqq \sqrt{X^\dag X}$ denotes the operator absolute value. If $X = X^\dag$ is self-adjoint with spectrum $(x_i)_i$, then $\|X\|_1 = \sum_i |x_i|$. In most of the paper we will be concerned only with finite-dimensional Hilbert spaces; however, right at the end of the proof of the generalised quantum Stein's lemma we will need to lift the problem to an infinite-dimensional space in order to solve it. When $\HH$ is finite dimensional, we will also denote with $\LL(\HH)$ the space of all linear operators on $\HH$.

Given two quantum systems $A$ and $B$ with Hilbert spaces $\HH_A$ and $\HH_B$, the composite quantum system $AB$ is represented by the tensor product Hilbert space $\HH_{AB} = \HH_A \otimes \HH_B$. For some positive integer $n\in \N^+$, we will also denote with $A^n$ the system with Hilbert space $\HH_{A^n} = \HH_A^{\otimes n}$ obtained by joining $n$ copies of $A$.

The distance between density operators is most commonly measured with the \deff{trace distance}, given by $\frac12 \|\rho - \sigma\|_1$ for $\rho,\sigma\in \D(\HH)$. Note that $\frac12 \|\rho - \sigma\|_1\in [0,1]$, and that the trace distance between two (normalised) pure states is given by
\bb
\frac12 \left\|\ketbra{\psi} - \ketbra{\phi} \right\|_1 = \sqrt{1 - |\!\braket{\psi|\phi}\!|^2}\, .
\label{trace_distance_pure_states}
\ee
An alternative measure of distance is the \deff{fidelity}, given by~\cite{Uhlmann-fidelity}
\bb
F(\rho,\sigma) \coloneqq \big\|\sqrt{\rho}\sqrt{\sigma}\big\|_1\, .
\label{fidelity}
\ee
Trace distance and fidelity are related by inequalities discovered by Holevo~\cite{Holevo1972} and Fuchs--van de Graaf~\cite{Fuchs1999}:
\bb
1 - F(\rho,\sigma) \leq 1 - \Tr \big[\sqrt\rho \sqrt\sigma \big] \leq \frac12 \left\|\rho-\sigma\right\|_1 \leq \sqrt{1 - F(\rho,\sigma)^2}\, .
\label{Holevo_Fuchs_van_de_Graaf}
\ee
Since these inequalities are independent of the dimension, we can think of the trace distance and the fidelity as essentially equivalent for our purposes.

A \deff{quantum channel}, or, simply, a channel, is a bounded linear map $\Lambda:\T(\HH_1) \to \T(\HH_2)$ that is completely positive, meaning that $\Lambda\otimes I_{\HH'}$ maps positive semi-definite operators on $\HH_1\otimes \HH'$ to positive semi-definite operators on $\HH_2\otimes \HH'$, for all auxiliary Hilbert spaces $\HH'$, and trace preserving, meaning that $\Tr \Lambda(X) = \Tr X$ for all input operators $X$~\cite{Stinespring, Jamiolkowski72, Choi, KRAUS}. In this paper we will also encounter \deff{sub-channels}, which are complete positive as well but only trace non-increasing, i.e.\ such that $\Tr \Lambda(X) \leq \Tr X$ for all $X$. The trace norm is contractive with respect to all (sub-)channels, meaning that~\cite{Ruskai1994}
\bb
\big\| \Lambda(X) \big\|_1 \leq \|X\|_1\qquad \forall\ X\in \T(\HH_1)
\label{contractivity_trace_norm}
\ee
whenever $\Lambda$ is a (sub-)channel.

In what follows, we will often need to study an operator function known as the \deff{positive part}. This is defined as follows: for any given self-adjoint trace class operator $X = X^\dag\in \TT(\HH)$ with spectral decomposition $X=\sum_i x_i P_i$, and all $x_i$'s distinct, its positive part is defined as
\bb
X_+ \coloneqq \sum_i \max\{x_i,0\}\, P_i\, .
\label{positive_part}
\ee
Note that the trace of the positive part can be written as
\bb
\Tr X_+ = \frac{\Tr X + \|X\|_1}{2}\, .
\label{relation_trace_positive_part_trace_norm}
\ee
The above function enjoys a wealth of well-known properties, collected in the lemma below. We provide also a brief proof for completeness.

\begin{lemma} \label{variational_program_trace_X_plus_lemma}
For all self-adjoint trace class operators $X=X^\dag \in \TT(\HH)$, it holds that
\begin{align}
\Tr X_+ &= \min\left\{ \Tr Y:\, Y\geq 0,\ Y\geq X\right\} \label{first_variational_program_trace_X_plus} \\
&= \max \left\{ \Tr QX:\, 0\leq Q\leq \id \right\} . \label{second_variational_program_trace_X_plus}
\end{align}
In particular, the trace of the positive part:
\begin{enumerate}[(a)]
\item is monotonically non-decreasing with respect to the positive semi-definite order, i.e.\ $X \leq X'$ implies $\Tr X_+ \leq \Tr X'_+$;
\item is sub-additive, meaning that $\Tr (X+Y)_+ \leq \Tr X_+ + \Tr Y_+$; and
\item obeys the data processing inequality, i.e.\ it is non-increasing under quantum channels.
\end{enumerate}
\end{lemma}

\begin{proof}
The inequality $\geq$ in~\eqref{first_variational_program_trace_X_plus} follows by setting $Y = X_+$, and the converse relation $\leq$ is deduced by taking matrix elements in the eigenbasis $\{\ket{\psi_i}\}_i$ of $X = \sum_i x_i \ketbra{\psi_i}$, which yields $\Tr Y = \sum_i \braket{\psi_i|Y|\psi_i} \geq \sum_i \max\{0, x_i\} = \Tr X_+$ for any $Y$ that satisfies both $Y\geq 0$ and $Y\geq X$. The proof of~\eqref{second_variational_program_trace_X_plus} is even simpler: from $X\leq X_+$ we get $\Tr QX \leq \Tr Q X_+ \leq \Tr X_+$ for all $Q$ such that $0\leq Q\leq \id$; setting $Q$ equal to the projector onto the span of the eigenvectors of $X$ with positive eigenvalues achieves this bound.

Claims~(a) and~(b) follow immediately from~\eqref{second_variational_program_trace_X_plus}. As for~(c), taking an arbitrary quantum channel $\Lambda:\TT(\HH) \to \TT(\HH')$, we have $\Tr \Lambda(X)_+ \leq \Tr \Lambda(X_+) = \Tr X_+$, where the inequality follows from~\eqref{first_variational_program_trace_X_plus} once one observes that $X_+ \geq 0$ and $X_+ \geq X$, and hence $\Lambda(X_+)\geq 0$ and $\Lambda(X_+) \geq \Lambda(X)$ because $\Lambda$ is (completely) positive. Note that, due to~\eqref{relation_trace_positive_part_trace_norm}, (b)~also follows from the triangle inequality for the trace norm, while (c)~is equivalent to the contractivity of the trace norm under (sub-)channels, expressed by~\eqref{contractivity_trace_norm}.
\end{proof}

\subsubsection{Relative entropies}

\begin{note}
Unless otherwise specified, all logarithms are to be taken to base $2$. Occasionally, we will also use the natural logarithm, denoted by $\ln$.
\end{note}

The quantum (or Umegaki) relative entropy between two finite-dimensional states $\rho,\sigma\in \D(\C^d)$ is given by the expression $D(\rho\|\sigma)  = \Tr\left[ \rho \left(\log \rho - \log \sigma\right)\right]$ (see~\eqref{Umegaki}). In their pioneering work~\cite{Hiai1991}, Hiai and Petz have shown that this quantity has an operational meaning as the asymptotic Stein exponent associated with quantum hypothesis testing of i.i.d.\ states.

However, while $D(\rho\|\sigma)$ is the operationally relevant quantity asymptotically, at the one-shot level other types of quantum relative entropies are more relevant. The first one is the \deff{hypothesis testing relative entropy}, given for two arbitrary states $\rho,\sigma\in \D(\C^d)$ and some $\e\in (0,1)$ by~\cite{Buscemi2010, Wang-Renner}
\bb
D_H^\e(\rho\|\sigma) \coloneqq - \log \min\left\{ \Tr Q\sigma:\ 0\leq Q\leq \id,\ \Tr Q\rho \geq 1-\e\right\} .
\label{hypothesis_testing_relative_entropy}
\ee
It is not difficult to verify that $2^{-D_H^\e(\rho\|\sigma)}$ coincides with the minimal type-II error probability when distinguishing $\rho$ from $\sigma$, with the type-I error probability bounded from above by $\e$.

Another type of one-shot quantum relative entropy is the \deff{max-relative entropy}, given by~\cite{Datta08}
\bb
D_{\max}(\rho\|\sigma) \coloneqq \min\left\{\lambda:\ \rho \leq 2^\lambda \sigma \right\} ,
\label{D_max}
\ee
where, as usual, $A\leq B$ means that $B-A$ is positive semi-definite. The operator monotonicity of the logarithm implies immediately that
\bb
D(\rho\|\sigma) \leq D_{\max}(\rho\|\sigma)\, .
\label{relation_D_D_max}
\ee
To obtain a truly one-shot quantity, however, we need to perform a suitable `smoothing' on $D_{\max}$, including an optimisation over nearby states. The \deff{smoothed max-relative entropy} is given by
\bb
D_{\max}^\e(\rho\|\sigma) \coloneqq \min_{\rho':\ \frac12 \left\|\rho - \rho'\right\|_1 \leq \e} D_{\max}(\rho' \| \sigma)\, ,
\label{smoothed_D_max}
\ee
where $\rho'$ is assumed to be a density operator on the same system as $\rho$ and $\sigma$, and the minimum exists because $D_{\max}(\cdot\|\cdot)$ is jointly lower semi-continuous. Another option for carrying out a similar smoothing procedure is to proceed directly at the operator level: the \deff{Datta--Leditzky smoothed max-relative entropy} is given by~\cite[Definition~4.1]{Datta-Leditzky}
\bb
\widetilde{D}_{\max}^\e(\rho\|\sigma) \coloneqq&\ \min\left\{ \lambda:\ \rho\leq 2^\lambda \sigma + \Delta ,\ \Delta\geq 0,\ \Tr \Delta \leq \e\right\} \\
=&\ \min\left\{ \lambda:\ \Tr(\rho - 2^\lambda \sigma)_+ \leq \e \right\} ,
\label{Datta_Leditzky}
\ee
where the equality between the two expressions follows from~\eqref{first_variational_program_trace_X_plus}.

The three quantities~\eqref{D_max}--\eqref{Datta_Leditzky} turn out to be closely related.  Informally,
\bb
D_{\max}^\e(\rho\|\sigma) \sim \widetilde{D}_{\max}^{\e'}(\rho\|\sigma) \sim D_H^{1 - \e''}(\rho\|\sigma)\, ,
\label{informal_wsc_duality}
\ee
where $\e' = \e'(\e)$ and $\e'' = \e''(\e)$ denote unspecified universal (continuous) functions of $\e$ with $0$ and $1$ among the fixed points. Which functions appear precisely will depend on what inequality one wishes to obtain. We will not review the specific form of these inequalities here; we refer the interested reader to the specialised literature~\cite{Tomamichel2013, Datta-Leditzky, Anshu2019,tight-relations}. We make an exception for the relation between~\eqref{smoothed_D_max} and~\eqref{Datta_Leditzky}, which takes the form~\cite[Corollary~8]{tight-relations}
\bb
\widetilde{D}_{\max}^\e(\rho\|\sigma) \leq D_{\max}^\e(\rho\|\sigma) \leq \widetilde{D}_{\max}^{\,\e^2}(\rho\|\sigma) + \log\frac{1}{1-\e^2}\, .
\label{Datta_Renner_lemma}
\ee
The second inequality, which is the non-trivial one, is equivalent to a beautiful lemma first found by Datta and Renner~\cite{Datta2009} and then adapted to the normalised case by Brand\~{a}o and Plenio~\cite[Lemma~C.5]{Brandao2010}. See also~\cite[Theorem~5]{tight-relations} for a refined statement and a different proof technique.

All of the above relative entropies can be calculated for classical probability distributions instead of quantum states. In the classical case, the tightest relation between the hypothesis testing relative entropy and the smoothed max-relative entropy is~\cite[Corollary~9, Eq.~(60)]{tight-relations}
\bb
D_H^{1-\e-\eta}(p \| q) + \log \eta \leq D_{\max}^{\e} (p\|q) \leq D_H^{1-\e}(p \| q) + \log \e\, .
\label{classical_wsc_duality}
\ee
Quantumly, the tightest possible relation is entirely analogous to~\eqref{classical_wsc_duality}, but it features an $\e^2$ instead of $\e$ on the rightmost side~\cite[Corollary~9, Eq.~(58)]{tight-relations}.

In what follows, we will often need to compute a relative entropy between a state and a \emph{family} of states. Given $\rho\in \D(\C^d)$ and some set of states $\FF \subseteq \D(\C^d)$, denoting with $\mathds{D}$ one of the above relative entropies, we will adopt the convention that
\bb
\mathds{D}(\rho\|\FF) \coloneqq \inf_{\sigma\in \FF} \mathds{D}(\rho\|\sigma)\, .
\label{relent_resource}
\ee
If $\FF$ represents instead a sequence $(\FF_n)_n$ of families of states $\FF_n\subseteq \D\big((\C^d)^{\otimes n}\big)$, with a slight abuse of notation and for the specific case of the relative entropy, we will also set
\bb
D^\infty(\rho\|\FF) \coloneqq \lim_{n\to\infty} \frac1n\, \rel{D}{\rho^{\otimes n}}{\FF_n}\, ;
\label{regularised_relent_resource}
\ee
note that the limit exists and can be replaced with an infimum over $n\in \N^+$, due to Fekete's lemma~\cite{Fekete1923}, provided that $\FF$ obeys Axiom~4 in Section~\ref{sec_main_result}. Axiom~2 guarantees instead that the limit is finite, because $D(\rho\|\FF_1) \leq D_{\max}(\rho\|\FF_1) \leq D_{\max}(\rho\|\sigma_0) \leq \log(1/c)$. When $\mathds{D} = D$ is the relative entropy and $\FF$ is taken to be the set of separable states, Eq.~\eqref{relent_resource} yields the relative entropy of entanglement~\eqref{relative_entropy_entanglement}, while~\eqref{regularised_relent_resource} yields its regularised version~\eqref{regularised_relative_entropy_entanglement}.

With the above notation, the Stein exponent for resource testing, which we constructed informally in~\eqref{Stein_informal}, can be defined more rigorously as
\bb
\stein(\rho \| \FF) \coloneqq \lim_{\e\to 0^+} \liminf_{n\to\infty} \frac1n\, \rel{D_H^\e}{\rho^{\otimes n}}{\FF_n}\, .
\label{Stein_rigorous}
\ee

\begin{rem} \label{reformulations_rem}
Since it is already known from the work of Brand\~{a}o and Plenio~\cite{Brandao2010} that the limit $\e\to 1^-$ of the left-hand side equals $D^\infty(\rho\|\FF)$, Eq.~\eqref{GQSL_resources} and Eq.~\eqref{GQSL_resources_D_H} are fully equivalent reformulations of the generalised Stein's lemma. Other equivalent formulations can be obtained by exploiting the relations between the various one-shot relative entropies, which are all equivalent asymptotically:
\bb
&\lim_{\e\to 0^+} \liminf_{n\to\infty} \frac1n\, \rel{D_H^\e}{\rho^{\otimes n}}{\FF_n} \\
&\qquad = \lim_{\e\to 1^-} \liminf_{n\to\infty} \frac1n\, \rel{D_{\max}^\e}{\rho^{\otimes n}}{\FF_n} \\
&\qquad = \lim_{\e\to 1^-} \liminf_{n\to\infty} \frac1n\, \rel{\widetilde{D}_{\max}^\e}{\rho^{\otimes n}}{\FF_n}\, .
\label{reformulations}
\ee
Here, the first equality follows from~\cite[Corollary~S4]{DNE-distillable}, while the second follows from~\eqref{Datta_Renner_lemma}. Proving that either of the above quantities equals $D^\infty(\rho\|\FF)$ would prove Theorem~\ref{GQSL_thm}. In fact, since Brand\~{a}o and Plenio~\cite{Brandao2010} also showed that $\lim_{\e\to 0^+} \liminf_{n\to\infty} \frac1n\, \rel{D_{\max}^\e}{\rho^{\otimes n}}{\FF_n} = D^\infty(\rho\|\FF)$, and the smoothed max-relative entropy is a decreasing function of $\e\in (0,1)$, it suffices to prove that
\bb
\liminf_{n\to\infty} \frac1n\, \rel{D_{\max}^\e}{\rho^{\otimes n}}{\FF_n} \geqt{?} D^\infty(\rho\|\FF)\, .
\label{D_max_reformulation}
\ee
This will be our line of attack.
\end{rem}

\section{The classical solution} \label{sec_classical_solution}

This section is devoted to the presentation of the classical solution. Although the asymptotic statement of the generalised classical Stein's lemma (Theorem~\ref{classical_GSL_thm} below) is subsumed by the corresponding quantum statement (Theorem~\ref{GQSL_thm}), it is very much worthwhile to start with the classical solution anyway. First, because it will teach us how to attack the quantum problem, providing us with a better intuition of why our proof strategy should work. And second, because the classical statement is fully one-shot and tells us something more concerning the behaviour of errors in resource testing in the finite-copy regime. The main result of this section is as follows.

\begin{thm}{(Generalised classical Stein's lemma)} \label{classical_GSL_thm}
Let $\XX$ be a finite alphabet, and let $(\FF_n)_n$ be a sequence of sets of probability distributions $\FF_n \subseteq \PP(\XX^n)$ that obeys the Brand\~{a}o--Plenio axioms (Axioms~1--5 in Section~\ref{sec_main_result}). Then, for all $n\in \N^+$ and $\eta,\e\in (0,1)$ with $\eta+\e < 1$, it holds that
\bb
\rel{D_{\max}^\eta}{p^{\otimes n}}{\FF_n} &\leq \rel{D_{\max}^\e}{p^{\otimes n}}{\FF_n} + \log \frac{1}{1\!-\!\e\!-\!\eta} + 2n g\big(\big(2\delta_n \!+\! \tfrac1n\big) |\XX| \big) + \left(2n\delta_n \!+\! 1\right) |\XX| \log\frac{1}{c}\, ,
\ee
where
\bb
g(x) \coloneqq (x+1)\log(x+1) - x\log x
\label{g}
\ee
is the `bosonic entropy function',\footnote{We set $g(0) \coloneqq \lim_{x\to 0^+} g(x) = 0$.} 
$c>0$ is the constant appearing in Axiom~2 of Section~\ref{sec_main_result}, and
\bb
\delta_n \coloneqq \sqrt{\frac{(|\XX|-1) \ln(n+1) - \ln \eta}{2n}}\, .
\label{delta_n}
\ee
In particular,
\bb
\lim_{n\to\infty} \frac1n\, \rel{D_{\max}^\e}{p^{\otimes n}}{\FF_n} = D^\infty(p\|\FF) = \stein(p\|\FF)
\label{classical_GSL}
\ee
for all $\e\in (0,1)$.
\end{thm}

\begin{rem}
Due to the relation~\eqref{classical_wsc_duality} between the smoothed max-relative entropy and the hypothesis testing relative entropy, Theorem~\ref{classical_GSL_thm} implies immediately the following inequality: for all $\e,\e',\delta\in (0,1)$ with $\e+\e' > 1$ and $\delta<\e$,
\bb
\rel{D_H^{\e'}}{p^{\otimes n}}{\FF_n} &\geq \rel{D_H^{\e-\delta}}{p^{\otimes n}}{\FF_n} + \log \frac{\delta(\e\!+\!\e'\!-\!1)}{1\!-\!\e'} - 2n g\big(\big(2\delta_n \!+\! \tfrac1n\big) |\XX| \big) - \left(2n\delta_n \!+\! 1\right) |\XX| \log\frac{1}{c}\, .\vphantom{\bigg\|}
\label{classical_relation_n_copy_DHs}
\ee
The above statement is significant because it connects the $n$-copy hypothesis testing relative entropies corresponding to \emph{different} type-I error probabilities. To see why this is useful, consider the case where $\e$ is close to $1$, while $\e'$ is much smaller, say of the order of (but slightly larger than) $1-\e$. Eq.~\eqref{classical_relation_n_copy_DHs} tells us that if we can find a `rough' test that achieves a type-I error probability $\e-\delta$, which we can imagine to be close to $1$ if $\delta$ is small, then there exists a much more refined test with type-I error probability $\e' \ll \e$ and `comparable' type-II error probability, provided that $n$ is large enough. The inequality pinpoints what penalty in the type-II error we incur because of this drastic type-I error reduction: using the fact that $g(x) \sim -x \log x$ for small $x$, we see that this penalty is a sub-exponential factor of the form $2^{\kappa \sqrt{n}\, \log n}$.
\end{rem}

\begin{rem}
Another type of information that can be extracted from Theorem~\ref{classical_GSL_thm} is the relation between the hypothesis testing relative entropy and the Umegaki relative entropy: for all $0<\eta<\e'<1$ and all $n\in \N^+$, it holds that
\bb
&\rel{D_H^{\e'}}{p^{\otimes n}}{\FF_n} \\
&\quad \geq \rel{D_{\max}^{\eta}}{p^{\otimes n}}{\FF_n} + \log\frac{\e'- \eta}{1-\e'} - 2n\, g\big(\big(2\delta_n + \tfrac1n\big) |\XX| \big) - \left(2n\delta_n + 1\right) |\XX| \log\frac{1}{c} \\
&\quad \geq \rel{D}{p^{\otimes n}}{\FF_n} - n\eta \log\frac1c - g(\eta) + \log \frac{\e'- \eta}{1-\e'} - 2n\, g\big(\big(2\delta_n + \tfrac1n\big) |\XX| \big) - \left(2n\delta_n + 1\right) |\XX| \log\frac{1}{c}\, ,
\ee
where the second inequality can be deduced by lower bounding the max-relative entropy with the standard relative entropy, according to~\eqref{relation_D_D_max}, and then employing the asymptotic continuity of this latter quantity minimised over free states, proved in~\cite[Lemma~7]{tightuniform} and stated here in Lemma~\ref{AC_lemma}.
\end{rem}

\subsection{Technical preliminaries on the hypergeometric distribution}

The solution of the generalised classical Stein's lemma relies heavily on the properties of the \deff{hypergeometric distribution}. This is parametrised by three integers $N,K,n$ with $N\geq K,n$, and it has support over the integers $k\in \{0,1,\ldots,n\}$. It takes the form
\bb
H(N,K;\,n,k) \coloneqq \frac{\binom{K}{k} \binom{N-K}{n-k}}{\binom{N}{n}}\, ,
\label{hypergeometric}
\ee
and it represents the probability of drawing $k$ white marbles from an urn that contains a total of $N$ marbles, $K$ of which white, when making $n$ draws in total without replacement. The duality relations
\bb
H(N,K;\,n,k) &= H(N,n;\,K,k) \\
&= H(N,N-K;\,n,n-k) \\
&= H(N,K;\,N-n,K-k)
\label{duality_hypergeometric}
\ee
are sometimes useful.

The average of the hypergeometric distribution is naturally $K/N$, i.e.\ the frequency of white marbles in the urn. In what follows, we will make ample use of the following tail bounds, which tell us how concentrated the hypergeometric distribution is around its mean: for all $N,K,n,k \in \N$ with $N\geq K$, and all $u>0$, it holds that~\cite{Chvatal1979}
\bb
\sum_{k\in \{0,\ldots,n\},\, |\frac{k}{n} - \frac{K}{N}| \geq u} H(N,K;\,n,k) \leq 2\, e^{-2nu^2}\, .
\label{first_tail_bound_hypergeometric}
\ee
When $n \geq N/2$, as it will sometimes be the case here, one can also use the duality relation~\eqref{duality_hypergeometric} to write the tighter bound
\bb
\sum_{k\in \{0,\ldots,n\},\, |\frac{k}{n} - \frac{K}{N}| \geq u} H(N,K;\,n,k) \leq 2\, e^{- \frac{2n^2 u^2}{N-n}}\, .
\label{second_tail_bound_hypergeometric}
\ee

The hypergeometric distribution can be generalised to include the case where the urn contains marbles of more than two colours. Let $N,n\in \N^+$ be positive integers, with $N\geq n$. For a given $N$-type $s\in \mathcal{T}_N$ on a finite alphabet $\XX$, the \deff{multivariate hypergeometric distribution} is a probability distribution on $\mathcal{T}_n$ whose value on $t\in \mathcal{T}_n$ yields the probability of finding $nt(x)$ marbles of colour $x$ for all $x\in \XX$, when sampling from an urn containing $N$ marbles in total, $N s(x)$ of which of colour $x$. It is given by
\bb
H_{N,s;\, n}(t) \coloneqq \frac{\prod_x \lsmatrix \! N s(x)\! \\[.5ex] n t(x) \rsmatrix}{\binom{N}{n}} = \frac{\binom{n}{nt}\binom{N-n}{Ns-nt}}{\binom{N}{Ns}} \, ,
\label{multivariate_hypergeometric}
\ee
where we used the multinomial notation in~\eqref{size_type_class}.

\begin{note}
We will reserve the notation $H(N,K;\,n,k)$ for the bivariate case, and $H_{N,s;\, n}(t)$ for the multivariate case.  
\end{note}

Clearly, the expression on the right-hand side of~\eqref{multivariate_hypergeometric} is non-zero provided that we do not demand to draw more marbles of any given colour than what are there in the urn to start with. In formula, this is expressed by requiring that $nt(x)\leq Ns(x)$ for all $x\in \XX$. Introducing the element-wise ordering $\preceq$ among vectors on $\R^\XX$, we can state this as
\bb
nt \preceq Ns\, .
\ee
When this condition is obeyed, it is possible to construct the following elementary lower bound on $H_{N,s;\, n}(t)$, which is very loose but ultimately sufficient for our purposes.

\begin{lemma} \label{hypergeometric_inequality_lemma}
Let $N,n\in \N^+$ be positive integers, with $N\geq n$. For an $N$-type $s\in \mathcal{T}_N$ and an $n$-type $t\in \mathcal{T}_n$ on a finite alphabet $\XX$, if $nt(x)\leq Ns(x)$ for all $x\in \XX$, i.e.\ $nt\preceq Ns$, then 
\bb
H_{N,s;\, n}(t) \geq 2^{-n g\left(\frac{N}{n}-1\right)} ,
\label{hypergeometric_inequality}
\ee
where $g$ is defined by~\eqref{g}.
\end{lemma}

\begin{proof}
The inequality is very loose, and its proof very brutal. Since $Ns(x)\geq nt(x)$ for all $x\in \XX$, we have $\binom{N s(x)}{n t(x)}\geq 1$ for all $x$. Hence,
\bb
H_{N,s;\, n}(t) \geqt{(i)} \frac{1}{\binom{N}{n}} \geqt{(ii)} 2^{- N h_2(n/N)} \eqt{(iii)} 2^{-n g\left(\frac{N}{n}-1\right)} ,
\ee
where in~(i) we looked at the first expression in~\eqref{multivariate_hypergeometric}, in~(ii) we used the elementary inequality $\binom{N}{n}\leq 2^{N h_2(n/N)}$, with $h_2(x) \coloneqq - x\log x - (1-x) \log(1-x)$ being the binary entropy function, and finally in~(iii) we observed that $g(x) = (1+x) \,h_2\!\left(\frac{1}{1+x}\right)$.
\end{proof}

\subsection{An informal description of the argument} \label{subsec_intuitive_description_classical}

Fix an $\e\in (0,1)$. Due to the discussion in Remark~\ref{reformulations_rem}, all we have to do is to show that $\liminf_{n\to\infty} \frac1n\, \rel{D_{\max}^\e}{p^{\otimes n}}{\FF_n} \geq D^\infty(p\|\FF)$. Ignoring for a moment subtleties that have to do with the definition of $\liminf$, we can proceed by contradiction and assume that $\rel{D_{\max}^\e}{p^{\otimes n}}{\FF_n} \lesssim n \lambda$ for some $\lambda < D^\infty(p\|\FF)$. This means that we can find a probability distribution $p_n$ that is $\e$-close to $p^{\otimes n}$, i.e.\ $p_n \approx_\e p^{\otimes n}$, and such that $p_n \leq 2^{n\lambda} q_n$ for some $q_n\in \FF_n$. Without loss of generality, due to Axioms~1 and~5 we can assume that both $p_n$ and $q_n$ are permutationally symmetric, i.e.\ they take on the same value on all sequences of a given type. This allows us to think of these objects alternatively as probability distributions on the set of types $\mathcal{T}_n$ instead of $\XX^n$. 

By using asymptotic continuity, it is not difficult to show that when $\eta \in (0,1)$ is really small ($\eta\to 0$) we have $\rel{D_{\max}^\eta}{p^{\otimes n}}{\FF_n} \gtrsim n D^\infty(p\|\FF)$. However, the key issue here is that this fact alone does not allow us to extrapolate the same statement when the small $\eta$ (later, $\eta\to 0$) is replaced by the much larger (but fixed) $\e \in (0,1)$. We will now explain how to overcome precisely this obstacle, lower bounding directly $\rel{D_{\max}^\e}{p^{\otimes n}}{\FF_n}$ with $\rel{D_{\max}^\eta}{p^{\otimes n}}{\FF_n}$, up to asymptotically vanishing remainder terms. 

In order to establish an upper bound on $\rel{D_{\max}^\eta}{p^{\otimes n}}{\FF_n}$, which is a minimisation over smoothed probability distributions $p'_n \approx_\eta p^{\otimes n}$, we need to construct a suitable ansatz for $p'_n$. There is a natural way to do so, and it is to take as $p'_n$ the typical part of $p^{\otimes n}$. In the type space, $p'_n$ will be approximately equal to $p^{\otimes n}$ on types that are close to $p$, and vanish anywhere else. If for all $\xi>0$ we can show that
\bb
p'_n \leqt{?} 2^{n(\lambda+\xi)} q'_n\, ,\quad q'_n\in \FF_n\, ,
\label{intuitive_description_classical_eq1}
\ee
then we will obtain that
\bb
n D^\infty(p\|\FF) \lesssim \rel{D_{\max}^\eta}{p^{\otimes n}}{\FF_n} \leq n(\lambda+\xi)\, ,
\ee
which will give $D^\infty(p\|\FF)\leq \lambda$, in contradiction with the assumption that $\lambda < D^\infty(p\|\FF)$, once we divide by $n$, take the limit $n\to\infty$, \emph{and then} send $\eta,\xi \to 0$. Therefore, proving~\eqref{intuitive_description_classical_eq1} would conclude the argument.

A second observation is that since $p^{\otimes n}$ is concentrated on types $\approx p$, and $p^{\otimes n} \approx_\e p_n$, the probability distribution $p_n$ needs to have at least a total weight $\gtrsim 1-\e$ on the set of types $\approx p$. We do not know how this weight is distributed on those types, but it has to be there. We depicted this situation in Figure~\ref{blurring_fig}(a), in which some of the weight of $p_n$ is scattered on the set of types $\approx p$.

Our key idea is to construct a stochastic map $B_{n,m}$ ($m$ is a parameter whose meaning will be explained soon), called the \emph{blurring map}, that adds a bit of noise to the input. This can be done rather simply by a three-step procedure:
\begin{enumerate}[(i)]
\item append to the input sequence $m$ symbols of each species (there are $|\XX|$ species in total);
\item apply a uniformly random permutation;
\item discard $|\XX| m$ symbols, so as to go back to a sequence of length $n$.
\end{enumerate}
The above three steps describe the action of the blurring map $B_{n,m}$ on any input sequence. Applying this on $p_n$ results in a `smeared', or `blurred', version of $p_n$, denoted $\widetilde{p}_n \coloneqq B_{n,m}(p_n)$ and depicted in Figure~\ref{blurring_fig}(b). The way to think of the action of $B_{n,m}$ is that if the initial distribution has some weight on a certain type class, then blurring makes that weight `spill over' to all type classes that are close to the initial one.

In our setting, the new probability distribution $\widetilde{p}_n$ will have two key properties:
\begin{enumerate}[(a)]
\item $\widetilde{p}_n$ has approximately the same max-relative entropy of resource as $p_n$, in the sense that
\bb
\widetilde{p}_n\leq 2^{(\lambda+ \delta') n} q'_n
\label{intuitive_description_classical_eq3}
\ee
for some arbitrarily small but fixed $\delta'$ and some other $q'_n\in \FF_n$. To see this one needs to observe that from $p_n \leq 2^{n\lambda} q_n$ it follows that $\widetilde{p}_n = B_{n,m}(p_n) \leq 2^{n\lambda} B_{n,m}(q_n)$; the claim will be proved once we show that $B_{n,m}(q_n) \leq 2^{\delta' n} q'_n$ for some $q'_n\in \FF_n$, i.e.\ $D_{\max}(B_{n,m}(q_n)\|\FF_n) \leq \delta' n$. 

But this is rather easy: of the three steps (i)--(iii) that describe the action of $B_{n,m}$, only~(i) is not a free operation, according to Axioms~3 and~5; also, because of Axioms~2 and~4, step~(i) can only increase the max-relative entropy of resource by $m |\XX| \log(1/c)$, as each symbol that gets added carries with it a max-relative entropy of resource of at most $\log(1/c)$. Hence, taking $\delta' n = m |\XX| \log(1/c)$, i.e.\ $m\approx 2\delta n$ with $\delta$ a rescaled version of $\delta'$, achieves the desired result. (Here, the factor of $2$ is included for notational convenience.)

\item While $p_n$ had \emph{some} weight $\gtrsim 1-\e$ on the set of types $\approx p$, possibly unevenly distributed, $\widetilde{p}_n$ will have a possibly slighter smaller but, crucially, roughly \emph{uniformly} distributed weight on the set of types $\approx p$.

In fact, consider a type $\mathcal{T}_n \ni t \approx p$, in the sense that $\|p-t\|_\infty \leq \delta$. What is the probability of ending up in $t$ by blurring $p_n$? Since the total weight of the types around $p$ according to $p_n$ is at least $1-\e$, the total probability of ending up in $t$ is at least about $1-\e$ times the minimal transition probability $s\to t$ induced by blurring, where $s$ is an arbitrary type close to $p$, so that $\|p-s\|_\infty\leq \delta$. By looking at the action of the blurring map $B_{n,m}$, it is not difficult to realise that this transition probability is given by the multivariate hypergeometric distribution $H_{N\!,\, v_s;\, n}(t)$ (see~\eqref{multivariate_hypergeometric}), where $N\coloneqq n+|\XX|m$, and the initial type $v_s$ satisfies that
\bb
N v_s(x) &= \big(\text{$\#$ symbols $x$ in the sequence after step~(i)}\big) \\
&= ns(x) + m\, .
\ee
Note that $\|s-t\|_\infty \leq 2\delta$ by the triangle inequality, so that
\bb
N v_s(x) \geq nt(x) - 2\delta n + m \approx nt(x)\, ,
\ee
because $m \approx 2\delta n$. In other words, there are enough symbols of each species to make the transition $s\to t$ via blurring physically possible. This means that the conditions in Lemma~\ref{hypergeometric_inequality_lemma} are met, and using that result we can estimate the probability of ending up at type $t$ starting from type $s$ as $\gtrsim 2^{-n g\left(\frac{N}{n} - 1\right)} = 2^{-n g\left(2\delta |\XX|\right)}$, where $g$ is defined by~\eqref{g}. Importantly, the coefficient $g\left(2\delta |\XX|\right)$ appearing at the exponential vanishes for $\delta \to 0$. This whole argument should convince the reader that $\widetilde{p}_n$ will have weight $\gtrsim (1-\e)\, 2^{-n g\left(2\delta |\XX|\right)}$ on \emph{each and every} type close to $p$. 
\end{enumerate}

The last claim in~(b) implies that the typical part $p'_n$ of $p^{\otimes n}$ satisfies
\bb
p'_n \leq \frac{1}{1-\e}\, 2^{n g\left(2\delta |\XX|\right)}\, \widetilde{p}_n \leq \frac{1}{1-\e}\, 2^{n \left(g\left(2\delta |\XX|\right) + \lambda + \delta' \right)} q'_n\, 
\ee
where in the last step we used also~\eqref{intuitive_description_classical_eq3}. Comparing this with~\eqref{intuitive_description_classical_eq1}, we are basically done: in fact, 
\bb
\xi \coloneqq \delta' + g\left(2\delta |\XX|\right) = \mathrm{const}\cdot \delta + g\left(2\delta |\XX|\right) \to 0
\ee
as $\delta \to 0$.

\begin{rem} \label{what_BP_axioms_do_rem}
The above high-level explanation contains already enough information to appreciate why all of the Brand\~{a}o--Plenio axioms are needed in our approach. The convexity assumption in Axiom~1 is used implicitly already when invoking asymptotic continuity. Axioms~1, 3, and~5 are used to justify why steps~(ii) and~(iii) in the definition of the blurring map are free operations, while Axioms~2 and~4 are needed to establish that step~(i) adds little extra resource to the input. Entirely analogous considerations apply to the quantum solution presented in Section~\ref{sec_quantum_solution}.
\end{rem}

\begin{center}
\begin{figure}[h!t] \centering
\begin{subfigure}{\textwidth}
\includegraphics[scale=.207]{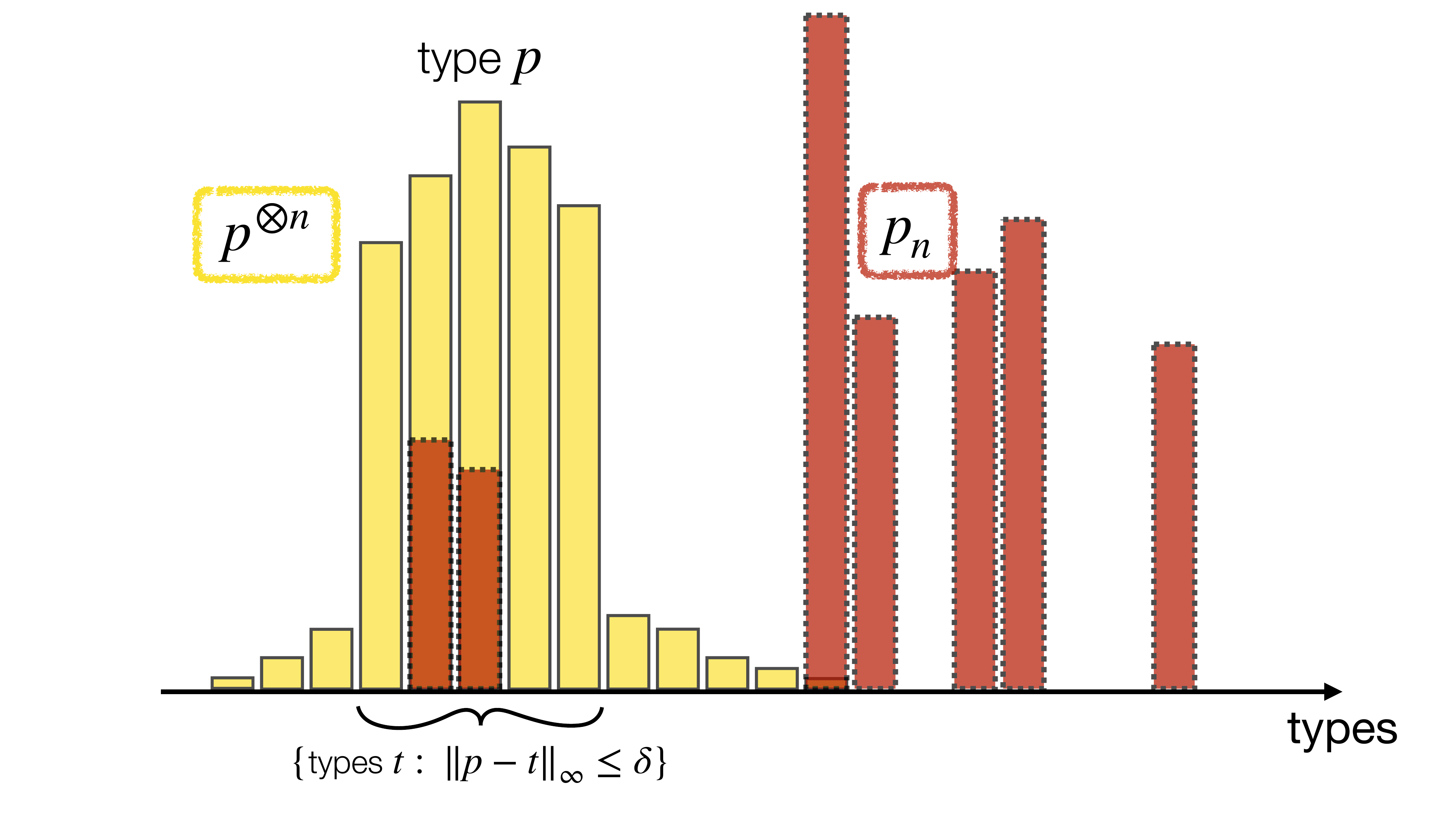}
\caption{}
\end{subfigure}
\par\bigskip
\begin{subfigure}{\textwidth}
\includegraphics[scale=.207]{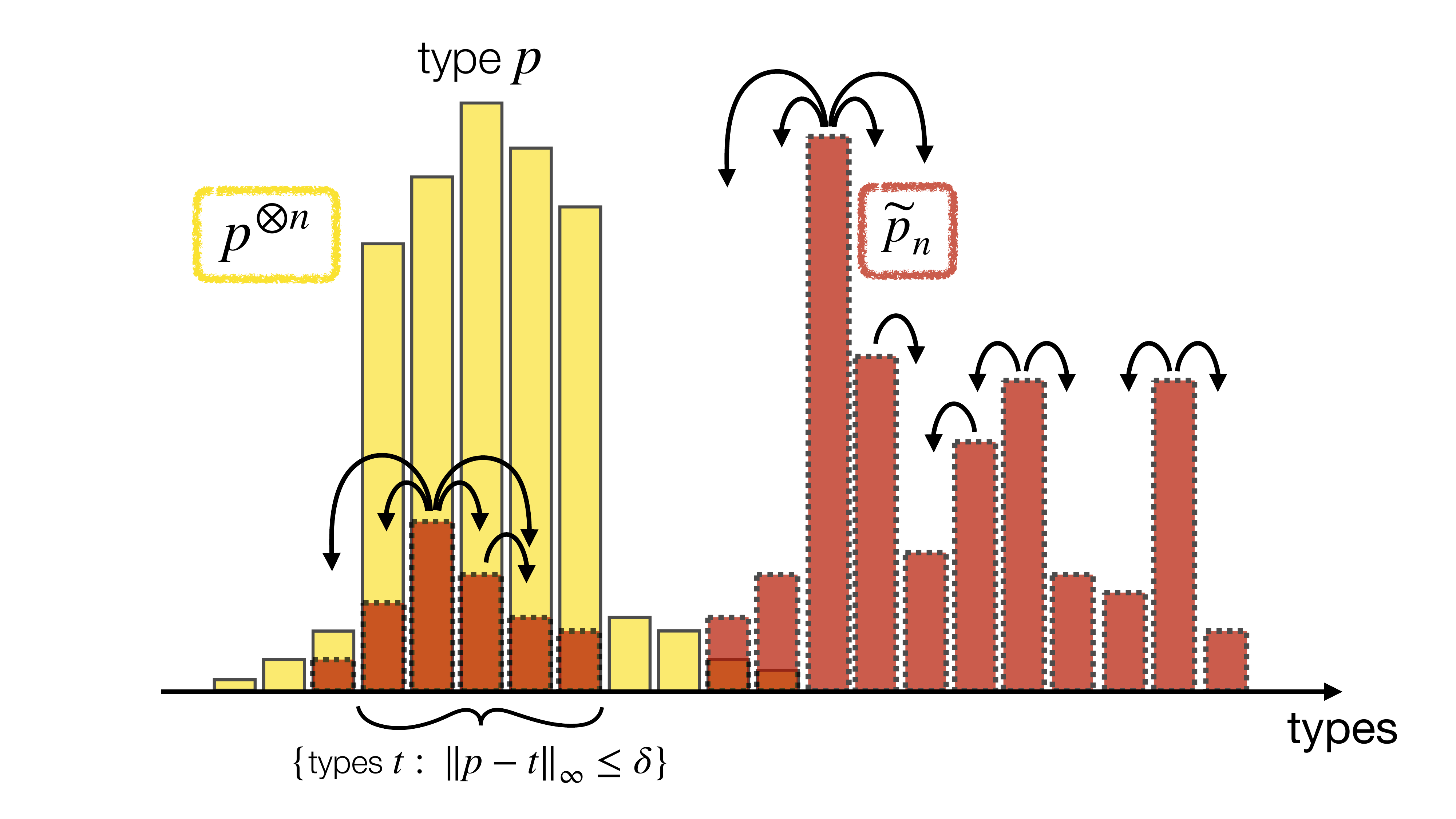}
\caption{}
\end{subfigure}
\caption{\justifying (a)~The probability distributions $p^{\otimes n}$ and $p_n$ depicted in type space. Note that $p^{\otimes n}$ is concentrated on types close to $p$. As per the discussion in Section~\ref{subsec_intuitive_description_classical}, $p^{\otimes n}$ and $p_n$ are at some distance $\e \in (0,1)$ away. In this case, $\e$ is rather close to $1$ (but strictly smaller), as most of the weight of $p^{\otimes n}$ and $p_n$ is distributed on complementary regions of the type space.
(b)~The blurring map acting on $p_n$ and yielding at the output a probability distribution $\widetilde{p}_n$. Blurring is described by steps~(i)--(iii) in Section~\ref{subsec_intuitive_description_classical}. By adding a few symbols of each type, mixing, and discarding some symbols randomly, it blurs, or smears, $p_n$, reducing the height of its peaks and making some of the associated weight spill over to nearby type classes. Importantly, as a consequence of this process $\widetilde{p}_n$ ends up having some weight on \emph{all} types close to $p$.}
\label{blurring_fig}
\end{figure}
\end{center}

\subsection{The classical blurring lemma} \label{subsec_classical_blurring_lemma}

Before we delve into the proof of Theorem~\ref{classical_GSL_thm}, it is useful to isolate and formalise the properties explained in step~(b) of the above discussion in Section~\ref{subsec_intuitive_description_classical}. To this end, here we introduce and prove the classical \emph{blurring lemma}, a fundamental statement that constitutes the linchpin of our proof of Theorem~\ref{classical_GSL_thm}. Remarkably, the exact same lemma is also used in the recent work~\cite{generalised-Sanov} to establish a complementary result in entanglement testing, the \emph{generalised quantum Sanov theorem}. Note that the classical blurring lemma will not suffice by itself to prove the generalised \emph{quantum} Stein's lemma (Theorem~\ref{GQSL_thm}), which is the main result of this paper; however, it will form the conceptual basis on which we will forge our quantum version of the blurring lemma (Lemma~\ref{quantum_blurring_lemma}).

For any pair of positive integers $n,m\in \N^+$ and, as usual, a fixed alphabet $\XX$, we define the \deff{classical blurring map} as a linear map $B_{n,m}^{\mathrm{cl}}:\R^{\XX^n} \to \R^{\XX^n}$ that transforms any input probability distribution by adding $m$ symbols of each kind $x\in \XX$, shuffling the resulting sequence, and discarding $m |\XX|$ symbols. In this way, if the input sequence is of length $n$, then the same is true of the output sequence. We can formalise the action of $B_{n,m}^{\mathrm{cl}}$ as
\bb
B_{n,m}^{\mathrm{cl}}(\cdot) \coloneqq \tr_{m} \mathcal{S}_{n+|\XX| m} \left( (\cdot) \otimes \bigotimes \nolimits_x \delta_x^{\otimes m} \right) ,
\label{classical_blurring}
\ee
where $\delta_x$ denotes the deterministic probability distribution concentrated on $x$ (i.e.\ such that $\delta_x(y) = 1$ if $y=x$, and $\delta_x(y)=0$ otherwise). We used the trace notation to prepare the ground for the extension to quantum: in this classical context, $\tr_m$ simply means `discard the last $m$ symbols'. Note that the output of the blurring map is always permutationally symmetric.

\begin{lemma}[(One-shot classical blurring lemma)] \label{blurring_lemma}
Let $p_n, q_n \in \PP\big(\XX^n\big)$ be two $n$-copy probability distributions, with $p_n$ permutationally symmetric. For some $\delta,\eta>0$, let $p_n$ be $(1-\eta)$-concentrated on the $\delta$-ball of $n$-types around a single-copy probability distribution $s\in \PP(\XX)$, in the sense that
\bb
p_n\left(\bigcup\nolimits_{t\in \mathcal{T}_n:\ \|s - t\|_\infty \leq \delta} T_{n,t} \right) \geq 1 - \eta\, ,
\ee
where $\|s - t\|_\infty = \max_{x\in \XX} \left|s(x) - t(x)\right|$. Then for $m=\ceil{2\delta n}$ it holds that
\bb
\rel{D_{\max}^{\eta}}{p_n}{B_{n,m}^{\mathrm{cl}}(q_n)} \leq &\log \frac{1}{q_n\left(\bigcup\nolimits_{t\in \mathcal{T}_n:\ \|s - t\|_\infty \leq \delta} T_{n,t} \right)} + n g\!\left( \left(2\delta + \tfrac1n\right) |\XX| \right) ,
\label{blurring_lemma_eq}
\ee
where the blurring map $B_{n,m}^{\mathrm{cl}}$ is defined by~\eqref{classical_blurring}, and the function $g$ is given by~\eqref{g}.
\end{lemma}

\begin{note}
If $q_n\big(\bigcup\nolimits_{t\in \mathcal{T}_n:\ \|s - t\|_\infty \leq \delta} T_{n,t} \big) = 0$ then~\eqref{blurring_lemma_eq} holds in a trivial way, provided that one adopts the convention that $\log 1/0 = \infty$.
\end{note}

\begin{proof}
Call $\eta' \coloneqq 1 - p_n\big(\bigcup_{t\in \mathcal{T}_n:\ \|s - t\|_\infty \leq \delta} T_{n,t} \big)$, so that $\eta'\leq \eta$. Let $p'_n$ be the probability distribution obtained from $p_n$ by cutting the `tails' that lie on types $t$ with $\|s-t\|_\infty > \delta$. In formula,
\bb
p'_n(x^n) = \left\{ \begin{array}{ll} \frac{p_n(x^n)}{1-\eta'} & \text{ if $\|t_{x^n} - s\|_\infty \leq \delta$,} \\[1.5ex] 0 & \text{ otherwise,} \end{array} \right.
\ee
where $t_{x^n}$ denotes the type of $x^n$. Note that
\bb
\frac12 \,\big\|p_n - p'_n\big\|_1 = \eta' \leq \eta\, ,
\ee
hence
\bb
\rel{D_{\max}^{\eta}}{p_n}{B_{n,m}^{\mathrm{cl}}(q_n)} \leq \rel{D_{\max}}{p'_n}{B_{n,m}^{\mathrm{cl}}(q_n)}\, .
\label{blurring_lemma_proof_first_estimate}
\ee
Our task now is to estimate the right-hand side. Since both $p'_n$ and $B_{n,m}^{\mathrm{cl}}(q_n)$ are permutationally invariant, for any sequence $x^n$ we have that $p'_n(x^n)$ and $\big(B_{n,m}^{\mathrm{cl}}(q_n)\big)(x^n)$ depend only on the type of $x^n$, denoted by $t_{x^n}$. Hence,
\bb
\rel{D_{\max}}{p'_n}{B_{n,m}^{\mathrm{cl}}(q_n)} &= \log \sup_{x^n\in \XX^n,\ \|t_{x^n} - s\|_\infty\leq \delta} \frac{p'_n(x^n)}{\big(B_{n,m}^{\mathrm{cl}}(q_n)\big)(x^n)} \\
&= \log \sup_{x^n\in \XX^n,\ \|t_{x^n} - s\|_\infty\leq \delta} \frac{p'_n\big(T_{n,\,t_{x^n}}\big) \big/ \big|T_{n,\,t_{x^n}}\big|}{\big(B_{n,m}^{\mathrm{cl}}(q_n)\big)\big(T_{n,\,t_{x^n}}\big) \big/ \big|T_{n,\,t_{x^n}}\big|} \\
&= \log \sup_{t_n \in \mathcal{T}_n,\ \|t_n - s\|_\infty\leq \delta} \frac{p'_n\big(T_{n,t_n}\big)}{\big(B_{n,m}^{\mathrm{cl}}(q_n)\big)\big(T_{n,t_n}\big)} \\
&\leq - \log \inf_{t_n\in \mathcal{T}_n,\ \|t_n - s\|_\infty\leq \delta} \big(B_{n,m}^{\mathrm{cl}}(q_n)\big)\big(T_{n,t_n}\big) .
\label{blurring_lemma_proof_second_estimate}
\ee
In other words, in order to guarantee that the left-hand side be small we need to make sure that $B_{n,m}^{\mathrm{cl}}(q_n)$ ends up having some sizeable weight on all the type classes that are close to $s$. To do this, we will observe that if $q_n$ has some non-zero weight on the types close to $s$, then blurring makes that weight spill over to \emph{all} types close to $s$. 

To formalise the above intuition we need to ask ourselves how to estimate $\big(B_{n,m}^{\mathrm{cl}}(q_n)\big)(T_{n,t})$ for some type $t\in \mathcal{T}_n$ with 
\bb
\|s-t\|_\infty\leq \delta\, .
\label{blurring_lemma_proof_closeness}
\ee
First of all, decompose $q_n = \sum_{x^n} q_n(x^n)\,\delta_{x^n}$ as the sum of the deterministic probability distributions 
\bb
\delta_{x^n}(y^n) = \left\{ \begin{array}{ll} 1 & \quad \text{$y^n=x^n$,} \\[.4ex] 0 & \quad \text{otherwise.} \end{array} \right. 
\ee
Up to summing, we only need to estimate $\big(B_{n,m}^{\mathrm{cl}}(\delta_{x^n})\big)(T_{n,t})$ for a fixed sequence $x^n$. Looking at the action of the blurring map, given by~\eqref{classical_blurring}, it is not difficult to realise that $B_{n,m}^{\mathrm{cl}}(\delta_{x^n})$ depends only on the \emph{type} of $x^n$, denoted by $u \coloneqq t_{x^n}$, rather than on the whole sequence. 

A little thought reveals that a way to represent the action of $B_{n,m}^{\mathrm{cl}}(\delta_{x^n})$ is to mix in an urn $n u(x) + m$ marbles of different colours $x\in \XX$ and to draw $n$ of them at random. In this representation, the value of $\big(B_{n,m}^{\mathrm{cl}}(\delta_{x^n})\big)(T_{n,t})$ is simply the probability of extracting $n t(x)$ marbles of each colour $x\in \XX$. To calculate it we can use the multivariate hypergeometric distribution in~\eqref{multivariate_hypergeometric}. Introducing the type
\bb
v_u(x) \coloneqq \frac{n u(x) + m}{n+m |\XX|} = \frac{n u(x) + \ceil{2\delta n}}{n+\ceil{2\delta n}|\XX|}
\ee
we obtain that
\bb
\big(B_{n,m}^{\mathrm{cl}}(\delta_{x^n})\big)(T_{n,t}) = \hyp{n+m|\XX|\,}{v_u}{n}(t)\, ,
\ee
and therefore
\bb
\big(B_{n,m}^{\mathrm{cl}}(q_n)\big)(T_{n,t}) &= \sum_{x^n} q_n(x^n)\, \big(B_{n,m}^{\mathrm{cl}}(\delta_{x^n})\big)(T_{n,t}) \\
&= \sum_{u\in \mathcal{T}_n} \,\sum_{x^n \in T_{n,u}} q_n(x^n)\, \hyp{n+m|\XX|\,}{v_u}{n}(t) \\
&= \sum_{u\in \mathcal{T}_n} q_n(T_{n,u})\, \hyp{n+m|\XX|\,}{v_u}{n}(t)\, . 
\ee
We can therefore estimate
\bb
\big(B_{n,m}^{\mathrm{cl}}(q_n)\big)(T_{n,t}) &\geqt{(i)} \!\sum_{u\in \mathcal{T}_n:\, \|s-u\|_\infty\leq \delta} q_n(T_{n,u})\, H_{n+m|\XX|,\, v_u;\, n}(t) \\
&\geqt{(ii)}\ 2^{-n g\big( \frac{m |\XX|}{n} \big)} \sum_{u\in \mathcal{T}_n:\, \|s-u\|_\infty\leq \delta} q_n(T_{n,u}) \\
&\geqt{(iii)}\ 2^{-n g\big( \big(2\delta + \frac1n\big) |\XX| \big)}\, q_n\left(\bigcup\nolimits_{u\in \mathcal{T}_n:\ \|s - u\|_\infty \leq \delta} T_{n,u} \right) .
\label{blurring_lemma_proof_third_estimate}
\ee
Here, (i)~follows by simply restricting the sum, while in~(ii) we used Lemma~\ref{hypergeometric_inequality_lemma}, which is applicable because for all $x\in \XX$
\bb
\big(n+m|\XX|\big)\, v_u(x) &= n u(x) + m = nu(x) + \ceil{2\delta n} \geq n \left(t(x) - 2\delta\right) + \ceil{2\delta n} \geq n t(x)\, ,
\ee
where the first inequality follows because
\bb
\|u-t\|_\infty \leq \|s - u\|_\infty + \|s - t\|_\infty \leq 2\delta
\ee
due to~\eqref{blurring_lemma_proof_closeness}. Finally, in~(iii) we used that $m=\ceil{2\delta n} \leq 2\delta n+1$, observing that the function $g$ is monotonically increasing.

Now we can really conclude. Using first~\eqref{blurring_lemma_proof_first_estimate}, then~\eqref{blurring_lemma_proof_second_estimate}, and finally~\eqref{blurring_lemma_proof_third_estimate}, yields
\bb
\rel{D_{\max}^{\eta}}{p_n}{B_{n,m}^{\mathrm{cl}}(q_n)} &\leq \rel{D_{\max}}{p'_n}{B_{n,m}^{\mathrm{cl}}(q_n)} \\
&\leq - \log \inf_{t\in \mathcal{T}_n,\ \|t - s\|_\infty\leq \delta} \big(B_{n,m}^{\mathrm{cl}}(q_n)\big)\big(T_{n,t}\big) \\
&\leq -\log q_n\!\left(\bigcup\nolimits_{u\in \mathcal{T}_n:\ \|s - u\|_\infty \leq \delta} T_{n,u} \right) + n g\!\left( \left(2\delta \!+\! \tfrac1n\right)\! |\XX| \right) ,
\ee
thereby concluding the proof.
\end{proof}

\subsection{The classical generalised Stein's lemma: hypothesis testing beyond the i.i.d.\ setting}
\label{subsec_classical_proof}

Armed with the classical blurring lemma, we will now present a full proof of the classical generalised Stein's lemma (Theorem~\ref{classical_GSL_thm}). Before we do that, we recall a useful result due to Sanov, and sometimes known as \emph{Sanov's theorem}~\cite{Sanov1957} (see also~\cite[Problem~11(a), p.~43]{CSISZAR-KOERNER}).

\begin{lemma}[(Sanov)] \label{Sanov_lemma}
Let $\XX$ be a finite alphabet, $\pazocal{A}\subseteq \PP(\XX)$ any set of probability distributions on $\XX$, and $p\in \PP(\XX)$. Then for all $n\in \N^+$ it holds that
\bb
p^{\otimes n} \left(\bigcup\nolimits_{t_n \in \mathcal{T}_n \cap \pazocal{A}} T_{n,\,t_n}\right) \leq (n+1)^{|\XX|-1}\, 2^{-n D(\pazocal{A}\|p)}\, .
\ee 
\end{lemma}

\begin{proof}[Proof of Theorem~\ref{classical_GSL_thm}]
Let $\eta,\e\in (0,1)$ be such that $\eta+\e < 1$, as in the statement of the theorem, and let $n\in \N^+$ be a positive integer. By construction (see~\eqref{smoothed_D_max}), there exists a probability distribution $p'_n\in\PP(\XX^n)$ such that
\bb
\frac12 \left\|p'_n - p^{\otimes n}\right\|_1\leq \e\, ,\qquad D_{\max}(p'_n\|\FF_n) = \rel{D_{\max}^\e}{p^{\otimes n}}{\FF_n}\, .
\label{classical_GSL_thm_proof_eq1}
\ee
Up to applying a uniformly random permutation, which never increases $D_{\max}(\cdot\|\FF_n)$ due to Axioms~1 and~5 in Section~\ref{sec_intro}, we can assume without loss of generality that $p'_n$ is permutationally symmetric.

As a preliminary calculation, note that
\bb
1 - p^{\otimes n}\left( \bigcup\nolimits_{t\in \mathcal{T}_n:\, \|p-t\|_\infty \leq \delta_n} T_{n,t} \right) &= p^{\otimes n}\left( \bigcup\nolimits_{t\in \mathcal{T}_n:\, \|p-t\|_\infty > \delta_n} T_{n,t} \right) \\
&\leqt{(i)} (n+1)^{|\XX|-1}\, 2^{- n \inf_{t\in \mathcal{T}_n:\, \|p-t\|_\infty > \delta_n} D(t\|p)} \\
&\leqt{(ii)} (n+1)^{|\XX|-1}\, e^{- 2n \delta_n^2} \\
&\eqt{(iii)} \,\eta\, .
\ee
Here, the inequality in~(i) follows from Sanov's theorem (Lemma~\ref{Sanov_lemma}), that in~(ii) descends from Pinsker's inequality~\cite{Pinsker, Csiszar1967, Kullback1967} via the calculation
\bb
D(t\|p) \geq \frac{\log e}{2}\, \|p-t\|_1^2 \geq 2 (\log e)\, \|p-t\|_\infty^2\, ,
\ee
and in~(iii) we employed~\eqref{delta_n}. As a consequence,
\bb
p'_n \left( \bigcup\nolimits_{t\in \mathcal{T}_n:\, \|p-t\|_\infty \leq \delta_n} T_{n,t} \right) \geq 1 - \e - \eta\, .
\ee
We can now apply Lemma~\ref{blurring_lemma} with $p_n \mapsto p^{\otimes n}$, $q_n\mapsto p'_n$, $s\mapsto p$, and $\delta \mapsto \delta_n$, obtaining
\bb
&\rel{D_{\max}^{\eta}}{p^{\otimes n}\!}{\!B_{n,m}^{\mathrm{cl}}(p'_n)} \leq \log \frac{1}{1\!-\!\e\!-\!\eta} + n g\!\left( \left(2\delta_n \!+\! \tfrac1n\right)\! |\XX| \right) ,
\ee
for $m = \ceil{2n\delta_n}$. Using the easily verified triangle inequality $D_{\max}(p\|q) \leq D_{\max}(p\|r) + D_{\max}(r\|q)$ for the max-relative entropy, we can now write that
\bb
\rel{D_{\max}^\eta}{p^{\otimes n}}{\FF_n} &\leq \rel{D_{\max}^{\eta}}{p^{\otimes n}}{B_{n,m}^{\mathrm{cl}}(p'_n)} + \rel{D_{\max}}{B_{n,m}^{\mathrm{cl}}(p'_n)}{\FF_n} \\
&\leq\, \log \frac{1}{1-\e-\eta} + n g\!\left( \left(2\delta_n + \tfrac1n\right) |\XX| \right) + \rel{D_{\max}}{B_{n,m}^{\mathrm{cl}}(p'_n)}{\FF_n} \\
&\leqt{(iv)}\, \log \frac{1}{1-\e-\eta} + n g\!\left( \left(2\delta_n + \tfrac1n\right) |\XX| \right) + \rel{D_{\max}}{p'_n}{\FF_n} + m |\XX| \log\tfrac{1}{c} \\
&\leqt{(v)}\, \log \frac{1}{1-\e-\eta} + n g\!\left( \left(2\delta_n + \tfrac1n\right) |\XX| \right) + \rel{D_{\max}^\e}{p^{\otimes n}}{\FF_n}  + \left(2n\delta_n + 1\right) |\XX| \log\tfrac{1}{c}\, .
\ee
Here, in~(iv) we observed that the classical blurring map can only increase the max-relative entropy of resource by at most $m|\XX| \log(1/c)$, because it adds at most $m$ copies of the probability distributions $\delta_x$, for all $x\in \XX$, and $\delta_x \leq \frac1c q_0$, where $q_0$ is the probability distribution with full support --- called $\sigma_0$ in the quantum case --- whose existence is guaranteed by Axiom~2. We are going to repeat this step in the quantum case, so we refer the reader to the forthcoming Lemma~\ref{D_max_elementary_lemma} for a more detailed justification of it. In~(v) we instead exploited~\eqref{classical_GSL_thm_proof_eq1} and remembered that $m=\ceil{2n\delta_n}$.
\end{proof}

\section{The quantum solution} \label{sec_quantum_solution}

This section is devoted to the full proof of the generalised quantum Stein's lemma (Theorem~\ref{GQSL_thm}). The linchpin of our approach will be, once again, a suitably quantised version of the classical blurring lemma, the forthcoming Lemma~\ref{quantum_blurring_lemma}. We will first show how to solve the generalised quantum Stein's lemma using Lemma~\ref{quantum_blurring_lemma}, and then devote the rest of the paper (Sections~\ref{subsec_alternative_expressions_blurring}--\ref{subsec_proof_quantum_blurring_lemma}) to the proof this latter result.

\subsection{An informal description of the argument} \label{subsec_informal_quantum}

The first part of the argument proceeds as in the classical case (Section~\ref{subsec_intuitive_description_classical}). For some $\e\in (0,1)$, due to Remark~\ref{reformulations_rem} we only have to prove that $\liminf_{n\to\infty} \frac1n\, \rel{D_{\max}^\e}{\rho^{\otimes n}}{\FF_n} \geq D^\infty(\rho\|\FF)$. As in the classical case, we proceed by contradiction and assume that $\rel{D_{\max}^\e}{\rho^{\otimes n}}{\FF_n} \lesssim n \lambda$ for some $\lambda < D^\infty(p\|\FF)$, so that, for all sufficiently large $n$, we can write that $\rho_n \leq 2^{n\lambda} \sigma_n$ for some approximating state $\rho_n \approx_\e \rho^{\otimes n}$ and some free state $\sigma_n$. Another similarity with the classical case is that both $\rho_n$ and $\sigma_n$ can be assumed to be permutationally invariant without loss of generality, due to Axioms~1 and~5.

In the classical case the argument continued with the application of the classical blurring map to $\rho_n$. We can try to replicate the same procedure here, constructing a quantum blurring map whose action is composed of the following three steps:
\begin{enumerate}[(i)]
\item append to the input state $\floor{\delta n}$ copies of $\rho$;
\item apply a uniformly random permutation over the $n+\floor{\delta n}$ systems;
\item discard $\floor{\delta n}$ systems, so as to go back to an $n$-copy system.
\end{enumerate}
We denote by $\widetilde{\rho}_n$ the state obtained by applying the above quantum blurring procedure to $\rho_n$. Since among the steps above only~(i) adds some resource to the input state, it is not difficult to prove that also $\widetilde{\rho}_n$ has a max-relative entropy of resource bounded by $n\lambda$ up to a small correction, in the sense that $\widetilde{\rho}_n \leq 2^{n(\lambda + \delta')} \sigma'_n$, where $\delta'$ is a constant multiple of $\delta$, and $\sigma'_n$ is another free state (in general different from $\sigma_n$).

Asymptotic continuity, which works in the quantum just as well as in the classical case, implies that when $\eta \in (0,1)$ is really small ($\eta\to 0$) we will have $\rel{D_{\max}^\eta}{\rho^{\otimes n}}{\FF_n} \gtrsim n D^\infty(\rho\|\FF)$. Once again, however, the problem is to connect a statement like this for $\eta \to 0$ with the same statement when $\eta$ is replaced by a much larger $\e\in (0,1)$. The quantum blurring lemma, like the classical one, will do precisely that.

Informally, it should tell us that $\widetilde{\rho}_n$ dominates $\rho^{\otimes n}$ \emph{approximately}, in the sense that there exists some $\rho'_n \approx_\eta \rho^{\otimes n}$, with $\eta$ small, such that $\rho'_n \leq 2^{n\zeta} \widetilde{\rho}_n$ for some other small parameter $\zeta$. If we could show that, then we would get that
\begin{align}
D^\infty(\rho\|\FF) &\lesssim \frac1n\, \rel{D_{\max}^\eta}{\rho^{\otimes n}}{\FF_n} \nonumber \\
&\leq \frac1n\, \rel{D_{\max}^\eta}{\rho^{\otimes n}}{\widetilde{\rho}_n} + \frac1n\, \rel{D_{\max}}{\widetilde{\rho}_n}{\FF_n} \\
&\leq \zeta + \lambda + \delta'\, . \nonumber
\end{align}
Taking $\delta'$ and $\zeta$ sufficiently small, we would reach a contradiction with the assumption that $\lambda < D^\infty(\rho\|\FF)$, and we would be done.

A beautiful lemma by Datta and Renner~\cite{Datta2009} entails that the existence of some $\rho'_n \approx_\eta \rho^{\otimes n}$ such that $\rho'_n \leq 2^{n\zeta} \widetilde{\rho}_n$ is equivalent to the inequality
\bb
\Tr \left( \rho^{\otimes n} - 2^{n\zeta} \widetilde{\rho}_n \right)_+ \leqt{?} \eta'\, ,
\label{informal_key_inequality}
\ee
where $X_+$, defined by~\eqref{positive_part}, denotes the positive part of an operator $X$ (where we keep only the positive eigenvalues), and $\eta'$ is another small parameter related to $\eta$. To prove the above statement, we first show that it suffices to solve the case where both $\rho$ and $\rho_n$ are pure states, with $\rho_n$ permutationally symmetric --- and hence supported on the symmetric space. This might appear counter-intuitive, but it really just follows from the existence of symmetric purifications of any permutationally symmetric state, together with the data processing inequality for the function $\Tr X_+$.

In fact, to simplify the notation we can set $\rho = \ketbra{0}$ equal to the projector onto the first vector of the computational basis. Since we have never chosen a basis so far, this is by no means a loss of generality. But it does suggest an idea. To describe it, it is convenient to consider the case where the local systems are single-qubit systems (i.e.\ $d=2$). The symmetric space $\mathrm{Sym}^n(\C^2)$, on which $\rho_n$ is supported, can be naturally embedded into the Fock space $\ell^2(\N)$ spanned by the Fock states $\ket{0}$, $\ket{1}$, $\ket{2}$, etc. In quantum physics, such a space models a quantum harmonic oscillator, or a single mode of light of definite frequency and polarisation. For this reason, it is also called a bosonic mode. The embedding is easy to understand: the Fock number counts the number of $1$'s appearing in the computational basis representation of a vector in $\mathrm{Sym}^n(\C^2)$. For example, $\ket{0^n}\in \mathrm{Sym}^n(\C^2)$ corresponds to the vacuum state $\ket{0}$, the symmetrised `one-excitation state' $\frac{1}{\sqrt{n}}\left(\ket{100\ldots 0} + \ket{010\ldots 0} + \ldots + \ket{0\ldots 01}\right)$ corresponds to $\ket{1}$, and so forth.

The second technical innovation we introduce is a procedure to lift the above problem to the bosonic space. We will prove that the quantum blurring map corresponds, roughly speaking, to a \emph{pure loss bosonic channel}, one of the most studied objects in bosonic quantum information theory. This is not entirely surprising, given that the action of blurring is to scatter some excitations into an environment, thereby losing them. However, looking at this from a distance, we found it quite surprising --- and, frankly, somewhat gratifying --- that this solution of the generalised quantum Stein's lemma hides at its core a pure loss channel. 

At any rate, most of the work needed to establish the quantum solution lies in proving that in the limit $n\to\infty$ the finite-dimensional problem~\eqref{informal_key_inequality} is mapped to a relation that is \emph{roughly} (but not precisely) equivalent to
\bb
\lim_{M\to\infty} \Tr \left(\ketbra{0} - M \EE_{1-\delta}(\omega) \right)_+ \eqt{?} 0\, ,
\label{informal_key_inequality_2}
\ee
where $\omega$ is an arbitrary state on a single bosonic mode, and $\EE_\lambda$ is the pure loss channel with transmissivity $\lambda$. (Note that the diverging coefficient $2^{n\zeta}$ in~\eqref{informal_key_inequality} has been replaced by an external limit $M\to\infty$ in~\eqref{informal_key_inequality_2}.) Some simple linear algebra reveals that~\eqref{informal_key_inequality_2} amounts to saying that $\ket{0}$ belongs to the support of $\EE_{1-\delta}(\omega)$, for all $\omega$. But this is clearly false: by choosing as $\omega$ a coherent state (see the forthcoming~\eqref{coherent_state}), due to~\eqref{pure_loss_action_on_coherent} we get at the output another coherent state, and in particular a pure state that is not the vacuum. However, not all hope is lost, because coherent states are really the \emph{only} states that remain pure after the action of the pure loss channel --- most states, on the contrary, will become heavily mixed. To fix this last problem it suffices to choose a slightly modified version of the blurring map in which $\delta$, the fraction of systems we mix in, is a \emph{random variable}, e.g.\ uniformly distributed in $[0,\Delta]$, for some small $\Delta$. Eq.~\eqref{informal_key_inequality_2} then is transformed into
\bb
\lim_{M\to\infty} \Tr \left(\ketbra{0} - M \int_0^\Delta \frac{\dd\delta}{\Delta}\ \EE_{1-\delta}(\omega) \right)_+ \eqt{?} 0\, ,
\label{informal_key_inequality_3}
\ee
As it turns out, the vacuum state $\ket{0}$ does belong to the support of $\int_0^\Delta \frac{\dd\delta}{\Delta}\ \EE_{1-\delta}(\omega)$, for all input states $\omega$. This will allow us to prove a statement analogous to~\eqref{informal_key_inequality_3} and thereby conclude the proof.

\subsection{The quantum blurring lemma} \label{subsec_quantum_blurring_lemma}

We will now discuss how to quantise the classical blurring map~\eqref{classical_blurring} that played a key role in the classical solution of the generalised Stein's lemma. For some positive integer $d\in \N$, some state $\rho\in \D\big(\C^d\big)$ in dimension $d$, and some $\delta\in (0,\frac12]$, consider the \deff{$\boldsymbol{\rho}$-dependent blurring map}
\bb
\widebar{B}_{n,\delta}^\rho &: \LL\big((\C^d)^{\otimes n} \big) \to \LL\big((\C^d)^{\otimes n} \big)\, , \\
\widebar{B}_{n,\delta}^\rho (X) &\coloneqq \Tr_{\floor{\delta n}} \mathcal{S}_{n+\floor{\delta n}}\big(X \otimes \rho^{\otimes \floor{\delta n}} \big)\, ,
\label{quantum_blurring}
\ee
where for some $N$-partite quantum system $A^N$ the symmetrisation map is defined by
\bb
\mathcal{S}_N(X) \coloneqq \frac{1}{N!} \sum_{\pi\in S_N} U_\pi^{\vphantom{\dag}} X U_\pi^{\dag}\, ,
\label{symmetrisation}
\ee
with $U_\pi$ being the unitary that implements the permutation $\pi$ over the $N$ copies of $A$, and $S_N$ denoting the symmetric group. Also, the operation $\Tr_{\floor{\delta n}}$ appearing in~\eqref{quantum_blurring} denotes the partial trace over $\floor{\delta n}$ of the $n+\floor{\delta n}$ copies of the space $\C^d$ --- which ones is irrelevant, as the state under examination is permutationally invariant. The action of the map in~\eqref{quantum_blurring} is similar to that of the classical blurring map in~\eqref{classical_blurring}: in both cases, we add some noise by tracing away randomly chosen sub-systems of the input state and shuffling in a small number of systems in a reference state.

In~\eqref{quantum_blurring} we chose to blur by adding copies of $\rho$. This is not really important, as we could equivalently choose to add copies of any other state with support no smaller than that of $\rho$. In fact, without loss of generality we could add copies of the maximally mixed state $\id/d$ on $\C^d$. However, blurring with $\rho$ simplifies the analysis considerably, especially when $\rho$ is pure. Since we will later see that we can always anyway assume that $\rho$ be pure without loss of generality, it is not only instructive but also useful to look at this case more closely. In fact, choosing wisely the basis we work with in $\C^d$, we can assume that $\rho = \ketbra{0}$ coincides with the projector onto the first vector of that basis. In the pure state case we are going to make also another simplification. Denoting with $\Pi_n$ the projector onto the symmetric subspace $\mathrm{Sym}^n(\C^d)$ of $(\C^d)^{\otimes n}$ and assuming that the input is also supported on the same space, we can also construct the alternative \deff{quantum blurring map}
\bb
&B_{n,\delta} (X): \LL\big(\mathrm{Sym}^n\big(\C^d\big)\big) \to \LL\big(\mathrm{Sym}^n\big(\C^d\big)\big)\, , \\[4pt]
&B_{n,\delta} (X) \coloneqq \Pi_n\, \widebar{B}^{\,\ket{0}\!\bra{0}}_{n,\delta} (X)\, \Pi_n = \Pi_n \left(\Tr_{\floor{\delta n}} \mathcal{S}_{n+\floor{\delta n}}\left(X \otimes \ketbra{0}^{\otimes \floor{\delta n}} \right) \right) \Pi_n\, .
\label{blurring}
\ee
(With a slight but very convenient abuse of notation, throughout this paper we will consider $\mathrm{Sym}^n(\C^d)$ alternatively as a subspace of $(\C^d)^{\otimes n}$ or as a space on its own.) Throughout most of Section~\ref{subsec_alternative_expressions_blurring}--\ref{subsec_convergence}, when we will refer to the quantum blurring map we will be talking about~\eqref{blurring} rather than~\eqref{quantum_blurring}. Note that unlike $\widebar{B}_{n,\delta}^\rho$, which is a proper quantum channel, due to the presence of the projector $\Pi_n$ the map $B_{n,\delta}$ is only a \emph{sub-}channel.

The key technical result that underpins our proof of the generalised quantum Stein's lemma is the following quantum version of the classical blurring lemma (Lemma~\ref{blurring_lemma}).

\begin{lemma}[(Asymptotic quantum blurring lemma)] \label{quantum_blurring_lemma}
Let $\rho\in \D\big(\C^d\big)$ be a finite-dimensional state, and for some infinite set $I\subseteq \N$ let $(\rho_n)_{n\in I}$ be a sequence of permutationally symmetric $n$-copy states $\rho_n = \mathcal{S}_n(\rho_n) \in \D\big( (\C^d)^{\otimes n}\big)$ (see~\eqref{symmetrisation} for a definition of $\mathcal{S}_n$) such that
\bb
\limsup_{n\in I} \frac12\, \big\| \rho_n - \rho^{\otimes n}\big\|_1 \leq \e
\ee
for some $\e\in (0,1)$. Then there exists an infinite subset $I'\subseteq I$ such that, for all $\Delta\in (0,\frac12]$,
\bb
\lim_{M\to\infty} \limsup_{n\in I'} \Tr \left( \rho^{\otimes n} - M \int_0^{\Delta} \frac{\dd \delta}{\Delta}\ \widebar{B}_{n,\delta}^\rho(\rho_n) \right)_+ = 0\, ,
\ee
where $\Tr X_+$ denotes the trace of the positive part of $X$ (see~\eqref{positive_part} for a definition).
\end{lemma}

The proof of the above result is deferred to Section~\ref{subsec_proof_quantum_blurring_lemma}. Before we move on, it is instructive to note the main difference between Lemma~\ref{blurring_lemma}, which is a one-shot result, and Lemma~\ref{quantum_blurring_lemma}, which is instead asymptotic. This is not a coincidence: while the $\rho$-dependent blurring map~\eqref{quantum_blurring} is very similar, from a conceptual standpoint, to its classical counterpart~\eqref{classical_blurring}, the availability of the formalism of types makes the analysis of its action much easier in the classical case, even in the regime of finite $n$. Without types, the quantum proof must rely on the aforementioned bosonic lifting procedure, which allows us to regain control of blurring in the asymptotic limit where $n\to\infty$. The intrinsically asymptotic nature of this argument is such that we have not yet been able to obtain a one-shot version of the quantum blurring lemma. This should be possible, however, and as an immediate implication one would get a one-shot control over the behaviour of the hypothesis testing relative entropy associated with resource testing.

\subsection{Proof of the generalised quantum Stein's lemma using quantum blurring} 
\label{subsec_proof_GQSL}

Here we present a proof of Theorem~\ref{GQSL_thm} that assumes the quantum blurring lemma (Lemma~\ref{quantum_blurring_lemma}). This latter result will be proved in Section~\ref{subsec_proof_quantum_blurring_lemma}. We start with two well-known preliminary results.

\begin{lemma} \label{D_max_elementary_lemma}
Let $\HH$ be a finite-dimensional Hilbert space, and let $(\FF_n)_n$ be a sequence of sets of states $\FF_n\subseteq \D\big(\HH^{\otimes n}\big)$ that obeys Axioms~2 and~4 in Section~\ref{sec_main_result}. Then the max-relative entropy of resource is sub-additive, meaning that
\bb
&\rel{D_{\max}}{\rho_n\!\otimes\! \rho_m\!}{\!\FF_{n+m}} \leq \rel{D_{\max}}{\rho_n\!}{\!\FF_n} + \rel{D_{\max}}{\rho_m\!}{\!\FF_m}
\label{subadditivity_D_max}
\ee
for all $\rho_n \in \D\big(\HH^{\otimes n}\big)$ and $\rho_m\in \D\big(\HH^{\otimes m}\big)$, and obeys the universal upper bound
\bb
\rel{D_{\max}}{\rho_n}{\FF_n} \leq n \log\frac1c \qquad \forall\ \rho_n \in \D\big(\HH^{\otimes n}\big)\, ,
\label{universal_upper_bound_D_max}
\ee
where $c>0$ is the constant whose existence is guaranteed by Axiom~2.
\end{lemma}

\begin{proof}
If $\rho_n \leq 2^{\lambda} \sigma_n$ and $\rho_m \leq 2^{\mu} \sigma_m$ for some $\sigma_n\in \FF_n$ and $\sigma_m\in \FF_m$, then clearly $\rho_n\otimes \rho_m \leq 2^{\lambda+\mu} \sigma_n\otimes \sigma_m$, which implies~\eqref{subadditivity_D_max} once one observes that $\sigma_n\otimes \sigma_m\in \FF_{n+m}$ via Axiom~4, and then minimises over $\lambda$ and $\mu$. Eq.~\eqref{universal_upper_bound_D_max} follows instead from the simple operator upper bound $\rho_n \leq \id \leq c^{-n} \sigma_0^{\otimes n}$, with $\sigma_0$ being the state given by Axiom~2.
\end{proof}

The following is a slight rephrasing of a result by Winter~\cite[Lemma~7]{tightuniform}, in turn a sharpening of the main result of~\cite{Donald1999} and of~\cite[Proposition~3.23]{MatthiasPhD}. We do not report its proof here.

\begin{lemma}[(Asymptotic continuity of the relative entropy of resource~\cite{Donald1999, MatthiasPhD, tightuniform})] \label{AC_lemma}
Let $\FF\subseteq \D(\HH)$ be a closed convex set of finite-dimensional states that contains a full-rank state $\sigma_0\in \FF$ with $\sigma_0 \geq c\id > 0$. Then, for all $\rho,\rho'\in \D(\HH)$ with $\frac12\|\rho - \rho'\|_1\leq \e$, it holds that
\bb
\big| D(\rho\|\FF) - D(\rho'\|\FF) \big| \leq \e \log\tfrac1c + g(\e)\, ,
\ee
where $g$ is the function defined by~\eqref{g}.
\end{lemma}

We are now ready to present the proof of our main result.

\begin{proof}[Proof of Theorem~\ref{GQSL_thm} (generalised quantum Stein's lemma)]
Fix some $\e\in (0,1)$. By the discussion in Remark~\ref{reformulations_rem}, it suffices to prove~\eqref{D_max_reformulation}. Assume by contradiction that 
\bb
\liminf_{n\to\infty} \frac1n\, \rel{D_{\max}^\e}{\rho^{\otimes n}}{\FF_n} < D^\infty(\rho\|\FF)\, . 
\ee
Then there exists an infinite set $I\subseteq \N$ and two numbers $\lambda,\lambda'$ such that
\bb
\lim_{n\in I} \frac1n\, \rel{D_{\max}^\e}{\rho^{\otimes n}}{\FF_n} = \lambda' < \lambda < D^\infty(\rho\|\FF)\, .
\label{GQSL_proof_eq2}
\ee
This means that for all sufficiently large $n\in I$ we can find a state $\rho_n \in \D\big(\HH^{\otimes n}\big)$ such that
\bb
\frac12\left\|\rho_n - \rho^{\otimes n}\right\|_1 \leq \e\, ,\qquad \rel{D_{\max}}{\rho_n}{\FF_n} \leq n\lambda\, .
\label{GQSL_proof_eq3}
\ee
Up to applying the symmetrisation operator~\eqref{symmetrisation}, it is clear that we can assume without loss of generality that $\rho_n$ is permutationally symmetric. By the quantum blurring lemma (Lemma~\ref{quantum_blurring_lemma}), there exists some infinite subset $I'\subseteq I$ such that
\bb
\lim_{M\to\infty} \limsup_{n\in I'} \Tr \left( \rho^{\otimes n} - M \widetilde{\rho}_n \right)_+ = 0
\label{GQSL_proof_eq4}
\ee
for all $\Delta \in (0,\frac12]$, where the blurred states $\widetilde{\rho}_n$ are defined by
\bb
\widetilde{\rho}_n \coloneqq \int_0^\Delta \frac{\dd\delta}{\Delta}\ \widebar{B}_{n,\delta}^{\rho}(\rho_n)\, ,
\ee
with $\widebar{B}_{n,\delta}^{\rho}$ being given by~\eqref{quantum_blurring}. Our first goal is to estimate the max-relative entropy of resource of the states $\widetilde{\rho}_n$. We start by observing that, for all $\delta \in [0,\Delta]$,
\bb
D_{\max} \Big(\widebar{B}_{n,\delta}^{\rho}(\rho_n)\, \Big\|\, \FF_n\Big) &\leqt{(i)} \rel{D_{\max}}{\rho_n \otimes \rho^{\otimes \floor{\delta n}}}{\FF_n} \\
&\leqt{(ii)} \rel{D_{\max}}{\rho_n}{\FF_n} + \floor{\delta n} \log \frac1c \\
&\leqt{(iii)} n \left( \lambda + \delta \log\tfrac1c \right) \\
&\leq n \left( \lambda + \Delta \log\tfrac1c \right)\, .
\label{GQSL_proof_eq5}
\ee
Here, in~(i) we noticed that the max-relative entropy of resource is monotonically non-increasing under any quantum channel that maps free states to free states, due to the data processing inequality for $D_{\max}$~\cite{Datta08}. Both the partial trace and the symmetrisation operation satisfy this assumption, due to Axiom~3 and to Axioms~1 and~5, respectively. In~(ii) we applied the above Lemma~\ref{D_max_elementary_lemma}, while (iii)~follows from~\eqref{GQSL_proof_eq3}. 

Eq.~\eqref{GQSL_proof_eq5} tells us that for all $\delta\in [0,\Delta]$ we can find some $\sigma_{n,\delta} \in \FF_n$ such that
\bb
\widebar{B}_{n,\delta}^{\rho}(\rho_n) \leq 2^{n \left(\lambda + \Delta \log\tfrac1c\right)} \sigma_{n,\delta}\, .
\ee
Integrating over $\delta$ and using the convexity of $\FF_n$, one obtains that
\bb
\widetilde{\rho}_n = \int_0^\Delta \frac{\dd\delta}{\Delta}\ \widebar{B}_{n,\delta}^{\rho}(\rho_n) \leq 2^{n \left(\lambda + \Delta \log(1/c)\right)}\,  \widetilde{\sigma}_{n}\, ,\qquad \widetilde{\sigma}_{n} \coloneqq \int_0^\Delta \frac{\dd\delta}{\Delta}\ \sigma_{n,\delta} \in \FF_n\, ,
\ee
that is,
\bb
\rel{D_{\max}}{\widetilde{\rho}_n}{\FF_n} \leq n \left(\lambda + \Delta \log\tfrac1c \right) .
\label{GQSL_proof_eq8}
\ee

Now, fix two arbitrarily small parameters $\eta,\zeta\in (0,1)$. Since for any fixed $M>0$
\bb
\limsup_{n\in I'} \Tr \left( \rho^{\otimes n} - 2^{n \zeta} \widetilde{\rho}_n \right)_+ \leq \limsup_{n\in I'}  \Tr \left( \rho^{\otimes n} - M \widetilde{\rho}_n \right)_+
\ee
due to Lemma~\ref{variational_program_trace_X_plus_lemma}(a), taking the limit $M\to\infty$ on the right-hand side and leveraging~\eqref{GQSL_proof_eq4} we conclude that
\bb
\lim_{n\in I'} \Tr \left( \rho^{\otimes n} - 2^{n \zeta} \widetilde{\rho}_n \right)_+ = 0\, ,
\ee
entailing that 
\bb
\Tr \left( \rho^{\otimes n} - 2^{n \zeta} \widetilde{\rho}_n \right)_+ \leq \eta
\ee
for all sufficiently large $n\in I'$. Looking at~\eqref{Datta_Leditzky}, it is clear that this is equivalent to the inequality
\bb
\rel{\widetilde{D}^\eta_{\max}}{\rho^{\otimes n}}{\widetilde{\rho}_n} \leq n\zeta\, .
\label{GQSL_proof_eq15}
\ee
We can now write that
\bb
\rel{\widetilde{D}^\eta_{\max}}{\rho^{\otimes n}}{\FF_n} &\leqt{(iv)} \rel{\widetilde{D}^\eta_{\max}}{\rho^{\otimes n}}{\widetilde{\rho}_n} + \rel{D_{\max}}{\widetilde{\rho}_n}{\FF_n} \\
&\leqt{(v)} n \left(\lambda + \zeta + \Delta \log\tfrac1c\right)
\ee
for all sufficiently large $n\in I'$. Here, in~(iv) we employed the following easily verified `triangle inequality' for the Datta--Leditzky smoothed max-relative entropy~\eqref{Datta_Leditzky}: for all triples of states $\rho,\sigma,\omega$ on the same system and all $\e\in [0,1]$,
\bb
\widetilde{D}^\e_{\max}(\rho\|\sigma) \leq \widetilde{D}^\e_{\max}(\rho\|\omega) + D_{\max}(\omega\|\sigma)\, .
\ee
Continuing, (v)~follows from~\eqref{GQSL_proof_eq8} and~\eqref{GQSL_proof_eq15}.

Hence
\bb
\limsup_{n\in I'} \frac1n\, \rel{\widetilde{D}^\eta_{\max}}{\rho^{\otimes n}}{\FF_n} \leq \lambda + \zeta + \Delta \log\tfrac1c\, .
\ee
Since $\eta$, $\zeta$, and $\Delta$ are arbitrary, we can now take them to zero, obtaining that\footnote{Note that $I'$ does not depend on either of these parameters.}
\bb
\lim_{\eta\to 0^+} \limsup_{n\in I'} \frac1n\, \rel{\widetilde{D}^\eta_{\max}}{\rho^{\otimes n}}{\FF_n} \leq \lambda < D^\infty(\rho\|\FF)\, .
\label{GQSL_proof_eq19}
\ee
It is now a routine matter to derive a contradiction from the above inequality. Using the relations~\eqref{Datta_Renner_lemma}, it is not difficult to show that~\eqref{GQSL_proof_eq19} is equivalent to stating that
\bb
\lim_{\eta\to 0^+} \limsup_{n\in I'} \frac1n\, \rel{D^\eta_{\max}}{\rho^{\otimes n}}{\FF_n} \leq \lambda < D^\infty(\rho\|\FF)\, .
\label{GQSL_proof_eq20}
\ee
This can be shown to be in contradiction with~\eqref{GQSL_proof_eq2}, due to the asymptotic continuity of the relative entropy of resource. In fact,
\bb
\rel{D^\eta_{\max}}{\rho^{\otimes n}}{\FF_n}\ &=\ \min_{\rho'_n:\ \frac12 \|\rho'_n - \rho^{\otimes n}\|_1 \leq \eta} \rel{D_{\max}}{\rho'_n}{\FF_n} \\
&\geqt{(vi)}\ \min_{\rho'_n:\ \frac12 \|\rho'_n - \rho^{\otimes n}\|_1 \leq \eta} \rel{D}{\rho'_n}{\FF_n} \\
&\geqt{(vii)}\ \rel{D}{\rho^{\otimes n}}{\FF_n} - \eta \log\tfrac{1}{c^n} - g(\eta)\, ,
\label{GQSL_proof_eq21}
\ee
where (vi)~follows from~\eqref{relation_D_D_max}, while in~(vii) we employed Lemma~\ref{AC_lemma} with $\FF \mapsto \FF_n$, $\rho\mapsto \rho^{\otimes n}$, $\rho' \mapsto \rho'_n$, $\e\mapsto \eta$, $\sigma_0 \mapsto \sigma_0^{\otimes n}$, and $c\mapsto c^n$ (here, $\sigma_0\in \FF_1$ and $c>0$ are given by Axiom~2). Now, dividing both sides of~\eqref{GQSL_proof_eq21} by $n$ and taking the limit in $n\in I'$ yields
\bb
\limsup_{n\in I'} \frac1n\, \rel{D^\eta_{\max}}{\rho^{\otimes n}}{\FF_n} &\geq \limsup_{n\in I'} \frac1n\, \rel{D}{\rho^{\otimes n}}{\FF_n} - \eta \log\tfrac1c \\
&= \lim_{n\to\infty} \frac1n\, \rel{D}{\rho^{\otimes n}}{\FF_n} - \eta \log\tfrac1c \\
&= D^\infty(\rho\|\FF) - \eta \log\tfrac1c\, ,
\label{GQSL_proof_eq22}
\ee
which in contradiction with~\eqref{GQSL_proof_eq2} once one takes the limit $\eta\to 0^+$ and uses~\eqref{GQSL_proof_eq20}. 
\end{proof}

\subsection{Technical preliminaries} \label{subsec_technical_preliminaries}

The remainder of Section~\ref{sec_quantum_solution} is devoted to proving the quantum blurring lemma (Lemma~\ref{quantum_blurring_lemma}). Before we delve into the proof, we introduce some technical tools that will be used in the subsequent discussion.

\subsubsection{Type basis of the symmetric space}

The symmetric space $\mathrm{Sym}^n\big(\C^d\big)$, i.e.\ the span of all vectors on $\big(\C^d\big)^{\otimes n}$ that are invariant under all permutations of the tensor factors, admits the canonical basis~\cite[Section~1]{Harrow-church}
\bb
\ket{n,t_n} \coloneqq \binom{n}{nt_n}^{-1/2} \sum_{x^n\in T_{n,t_n}} \ket{x^n}\, ,
\label{ket_n_k}
\ee
where $t_n\in \mathcal{T}_n$ ranges over all $n$-types, $T_{n,t_n}$ is the type class associated with $t_n$ (see~\eqref{type_class}), and we used the notation in~\eqref{size_type_class} for multinomials. Some elementary calculations involving these vectors are presented in the following lemma.

\begin{lemma} \label{n_k_calculations_lemma}
For all $n,r\in \N^+$ with $n\geq r$ and all sequences $x^r$ of type $w_r\in \mathcal{T}_r$, the partial overlap\footnote{That is, the overlap calculated only on the first $r$ copies of $\C^d$.} between $\ket{x^r}$ or $\ket{r,w_r}$ and a vector of the basis~\eqref{ket_n_k} is given by
\bb
\braket{x^r|n,t_n} &= \binom{r}{rw_r}^{-1/2} \braket{r,w_r|n,t_n} \\
&= \sqrt{\frac{\binom{n-r}{nt_n - rw_r}}{\binom{n}{nt_n}}}\, \Ket{n-r,\,\tfrac{1}{n-r}(nt_n - rw_r)}\, ,
\label{n_k_overlap}
\ee
where it is understood that the rightmost side is zero unless $nt_n \succeq rw_r$, i.e.\ $nt_n(x) \geq rw_r(x)$ for all $x\in \XX$.
\end{lemma}

\begin{proof}
In order to prove the equality between the leftmost and the rightmost side of~\eqref{n_k_overlap} in the non-trivial case where $nt_n\succeq rw_r$, we need to establish that
\bb
\bra{x^r} \left(\sumno_{y^n\in T_{n,t_n}} \ket{y^n}\right) = \sum_{z^{n-r}\, \in\, T_{n,\,\frac{1}{n-r}(nt_n -rw_r)}} \ket{z^{n-r}}\, .
\ee
This is rather obvious, and it follows from the fact that the only sequences $y^n$ that contribute to the sum on the left-hand side are those in which the first $r$ symbols form the sequence $x^r$. The remaining $n-r$ symbols have type $\frac{1}{n-r}(nt_n-rw_r)$, and moreover the sequence they form is uniformly distributed on $T_{n,\,\frac{1}{n-r}(nt_n -rw_r)}$. That $\braket{x^r|n,t_n} = \binom{r}{rw_r}^{-1/2} \braket{r,w_r|n,t_n}$ holds follows directly by substituting the explicit expression of $\ket{r,w_r}$ as in~\eqref{ket_n_k}.
\end{proof}

\subsubsection{Continuous-variable systems} \label{subsubsec_CV}

In what follows, we will review the basic formalism of continuous-variable quantum systems. For a more detailed presentation with additional details and derivations, we refer the reader to the review~\cite{Braunstein-review}, as well as to the monograph~\cite{BUCCO}. A system composed of $m\in \N^+$ bosonic modes is modelled by the Hilbert space
\bb
\HH_m \coloneqq \big(\ell^2(\N)\big)^{\otimes m}
\label{H_m}
\ee
spanned by the \deff{Fock states} $\ket{k_1,\ldots, k_m} = \ket{k}$, where $k\in \N^m$ ranges over all non-negative integer vectors of length $m$. The Fock state corresponding to $k=0$, which we will also denote by $\ket{0}^{\otimes m}$, is called the \deff{vacuum state}. On a single bosonic mode ($m=1$), the \deff{annihilation and creation operators}, denoted $a$ and $a^\dag$, respectively, are defined by the action $a\ket{k} = \sqrt{k}\ket{k-1}$ and $a^\dag\ket{k} = \sqrt{k+1}\ket{k+1}$ on the dense subspace of vectors with finite expansion in Fock basis. The \deff{number operator} $a^\dag a$ thus satisfies that $a^\dag a\ket{k} = k \ket{k}$ for all $k\in \N$. The eigenvectors of the annihilation operator are the \deff{coherent states}, defined as
\bb
\ket{\alpha} \coloneqq e^{-|\alpha|^2/2} \sum_{k=0}^\infty \frac{\alpha^k}{\sqrt{k!}}\, \ket{k}\, ,
\label{coherent_state}
\ee
where $\alpha\in \C$ is the corresponding eigenvalue: indeed, it is straightforward to verify that $a\ket{\alpha} = \alpha \ket{\alpha}$.

The physical process by which photons, or excitations in Fock space, are scattered into a large environment in the vacuum state is modelled by a \deff{pure loss channel} (also sometimes called `quantum-limited attenuator'), whose action is defined by the Kraus representation~\cite[Eq.~(4.6)]{Ivan2011}
\bb
\EE_\lambda (X) &= \sum_{\ell=0}^\infty \frac{1}{\ell!} \left(\frac{1}{\lambda} - 1\right)^\ell a^\ell \lambda^{\frac{a^\dag a}{2}} X \lambda^{\frac{a^\dag a}{2}} (a^\dag)^\ell ,
\label{pure_loss_Kraus}
\ee
where $\lambda\in (0,1)$ is called the \deff{transmissivity} of the channel. Note that $a^\ell \lambda^{\frac{a^\dag a}{2}}$ is effectively a bounded operator for all integers $\ell$. The pure loss channel acts in a particularly simple way on coherent states: for all $\lambda\in (0,1)$ and all $\alpha\in \C$, 
\bb
\EE_\lambda(\ketbra{\alpha}) = \ketbra{\sqrt\lambda\,\alpha}\, .
\label{pure_loss_action_on_coherent}
\ee
The action on Fock states, instead, can be described as follows: 
\bb
\EE_\lambda(\ketbraa{h}{k}) = \lambda^{\frac{h+k}{2}} \sum_{\ell=0}^{\infty} \sqrt{\binom{h}{\ell} \binom{k}{\ell}} \left(\frac{1}{\lambda} - 1\right)^{\ell} \ketbraa{h-\ell}{k-\ell}\, .
\label{pure_loss_action_on_Fock}
\ee
Note that the sum contains only finitely many (to be precise, $\min\{h,k\}+1$) non-zero terms.

\subsubsection{Operator topologies} \label{subsubsec_operator_topologies}

Given a separable Hilbert space $\HH$, a natural pair of Banach spaces associated with it are the space of trace class operators $\TT(\HH)$ equipped with the trace norm $\|\cdot\|_1$, which we have already encountered, and the space of compact operators $\K(\HH)$ equipped with the operator norm $\|X\|_\infty \coloneqq \sup_{\ket{\psi}\in \HH\setminus\{0\}} \|X \ket{\psi}\| \big/ \|\ket{\psi}\|$. The former space can be thought of as the dual of the latter, in formula $\TT(\HH) = \K(\HH)^*$. This means that besides its native \deff{trace norm topology}, $\TT(\HH)$ can also be equipped with a \deff{weak* topology} induced by this duality. According to this latter topology, a sequence\footnote{The weak* topology is not metrisable in general, unless $\HH$ is finite dimensional. However, the Banach--Alaoglu theorem implies that it is metrisable on bounded subsets (see, e.g., \cite[Corollary~2.6.20]{MEGGINSON}. Since here will only consider bounded sequences, we can avoid introducing nets and deal only with sequences.} $(X_n)_n$ of operators $X_n\in \TT(\HH)$ converges to some $X\in \TT(\HH)$, and we write $X_n\tendsn{w*} X$, if
\bb
\Tr X_n K \tendsn{} \Tr XK\qquad \forall\ K\in \K(\HH)\, .
\label{weak_star_convergence}
\ee
See~\cite[Chapter~2]{MEGGINSON} for a review of the properties of the weak* topology on dual Banach spaces. Applications of this topology in quantum information theory have flourished recently~\cite{Furrer2011, achievability, continuity-relent-resource}. An important observation of which we will make use here is that due to the Banach--Alaoglu theorem and because of the separability of $\TT(\HH)$ (in turn a consequence of the separability of $\HH$), every trace-norm-bounded and weak*-closed set $X\subset \TT(\HH)$ is (sequentially) weak* compact. We record this observation below.

\begin{lemma} \label{weak_star_seq_compactness_lemma}
Let $(X_n)_n$ be a trace-norm-bounded sequence of operators $X_n\in \TT(\HH)$, where $\HH$ is a separable Hilbert space. Then one can find a weak*-converging subsequence $\big(X_{n_k}\big)_{n_k}$, i.e. $X_{n_k} \tendsk{w^*} X$ for some $X\in \TT(\HH)$. If $X_n\geq 0$ and $\Tr X_n\leq 1$ for all $n$, then similarly $X\geq 0$ and $\Tr X\leq 1$.
\end{lemma}

We also need to briefly discuss a useful notion of topology on the set of quantum channels. Given a sequence $(\Lambda_n)_n$ of quantum channels $\Lambda_n: \TT(\HH) \to \TT(\HH')$ with the same input and output, we say that $\Lambda_n$ converges to another channel $\Lambda: \TT(\HH) \to \TT(\HH')$ with respect to the \deff{strong convergence topology}, and we write $\Lambda_n \tendsn{s} \Lambda$, if
\bb
\lim_{n\to\infty} \left\| \Lambda_n(\rho) - \Lambda(\rho) \right\|_1 = 0\qquad \forall\ \rho\in \D(\HH)\, .
\label{strong_convergence_channels}
\ee

\subsection{Alternative expressions of the blurring map} \label{subsec_alternative_expressions_blurring}

We now present some alternative and more insightful expressions for the action of the blurring map.

\begin{lemma} \label{blurring_Gammas_lemma}
For all $n\in \N^+$ and $\delta\in (0,\frac12]$, it holds that
\begin{align}
B_{n,\delta}(X) =&\ \sum_{r=0}^{\floor{\delta n}} H(n+\floor{\delta n},\floor{\delta n};\, n,r)\ \Gamma_{n,r}(X) \, \label{blurring_Gammas} \\
\Gamma_{n,r}(X) \coloneqq&\ \Pi_n \left( \ketbra{0}^{\otimes r} \otimes \Tr_r X \right) \Pi_n\, , \label{Gamma}
\end{align}
where the blurring map $B_{n,\delta}$ is defined by~\eqref{blurring}, and the hypergeometric distribution is given by~\eqref{hypergeometric}.
\end{lemma}

\begin{proof}
It suffices to remember the description of the action of the blurring map presented in Section~\ref{subsec_informal_quantum}. Calling $r$ the number of copies of the added state $\ketbra{0}$ that end up in one of the first $n$ systems --- those that are not traced out --- a little thought shows that the law of $r$ is hypergeometric by construction, i.e.\ $r\sim H(n+\floor{\delta n},\floor{\delta n};\, n,\cdot)$. For a fixed $r$, the action of $B_{n,\delta}$ is effectively represented by the map $\Gamma_{n,r}$ in~\eqref{Gamma}. This proves the claim.
\end{proof}

An important technical step is to derive a Kraus representation for the maps $\Gamma_{n,r}$ appearing in~\eqref{Gamma}. We achieve this via the lemma below. 

\begin{lemma} \label{Gamma_Kraus_lemma}
For all $n\in \N^+$, $\delta\in (0,\frac12]$, and $r\in \{0,1,\ldots, \floor{\delta n}\}$, 
\bb
\Gamma_{n,r} (X) =&\ \sum_{w_r \in \mathcal{T}_r} M_{r,w_r}^{\vphantom{\dag}} X M_{r,w_r}^\dag\, , \\
M_{r,w_r} \coloneqq& \sum_{\substack{t_n \in \mathcal{T}_n,\\ nt_n \succeq rw_r}} \sqrt{\frac{\binom{n-r}{nt_n - rw_r}^2 \binom{r}{rw_r}}{\binom{n}{nt_n - r w_r + re_0} \binom{n}{nt_n}}}\ \Ketbraa{n,\, t_n - \tfrac{r}{n}\, w_r + \tfrac{r}{n}\,e_0}{n,\, t_n\phantom{\big|}}\, ,
\ee
where $\ket{n,t_n}$ is defined by~\eqref{ket_n_k}.
\end{lemma}

\begin{proof}
Let $X$ be an arbitrary operator supported on the symmetric space $\mathrm{Sym}^n(\C^d)$. Then
\begin{align}
\sum_{w_r\in\mathcal{T}_r} \braket{r,w_r|X|r,w_r} &\eqt{(i)} \sum_{w_r\in\mathcal{T}_r} \binom{r}{rw_r}^{-1} \!\! \sum_{x^r\!,\,y^r\in T_{r,w_r}} \!\!\braket{x^r|X|y^r} \nonumber \\
&\eqt{(ii)} \sum_{w_r\in\mathcal{T}_r} \binom{r}{rw_r}^{-1} \!\! \sum_{x^r\!,\,y^r\in T_{r,w_r}} \!\!\braket{x^r|X U_{\pi(x^r\!,\,y^r)} |x^r} \nonumber \\
&\eqt{(iii)} \sum_{w_r\in\mathcal{T}_r} \binom{r}{rw_r}^{-1} \!\! \sum_{x^r\!,\,y^r\in T_{r,w_r}} \!\!\braket{x^r|X|x^r} \\
&= \sum_{w_r\in\mathcal{T}_r} \sum_{x^r\in T_{r,w_r}} \!\!\braket{x^r|X|x^r} \nonumber \\
&= \sum_{x^r} \braket{x^r|X|x^r} \nonumber \\
&= \Tr_r X\, . \nonumber
\end{align}
The justification of the above derivation is as follows: in~(i) we substituted~\eqref{ket_n_k}; in~(ii) we noticed that, since $x^r$ and $y^r$ have the same type, there must exist a permutation $\pi(x^r,y^r)$ such that the unitary implementing it, denoted by $U_{\pi(x^r,y^r)}$, will satisfy $\ket{y^r} = U_{\pi(x^r,y^r)} \ket{x^r}$; and in~(iii) we observed that $X U_{\pi(x^r,y^r)} = X$ because $X$ is supported on the symmetric space. 

Plugging this expression into the definition of $\Gamma_{n,r}$ in~\eqref{Gamma} yields
\bb
\Gamma_{n,r}(X) &= \Pi_n \!\left( \ketbra{0}^{\otimes r} \!\!\otimes \Tr_r \!X \right)\! \Pi_n =\! \sum_{w_r\in\mathcal{T}_r} M_{r,w_r} X M_{r,w_r}^\dag ,
\ee
where
\bb
M_{r,w_r} \coloneqq&\ \Pi_n \left( \ketbraa{0^r}{r,w_r} \otimes \id^{\otimes (n-r)} \right) \Pi_n \\
\eqt{(iv)}&\ \sum_{s_n,t_n\in \mathcal{T}_n} \braket{n,s_n|0^r} \braket{r,w_r|n,t_n} \ketbraa{n,s_n}{n,t_n} \\
\eqt{(v)}&\ \sum_{t_n\in \mathcal{T}_n} \sqrt{\frac{\binom{n-r}{nt_n - r w_r}^2 \binom{r}{rw_r}}{\binom{n}{nt_n - rw_r + r e_0} \binom{n}{nt_n}}}\ \Ketbraa{n,\, t_n - \tfrac{r}{n}\, w_r + \tfrac{r}{n}\, e_0)}{n,\,t_n\phantom{\big|}} ,
\ee
as claimed. Here, in~(iv) we used the decomposition $\Pi_n = \sum_{t_n\in\mathcal{T}_n} \ketbra{n,t_n}$ twice, while in~(v) we considered an arbitrary (fixed) sequence $x^r \in T_{r,w_r}$ and wrote 
\bb
\braket{n,s_n|0^r} \braket{r,w_r|n,t_n} &\eqt{(vi)} \sqrt{\frac{\binom{n-r}{ns_n - re_0} \binom{n-r}{nt_n - r w_r} \binom{r}{rw_r}}{\binom{n}{ns_n} \binom{n}{nt_n}}}\ \Braket{n-r,\, \tfrac{1}{n-r}(ns_n - r e_0) \,\big|\, n-r,\, \tfrac{1}{n-r}(nt_n - r w_r)} \\
&\eqt{(vii)} \sqrt{\frac{\binom{n-r}{nt_n - r w_r}^2 \binom{r}{rw_r}}{\binom{n}{nt_n - rw_r + r e_0} \binom{n}{nt_n}}}\ \delta_{ns_n - r e_0,\, nt_n - rw_r}\, ,
\ee
where in~(vi) we employed Lemma~\ref{n_k_calculations_lemma} twice, and in~(vii) we denoted by $\delta_{a,b}$ the Kronecker delta on integer vectors $a,b\in \N^\XX$. This completes the proof.
\end{proof}

We now observe that
\bb
\sum_{w_r \in \mathcal{T}_r} M_{r,w_r}^\dag M_{r,w_r}^{\vphantom{\dag}} &= \sum_{t_n\in \mathcal{T}_n} d_r(t_n) \ketbra{n,t_n} \eqqcolon D_r\, ,
\label{D_r}
\ee 
where
\bb
d_r(t_n) \coloneqq \sum_{\mathcal{T}_r \ni w_r \preceq \frac{n}{r} t_n} \frac{\binom{n-r}{nt_n - rw_r}^2 \binom{r}{rw_r}}{\binom{n}{nt_n - r w_r + re_0} \binom{n}{nt_n}}\, .
\label{d_r(t_n)}
\ee
The matrix $D_{r}$ is strictly positive definite, because for all $t_n\in \mathcal{T}_n$ we can always choose some $w_r\in \mathcal{T}_r$ such that $rw_r \preceq nt_n$. Thus, the map
\bb
\Theta_{n,r} (X) &\coloneqq \sum_{w_r\in \mathcal{T}_r} N_{r,w_r}^{\vphantom{\dag}} X N_{r,w_r}^\dag\, , \\
N_{r,w_r} &\coloneqq M_{r,w_r} D_{r}^{-1/2}
\label{Theta}
\ee
is a quantum channel, simply because
\bb
\sum_{w_r\in \mathcal{T}_r} N_{r,w_r}^\dag N_{r,w_r}^{\vphantom{\dag}} = D_r^{-1/2} \left( \sum_{w_r\in \mathcal{T}_r} M_{r,w_r}^\dag M_{r,w_r}^{\vphantom{\dag}}\right) D_r^{-1/2} = \id\, . 
\label{Theta_is_quantum_channel}
\ee
Clearly, the maps $\Gamma_{n,r}$ and $\Theta_{n,r}$ are related by the identity
\bb
\Gamma_{n,r}(X) = \Theta_{n,r}\Big( D_r^{1/2} X D_r^{1/2} \Big) .
\label{Gamma_in_terms_of_Theta}
\ee
Putting all together, from Lemma~\ref{blurring_Gammas_lemma} we obtain the following expression for the action of $B_{n,\delta}$.

\begin{cor} \label{blurring_Theta_cor}
For all $n\in \N^+$ and all $\delta\in (0,\frac12]$, the action of the blurring map~\eqref{blurring} can be expressed as
\begin{align}
B_{n,\delta}(X) &= \sum_{r=0}^{\floor{\delta n}} H(n+\!\floor{\delta n}, \floor{\delta n};\, n,r)\ \Theta_{n,r}\Big( \!D_r^{1/2} X D_r^{1/2}\! \Big) , \label{blurring_Theta}
\end{align}
where 
the quantum channel $\Theta_{n,r}$ is defined by~\eqref{Theta}.
\end{cor}

For what follows, it will be crucial to have a precise estimate of the values of the entries of the matrix $D_r$, i.e.\ of the sum appearing in~\eqref{d_r(t_n)}. We provide this estimate below.

\begin{lemma} \label{d_r_bound_lemma}
Let $n,r$ be positive integers such that 
\bb
\eta \leq \frac{r}{n} \leq 1 - \eta
\label{bounds_ratio_r_n}
\ee
for some $\eta \in (0,\frac12]$. Let $t_n\in \mathcal{T}_n$ be an $n$-type with the property that
\bb
n t_n(x) \geq N
\ee
for some $x\neq 0$ and some non-negative integer $N \leq n$. Then the positive real number $d_r(t_n)$ defined by~\eqref{d_r(t_n)} satisfies the inequality
\bb
d_r(t_n) \leq 3\, e^{- N \eta^2 /2}\, .
\ee
\end{lemma}

\begin{proof}
Without loss of generality, assume that $\max_{x\neq 0} t_n(x) = t_n(1) \geq N/n$. Let $w_r\in \mathcal{T}_r$ be a type such that $w_r \preceq \frac{n}{r} t_n$. Then
\begin{align}
\frac{\binom{n-r}{nt_n - rw_r}^2 \binom{r}{rw_r}}{\binom{n}{nt_n - r w_r + re_0} \binom{n}{nt_n}} &= \frac{\binom{n-r}{nt_n - rw_r} \binom{r}{rw_r}}{\binom{n}{nt_n}} \frac{\binom{n-r}{nt_n - rw_r}}{\binom{n}{nt_n - r w_r + re_0}} \nonumber \\
&\eqt{(i)} H_{n,t_n;\, r}(w_r)\, \frac{\binom{n-r}{nt_n - rw_r}}{\binom{n}{nt_n - r w_r + re_0}} \nonumber \\
&= H_{n,t_n;\, r}(w_r)\, \frac{(n-r)!}{n!} \prod_x \frac{\big((nt_n - r w_r + re_0)(x)\big)!}{\big((nt_n - r w_r)(x)\big)!} \nonumber \\
&\eqt{(ii)} H_{n,t_n;\, r}(w_r)\, \frac{(n-r)!}{n!}\, \frac{\big((nt_n - r w_r)(0) + r\big)!}{\big((nt_n - r w_r)(0)\big)!} \label{d_r_bound_proof_eq1} \\
&\eqt{(iii)} H_{n,t_n;\, r}(w_r)\, \prod_{j=0}^{n - n t_n(0) + rw_r(0) - r - 1} \frac{n-r-j}{n-j} \nonumber \\
&\leqt{(iv)} H_{n,t_n;\, r}(w_r) \left(1-\frac{r}{n}\right)^{n - nt_n(0) + rw_r(0) - r} \nonumber \\
&\leqt{(v)} H_{n,t_n;\, r}(w_r) \left(1-\eta\right)^{n - nt_n(0) + rw_r(0) - r} . \nonumber
\end{align}
Here: in~(i) we recalled the expression~\eqref{multivariate_hypergeometric} of the multivariate hypergeometric distribution; in~(ii) we noticed that all factors in the product are equal to one except for that corresponding to $x=0$; in~(iii) we observed that $n \geq nt_n(0) - rw_r(0) + r$, because
\bb
n - \left(nt_n(0) - rw_r(0) + r\right) &= n(1-t_n(0)) - r (1-w_r(0)) = \sum_{x\neq 0} \left(nt_n(x) - rw_r(x)\right) \geq 0
\ee
due to the fact that types are probability distribution and hence sum to one; in~(iv) we leveraged the fact that $\frac{a-j}{b-j} \leq \frac{a}{b}$ for all triples of integers $a,b,j$ such that $j\leq a\leq b$; and finally in~(v) we inserted the inequality $1-r/n \leq 1-\eta$, which follows from~\eqref{bounds_ratio_r_n}. 

Summing over all types $w_r$ that appear in the expression~\eqref{d_r(t_n)}, we obtain that
\bb
d_r(t_n) &\leq \sum_{\mathcal{T}_r \ni w_r \preceq \frac{n}{r} t_n} \!\!\! H_{n,t_n;\, r}(w_r) \left(1-\eta\right)^{n - nt_n(0) + rw_r(0) - r} \\
&= \sum_{\mathcal{T}_r \ni w_r \preceq \frac{n}{r} t_n} \!\!\! H_{n,t_n;\, r}(w_r) \left(1-\eta\right)^{\sum_{x\neq 0} (nt_n(x) - rw_r(x))}\! .
\label{d_r_bound_proof_eq2} 
\ee
Now, since the sum in~\eqref{d_r_bound_proof_eq2} features only types $w_r\in \mathcal{T}_r$ such that $rw_r(x) \leq nt_n(x)$ for all $x\in \XX$, we have
\begin{align}
d_r(t_n) \nonumber &\leq\ \sum_{\mathcal{T}_r \ni w_r \preceq \frac{n}{r} t_n} H_{n,t_n;\, r}(w_r) \left(1-\eta \right)^{\max_{x\neq 0} (nt_n(x) - rw_r(x))} \nonumber \\
&\leq\ \sum_{\mathcal{T}_r \ni w_r \preceq \frac{n}{r} t_n} H_{n,t_n;\, r}(w_r) \left(1-\eta\right)^{nt_n(1) - rw_r(1)} \nonumber \\
&\eqt{(vi)}\ \sum_{k=0}^{nt_n(1)} H\big(n, nt_n(1);\, r,k\big) \left(1-\eta\right)^{nt_n(1) - k} \nonumber \\
&\ \eqt{(vii)}\ \sum_{k=0}^{nt_n(1)} H\big(n, r;\, nt_n(1),k\big) \left(1-\eta\right)^{nt_n(1) - k} \nonumber \\
&\leqt{(viii)}\ 2\, e^{-2n t_n(1)\, u^2} + \sum_{\substack{k\in \{0,\ldots, r\}:\\ \big|\frac{k}{nt_n(1)} - \frac{r}{n}\big| \leq u}} \!\!\! H\big(n, r;\, nt_n(1),k\big) \left(1-\eta\right)^{nt_n(1) - k} \label{d_r_bound_proof_eq4} \\
&\leqt{(ix)}\ 2\, e^{- 2nt_n(1)\, u^2} + \left(1-\eta\right)^{nt_n(1) \left(1-\frac{r}{n} - u\right)} \nonumber \\
&\leqt{(x)}\ 2\, e^{- nt_n(1)\, \eta^2/2} + \left(1-\eta\right)^{n t_n(1) \eta/2} \nonumber \\
&\leqt{(xi)}\ 3\, e^{- nt_n(1)\, \eta^2/2} \nonumber \\
&\leq\ 3\, e^{- N \eta^2/2}\, . \nonumber
\end{align}
Here: in~(vi) we observed that if $w_r\sim H_{n,t_n;\, r}$ then the random variable $rw_r(1)$ is distributed as a univariate hypergeometric distribution, i.e.\ $rw_r(1) \sim H(n, nt_n(1);\, r,\cdot)$; in~(vii) we employed the duality relation~\eqref{duality_hypergeometric}; in~(viii) we used the concentration bound~\eqref{first_tail_bound_hypergeometric}, for some $u>0$ to be fixed shortly; to deduce~(ix), we simply wrote
\bb
nt_n(1) - k = n t_n(1) \left(1-\frac{k}{n t_n(1)} \right) \geq nt_n(1) \left(1 - \frac{r}{n} - u \right) ;
\ee
continuing, in~(x) we fixed $u=\eta/2$, so that $1-\frac{r}{n} - u \geq \eta - \eta/2 = \eta/2$; finally, in~(xi) we noticed that $(1-\eta)^\eta = e^{\eta \ln(1-\eta)} \leq e^{-\eta^2}$.
\end{proof}

\subsection{Controlling the norm at the output of the blurring map} \label{subsec_control_norm}

Here we will show how to give some precise estimates on the trace norm at the output of the quantum blurring map when the input operator has a large zero sub-matrix in the top left corner. This is the goal of Proposition~\ref{blurring_deficient_operators_prop}. This result will be key to proving that a suitably `second-quantised' version of the quantum blurring map, which acts on an infinite-dimensional bosonic system, is continuous with respect to the \emph{weak* topology} on the input space (Corollary~\ref{damping_weakly_vanishing_seq_cor}).

Before we can state and prove Proposition~\ref{blurring_deficient_operators_prop}, we need the following general lemma. The setting involves a Hermitian operator $T$ that is assumed to be a (non-strict) contraction, i.e.\ $\|T\|_\infty = 1$; if $T$ behaves instead like a \emph{strict} contraction on the orthogonal complement to one of its invariant subspaces $V$, then the lemma tells us that it also acts as a strict contraction when acting from the left and from the right on any other operator $Z$ with zero restriction to $V$. This explains the name: the `tail' in question is the difference between $Z$ and its restriction to $V$ (which might involve off-diagonal terms between $V$ and $V^\perp$).

\begin{lemma}[(Tail contraction)] \label{tail_contraction_lemma}
Let $T = T^\dag$ be an $N\times N$ Hermitian matrix, and let $V \subseteq \C^N$ be one of its invariant subspaces, so that $TV\subseteq V$. Assume that $\|T\|_\infty = 1$ and that
\bb
\big\| (\id - P_V) T\big\|_\infty \leq \mu < 1\, ,
\ee
where $P_V$ is the orthogonal projector onto $V$. If an $N\times N$ matrix $Z$ satisfies that $P_V Z P_V = 0$, then
\bb
\|TZT\|_1 \leq \left(1 - (1-\mu)^2\right) \|Z\|_1 \leq 2\mu \|Z\|_1\, .
\ee
\end{lemma}

\begin{proof}
Since $V$ is an invariant subspace for $T$, we have $[T,P_V]=0$. Hence, for an arbitrary real number $x\in \R$ we can write
\begin{align}
\|TZT\|_1 &= \|TZT - x^2 T P_V Z P_V T\|_1 \nonumber \\
&=  \|TZT - x^2 P_V TZT P_V \|_1 \nonumber \\
&= \frac12 \,\Big\| (\id - x P_V) T Z T (\id + x P_V) + (\id + x P_V) TZT (\id - x P_V) \Big\|_1 \label{tail_contraction_lemma_proof_eq1} \\
&\leq \big\| (\id - x P_V)T \big\|_\infty \big\|(\id + x P_V)T\big\|_\infty \|Z\|_1 \nonumber \\
&= \max\Big\{ (1\!-\!x)\, \big\|P_V T\big\|_\infty,\, \big\| (\id\!-\!P_V)T\big\|_\infty \Big\} \max\Big\{ (1\!+\!x)\, \big\|P_V T\big\|_\infty,\, \big\| (\id\!-\!P_V)T\big\|_\infty \Big\}\, \|Z\|_1 . \nonumber 
\end{align}
Note that
\bb
1 &= \|T\|_\infty \\
&= \max \big|\spec\,(T)\big| \\
&= \max \big|\spec\big(P_VT\big) \cup \spec\big( (\id-P_V) T\big)\big| \\
&= \max\big\{ \big|\spec\big(P_VT\big)\big|,\, \big|\spec\big( (\id-P_V) T\big)\big| \big\} .
\ee
Since $\max \big|\spec\big( (\id-P_V) T\big)\big| = \big\| (\id - P_V) T\big\|_\infty \leq \mu < 1$, it follows that $\max\big|\spec(P_VT)\big| = \|P_V T\|_\infty = 1$. Thus, the relation~\eqref{tail_contraction_lemma_proof_eq1} can be rewritten as
\bb
\|TZT\|_1 \leq \max\big\{1-x,\, \big\| (\id-P_V)T\big\|_\infty\big\} \max\big\{1+x,\, \big\| (\id-P_V)T\big\|_\infty\big\}\, \|Z\|_1\, .
\ee
Setting $x \coloneqq 1-\mu$ yields the claim.
\end{proof}

\begin{prop} \label{blurring_deficient_operators_prop}
For two integers $n\geq N\geq 1$, let $X$ be an operator on $\mathrm{Sym}^n\big(\C^d\big)$ such that $\braket{n,t_n|X|n,s_n} = 0$ whenever $t_n, s_n\in \mathcal{T}_n$ satisfy that $\max_{x\neq 0} nt_n(x) \leq N$ and $\max_{x\neq 0} ns_n(x) \leq N$. Then, for all $\delta\in (0,\frac12]$ it holds that
\bb
\big\| B_{n,\delta}(X) \big\|_1 \leq 2\left( e^{-\floor{\delta n}/18} + \sqrt3\, e^{-N \delta^2/16} \right) \|X\|_1\, ,
\label{blurring_deficient_operators}
\ee
where the blurring map is defined by~\eqref{blurring}.
\end{prop}

\begin{proof}
We start by recalling the tail bound~\eqref{second_tail_bound_hypergeometric}, which becomes in this case
\bb
\sum_{\substack{r\in \{0,\ldots,n\}:\\ \left|\frac{r}{n} - \frac{\floor{\delta n}}{n+\floor{\delta n}}\right|\, \geq\, u_n}} \!\!\!\!\! H\left(n+\floor{\delta n},\floor{\delta n};\, n,r\right) \leq 2\,e^{-2n^2u_n^2/\floor{\delta n}} 
\ee
for all $u_n>0$. To simplify our expressions, we can set 
\bb
u_n \coloneqq \frac{\floor{\delta n}}{n+\floor{\delta n}} - \frac{\floor{\delta n}}{2n} \geq \frac{\floor{\delta n}}{6n}\, , 
\ee
where the last inequality follows from our assumption that $0<\delta\leq 1/2$; since for this choice 
\bb
\frac{\floor{\delta n}}{n+\floor{\delta n}} - u_n = \frac{\floor{\delta n}}{2n} \geq \frac1n \floor{\frac{\delta n}{2}}
\ee
and
\bb
\frac{\floor{\delta n}}{n+\floor{\delta n}} + u_n \leq \frac32 \delta\,, 
\ee
we have
\bb
\sum_{r = 0}^{\floor{\delta n/2}} H(n+\floor{\delta n}, \floor{\delta n};\, n,r) + \sum_{r = \ceil{3\delta n/2}}^n H(n+\floor{\delta n},\floor{\delta n};\, n,r) \leq 2\, e^{-\floor{\delta n}/18} .
\ee
Therefore, using the expression in Lemma~\ref{blurring_Gammas_lemma} and remembering that $\Gamma_{n,r}$ is a sub-channel and hence a trace-norm-contractive map (see~\eqref{contractivity_trace_norm}), we have
\bb
\left\| B_{n,\delta}(X) - \sum_{r=\floor{\delta n/2}+1}^{\ceil{3\delta n/2}-1} H(n+\floor{\delta n},\floor{\delta n};\, n,r)\ \Gamma_{n,r}(X) \right\|_1 \leq 2\,e^{-\floor{\delta n}/18}\, \|X\|_1
\label{blurring_truncated}
\ee
for all input operators $X$. Due to~\eqref{Gamma_in_terms_of_Theta} (see also Corollary~\ref{blurring_Theta_cor}), from this we deduce that
\bb
\left\|B_{n,\delta}(X)\right\|_1 &\leq 2\,e^{-\floor{\delta n}/18} \|X\|_1 + \sum_{r=\floor{\delta n/2}+1}^{\ceil{3\delta n/2}-1} \!\!\!\! H(n+\!\floor{\delta n},\floor{\delta n};\, n,r) \left\|\Theta_{n,r}\Big( \!D_r^{1/2} X D_r^{1/2}\! \Big)\right\|_1 \\[4pt]
&\leq 2\, e^{-\floor{\delta n}/18} \|X\|_1 + \sum_{r=\floor{\delta n/2}+1}^{\ceil{3\delta n/2}-1} \!\!\!\! H(n+\floor{\delta n},\floor{\delta n};\, n,r) \left\| D_r^{1/2} X D_r^{1/2} \right\|_1\, ,
\label{blurring_deficient_operators_proof_eq5}
\ee
where in the last inequality we used once again the contractivity of the trace norm under quantum channels (Eq.~\eqref{contractivity_trace_norm}). Next, for some $r$ in the above sum we will estimate $\big\| D_r^{1/2} X D_r^{1/2} \big\|_1$ by means of Lemma~\ref{tail_contraction_lemma}. To this end, we make the substitutions $T=D_r^{1/2}$ and $V = \Span\{\ket{n,t_n}\!:\, t_n\!\in\! \mathcal{T}_n,\, \max_{x\neq 0} nt_n(x) \!\leq\! N\}$. Note that $V$ is $D_r$-invariant because it is spanned by some eigenvectors of $D_r$, and $P_V X P_V =0$ by assumption. Furthermore,
\bb
\big\|D_r^{1/2}\big\|_\infty = \max_{t_n \in \mathcal{T}_n} \sqrt{d_r(t_n)} = 1\, ,
\ee
simply because $d_r(t_n) \in [0,1]$ for all $t_n\in\mathcal{T}_n$, as one deduces immediately from~\eqref{d_r_bound_proof_eq2}, and moreover $d_r(e_0) = 1$ by inspection, where $e_0$ is as usual the type concentrated on $x=0$. Finally, we estimate
\bb
\big\| (\id-P_V) D_r^{1/2}\big\|_\infty = \max_{t_n\in \mathcal{T}_n,\ \max_{x\neq 0} nt_n(x) > N} \sqrt{d_r(t_n)} \leq \sqrt{3}\, e^{- N \delta^2 /16}
\ee
using Lemma~\ref{d_r_bound_lemma} with $\eta=\delta/2$; note that~\eqref{bounds_ratio_r_n} is satisfied because
\bb
\frac{\delta}{2} \leq \frac{\floor{\delta n/2} \!+\! 1}{n} \leq \frac{r}{n} \leq \frac{\ceil{3\delta n/2} \!-\! 1}{n} \leq \frac{3\delta}{2} \leq 1 -\frac{\delta}{2}\, .
\ee
Setting $\mu = \sqrt{3}\, e^{-N \delta^2/16}$, Lemma~\ref{tail_contraction_lemma} applied on the right-hand side of~\eqref{blurring_deficient_operators_proof_eq5} gives us
\bb
\left\|B_{n,\delta}(X)\right\|_1 &\leq 2\,e^{-\floor{\delta n}/18} \|X\|_1 + 2\sqrt3\, e^{-N \delta^2/16} \|X\|_1\, ,
\ee
which completes the proof.
\end{proof}

It is useful to define also a modified version of the blurring map $B_{n,\delta}$, in which we embed $\mathrm{Sym}^n\big(\C^d\big)$ into the Fock space $\HH_{d-1} = \big(\ell^2(\N)\big)^{\otimes (d-1)}$ (see~\eqref{H_m}). Note that $\HH_{d-1}$ is spanned by the number vectors $\ket{k_1,\ldots,k_{d-1}} = \ket{k}$, where $k_i\in \N$ and $k\in \N^{d-1}$, which can be connected to the basis $\{\ket{n,t_n}\}_{t_n\in \mathcal{T}_n}$ given by~\eqref{ket_n_k} via the isometry $U_n: \mathrm{Sym}^n\big(\C^d\big)\to \HH_{d-1}$ defined by
\bb
U_n\ket{n,t_n} = \ket{nt_n(1),\, nt_n(2),\, \ldots, n t_n(d-1)} \quad \forall\, t_n \in \mathcal{T}_n\, .
\label{embedding}
\ee
We can then construct the `lifted' blurring map $\widetilde{B}_{n,\delta}:\T(\HH_{d-1}) \to \T(\HH_{d-1})$ defined by
\bb
\widetilde{B}_{n,\delta}(X) \coloneqq U_n^{\vphantom{\dag}} B_{n,\delta}\big(U_n^\dag X U_n^{\vphantom{\dag}} \big) U_n^\dag .
\label{lifted_blurring}
\ee
Note that $\widetilde{B}_{n,\delta}$, just like $B_{n,\delta}$, is a sub-channel. 

Now, Proposition~\ref{blurring_deficient_operators_prop} tells us that $B_{n,\delta}$ is highly contractive on operators $X$ that have zero restriction to the space of types with most of the weight concentrated on $x=0$. Upon lifting, this translates to the statement that $\widetilde{B}_{n,\delta}$ is highly contractive on operators $X$ with zero restriction to the space of vectors with bounded Fock number. A direct consequence is that if a sequence of operators vanishes in the limit with respect to the weak* topology, the lifted blurring map $\widetilde{B}_{n,\delta}$ will be increasingly trace norm contractive on that sequence. This is formalised as follows.

\begin{cor} \label{damping_weakly_vanishing_seq_cor}
Let $(X_n)_n$ be a sequence of trace class operators $X_n\in \T(\HH_{d-1})$ on the Hilbert space $\HH_{d-1}$ given by~\eqref{H_m}. Let $\delta \in (0,\frac12]$. If $\|X_n\|_1\leq C$ for some constant $C$ and all $n$, and moreover $X_n \tendsn{w^*} 0$ (see Section~\ref{subsubsec_operator_topologies}), then
\bb
\lim_{n\to\infty} \big\|\widetilde{B}_{n,\delta}(X_n) \big\|_1 = 0\, ,
\label{damping_weakly_vanishing_seq}
\ee
where $\widetilde{B}_{n,\delta}$ is defined by~\eqref{lifted_blurring}.
\end{cor}

\begin{proof}
Fix some $N \in \N^+$, and consider the projector
\bb
P_N\coloneqq \sum_{k_1,\ldots, k_{d-1}\in \{0,\ldots, N\}}\ketbra{k_1,\ldots, k_{d-1}} = \sum_{k\in \{0,\ldots,N\}^{d-1}} \ketbra{k}
\ee
onto the subspace of $\HH_{d-1}$ spanned by Fock states with Fock number bounded by $N$ on each mode. Since the support of $P_N$ is finite dimensional and $X_n \tendsn{w^*} 0$, it follows that
\bb
\big\|P_N X_n P_N\big\|_1 \tendsn{} 0\, . 
\label{damping_weakly_vanishing_seq_proof_eq1}
\ee
Therefore, for $n\geq N$,
\bb
\big\|\widetilde{B}_{n,\delta}(X_n) \big\|_1 &\leq \big\|\widetilde{B}_{n,\delta}\big(P_N X_n P_N\big) \big\|_1 + \big\|\widetilde{B}_{n,\delta}\big(X_n - P_N X_n P_N \big) \big\|_1 \\
&\leqt{(i)} \big\|P_N X_n P_N \big\|_1 + \big\|\widetilde{B}_{n,\delta}\big(X_n - P_N X_n P_N \big) \big\|_1 \\
&= \big\|P_N X_n P_N \big\|_1 + \big\| B_{n,\delta}\big(U_n^{\dag}(X_n - P_N X_n P_N) U_n^{\vphantom{\dag}} \big) \big\|_1 \\
&\leqt{(ii)} \big\|P_N X_n P_N \big\|_1 + 2\left( e^{-\floor{\delta n}/18} + \sqrt3\, e^{-N \delta^2/16} \right) \big\|X_n - P_N X_n P_N\big\|_1 \\
&\leqt{(iii)} \big\|P_N X_n P_N \big\|_1 + 4C\left( e^{-\floor{\delta n}/18} + \sqrt3\, e^{-N \delta^2/16} \right) .
\label{damping_weakly_vanishing_seq_proof_eq2}
\ee
Here: (i)~follows from the fact that $\widetilde{B}_{n,\delta}$ is a sub-channel; in~(ii) we used Proposition~\ref{blurring_deficient_operators_prop}, which is applicable because
\bb
\bra{n,t_n} U_n^{\dag}\left(X_n - P_N X_n P_N\right) U_n^{\vphantom{\dag}} \ket{n,s_n} = 0
\ee
for all $t_n,s_n\in \mathcal{T}_n$ such that $\max_{x\neq 0} nt_n(x) \leq N$ and $\max_{x\neq 0} ns_n(x)\leq N$; finally, in~(iii) we simply observed that 
\bb
\|X_n - P_NX_nP_N\|_1 \leq \|X_n\|_1 + \|P_N X_n P_N\|_1 \leq 2\|X_n\|_1 \leq 2C
\ee
by the triangle inequality. Taking the limit $n\to\infty$ in~\eqref{damping_weakly_vanishing_seq_proof_eq2} and using~\eqref{damping_weakly_vanishing_seq_proof_eq1} yields
\bb
\limsup_{n\to\infty} \big\|\widetilde{B}_{n,\delta}(X_n) \big\|_1 \leq 4\sqrt3\, C\, e^{-N \delta^2/16}\, .
\ee
Since $N\in \N^+$ was arbitrary, we can now take the limit $N\to\infty$, which completes the proof.
\end{proof}

\subsection{Convergence in second quantisation} \label{subsec_convergence}

The newly defined lifted blurring maps act on a continuous-variable system composed of $d-1$ bosonic modes. Since all maps act on the same space, we can now ask ourselves whether they converge to anything. We will see that this is indeed the case, and that the limiting map is a composition of well-known bosonic (sub-)channels (Corollary~\ref{lifted_blurring_strong_convergence_cor}). 

To arrive at that result, we need to prove an elementary lemma first.

\begin{lemma} \label{uniform_convergence_lemma}
Let $(v_n)_n$ be a sequence of positive numbers $v_n > \frac{1}{2n}$ such that $\lim_{n\to\infty} v_n = 0$, and let $\theta >0$ be another (fixed) positive number. Then, for all $k\in \N$,
\bb
\lim_{n\to\infty} \max_{N\in \N:\, \left|\frac{N}{n} - \theta\right|\, \leq\, v_n} \left| \frac{1}{n^k} \binom{N}{k} - \frac{\theta^k}{k!}\right| = 0\, .
\label{uniform_convergence}
\ee
\end{lemma}

\begin{proof}
Note that the condition $v_n > \frac{1}{2n}$ ensures that the set of which we are taking the maximum in~\eqref{uniform_convergence} is not empty. The claim follows immediately from the inequalities $\frac{(N-k)^k}{k!} \leq \binom{N}{k} \leq \frac{N^k}{k!}$.
\end{proof}

Before we state and prove the key lemma of this section, we need to introduce some notation. For some $p\in \N^{\XX\setminus \{0\}}$ and some integer $n\geq \sum_{x\neq 0} p(x)$, form the type $t_{n,p}\in \mathcal{T}_n$ defined by
\bb
t_{n,p}(x) \coloneqq \left\{ \begin{array}{ll} 1 - \frac1n \sum_{x\neq 0} p(x) & \text{ if $x=0$,} \\[1ex] p(x)/n & \text{ if $x\neq 0$,} \end{array} \right.
\label{type_t_n_p}
\ee
so that
\bb
U_n \ket{n, t_{n,p}} = \ket{p}\, , \qquad U_n^\dag \ket{p} = \ket{n,t_{n,p}}\, ,
\label{action_U_on_type_t_n_p}
\ee
where $U_n$ is the embedding defined by~\eqref{embedding}. 

\begin{lemma} \label{blurring_convergence_lemma}
Let $h,k,h',k'\in \N^{\XX\setminus \{0\}}$ be integer vectors of length $d-1$. Then, with the notation introduced in~\eqref{type_t_n_p}, for all $\delta\in (0,\frac12]$ it holds that
\bb
\lim_{n\to\infty} \bra{n,t_{n,h'}} B_{n,\delta}\big(\ketbraa{n,t_{n,h}}{n,t_{n,k}}\big) \ket{n,t_{n,k'}} &= \!\left\{ \begin{array}{ll} \!\!\prod\limits_{x\neq 0} \!\!\sqrt{\binom{h(x)}{\ell(x)}\!\binom{k(x)}{\ell(x)}} \frac{\delta^{\ell(x)}}{(1+\delta)^{(h+k-\ell)(x)}} & \text{if $h\!-\!h'=k\!-\!k'=\ell \geq 0$,} \\[4ex] 0 & \text{otherwise.} \end{array} \right.
\label{blurring_entrywise_convergence}
\ee
In particular, the lifted blurring map~\eqref{lifted_blurring} satisfies that
\bb
\lim_{n\to\infty} \widetilde{B}_{n,\delta}\big(\ketbraa{h}{k}\big) &= \sum_{\ell\in \N^{\XX\setminus \{0\}}} \prod_{x\neq 0} \sqrt{\binom{h(x)}{\ell(x)}\!\binom{k(x)}{\ell(x)}}\ \frac{\delta^{\ell(x)}}{(1+\delta)^{(h+k-\ell)(x)}} \ketbraa{h-\ell}{k-\ell}\, ,
\label{blurring_limit}
\ee
where the convergence is in trace norm, and the sum on the right-hand side has only finitely many non-zero terms.\footnote{Terms where $h-\ell$ has some negative entry, and hence, strictly speaking, $\ket{h-\ell}$ would not be well defined, do not appear, because in that case $\prod_{x\neq 0} \binom{h(x)}{\ell(x)} =0$.}
\end{lemma}

\begin{proof}
We start with~\eqref{blurring_entrywise_convergence}. From Lemma~\ref{Gamma_Kraus_lemma} it is clear that the statement holds when $h - h' \neq k - k'$, or $h - h' = k - k' \eqqcolon \ell$ but $\ell(x) < 0$ for some $x\neq 0$, because the sequence on the left-hand side of~\eqref{blurring_entrywise_convergence} is identically zero for all $n$, and similarly the right-hand side is also zero. The latter claim is obvious, so let us justify the former one. Using Lemmas~\ref{blurring_Gammas_lemma} and~\ref{Gamma_Kraus_lemma}, it is clear that the left-hand side of~\eqref{blurring_entrywise_convergence} vanishes unless there is an $r\in \{0,1,\ldots,\floor{\delta n}\}$ and some $w_r\in \mathcal{T}_r$ such that $\braket{n,t_{n,h'}|M_{r,w_r}|n,t_{n,h}}\neq 0$ and similarly $\braket{n,t_{n,k'}|M_{r,w_r}|n,t_{n,k}}\neq 0$. Defining the vector $\ell\in \N^{\XX\setminus \{0\}}$ by $\ell(x) = rw_r(x)$ for all $x\neq 0$, this implies immediately that $h\succeq \ell$ and $k\succeq \ell$, while at the same time $t_{n,h'} = t_{n,h} - \frac{r}{n} w_r + \frac{r}{n} e_0$ and $t_{n,k'} = t_{n,k} - \frac{r}{n} w_r + \frac{r}{n} e_0$. These latter two identities can be recast as $h' = h - \ell$ and $k' = k - \ell$, entailing that $h - h' = k - k' = \ell \geq 0$.

We will thus assume that $h - h' = k - k' \eqqcolon \ell \in \N^{\XX\setminus\{0\}}$. Using the tail bound~\eqref{first_tail_bound_hypergeometric} together with the expression~\eqref{blurring_Gammas} and the fact that each $\Gamma_{n,r}$ is a sub-channel, it is not difficult to see that there is a sequence of positive numbers $(v_n)_n$ such that $v_n\tendsn{} 0$ and
\bb
\Bigg\|\, B_{n,\delta}\big(\ketbraa{n,t_{n,h}}{n,t_{n,k}}\big) - \!\!\!\sum_{\substack{r\in \{0,\ldots,n\},\\ |\frac{r}{n} - \frac{\delta}{1+\delta}| \,\leq\, v_n}} H(n+\!\floor{\delta n},\floor{\delta n};\, n,r)\ \Gamma_{\!n,r}\big(\!\ketbraa{n,t_{n,h}}{n,t_{n,k}}\!\big)\, \Bigg\|_1 \tendsn{} 0\, .
\ee
To this end, it suffices to take $u = v_n \sim n^{-1/3}$ in~\eqref{first_tail_bound_hypergeometric}, so that the right-hand side vanishes as $n\to\infty$. 
(Note also that $\big| \frac{\floor{\delta n}}{n+\floor{\delta n}} - \frac{\delta}{1+\delta}\big| \leq \frac{1}{n-1}$.) Cf.~\eqref{blurring_truncated}, where we considered instead an asymptotically constant $u$. Thus,
\bb
&\lim_{n\to\infty} \bra{n, t_{n,h'}} B_{n,\delta}\big(\ketbraa{n,t_{n,h}}{n,t_{n,k}}\big) \ket{n,t_{n,k'}} \\
&\quad = \lim_{n\to\infty} \bra{n, t_{n,h-\ell}} B_{n,\delta}\big(\ketbraa{n,t_{n,h}}{n,t_{n,k}}\big) \ket{n,t_{n,k-\ell}} \\[4pt]
&\quad = \lim_{n\to\infty} \sum_{\substack{r\in \{0,\ldots,n\},\\ |\frac{r}{n} - \frac{\delta}{1+\delta}| \,\leq\, v_n}} H(n+\floor{\delta n},\floor{\delta n};\, n,r) \bra{n, t_{n,h-\ell}} \Gamma_{\!n,r}\big(\!\ketbraa{n,t_{n,h}}{n,t_{n,k}}\!\big)\! \ket{n,t_{n,k-\ell}} \\[4pt]
&\quad = \lim_{n\to\infty} \sum_{\substack{r\in \{0,\ldots,n\},\\ |\frac{r}{n} - \frac{\delta}{1+\delta}| \,\leq\, v_n}} H(n+\floor{\delta n},\floor{\delta n};\, n,r) \frac{\binom{n-r}{n t_{n,h} - rw_r} \binom{n-r}{nt_{n,k}-rw_r} \binom{r}{rw_r}}{\sqrt{\binom{n}{nt_{n,h} - rw_r + re_0}\! \binom{n}{nt_{n,k}- rw_r +r e_0}\!\binom{n}{nt_{n,h}}\!\binom{n}{nt_{n,k}}}}\, , \label{blurring_convergence_proof_eq2}
\ee
where we introduced a type $w_r\in \mathcal{T}_r$ defined by
\bb
rw_r(x) = \left\{ \begin{array}{ll} r - \sum_{x\neq 0} \ell(x) & \text{ if $x=0$,} \\[1ex] \ell(x) & \text{ if $x\neq 0$,} \end{array} \right.
\ee
so that $t_{n,h} - \frac{r}{n} w_r + \frac{r}{n} e_0 = t_{n,h'}$, and similarly $t_{n,k} - \frac{r}{n} w_r + \frac{r}{n} e_0 = t_{n,k'}$. (Note that $w_r$ is a type for sufficiently large $n$, because $r\to\infty$ when $n\to\infty$.)
Now, consider an arbitrary $n$ and some $r\in \{0,\ldots,n\}$ such that 
\bb
\left|\frac{r}{n} - \frac{\delta}{1+\delta}\right| \leq v_n\, .
\label{r_sublinear_range}
\ee
Writing $a_n \sim b_n$ when $a_n/b_n \tendsn{} 1$, for any $p\in \N^{\XX\setminus \{0\}}$ we can write the estimate
\bb
\binom{n}{n t_{n,p}} = \frac{n!}{\left(n-\sum_{x\neq 0} p(x)\right)! \prod_{x\neq 0} p(x)!} \sim \frac{n^{\sum_{x\neq 0} p(x)}}{\prod_{x\neq 0} p(x)!}\, .
\ee
Combining these estimates repeatedly, we obtain
\bb
\frac{\binom{n-r}{nt_{n,h}-rw_r} \binom{n-r}{nt_{n,k}-rw_r} \binom{r}{rw_r}}{\sqrt{\binom{n}{nt_{n,h} - rw_r + re_0} \binom{n}{nt_{n,k}- rw_r +r e_0}\binom{n}{nt_{n,h}}\binom{n}{nt_{n,k}}}} &\sim \frac{\frac{(n-r)^{\sum_{x\neq 0} (h'(x)+k'(x))}}{\prod_{x\neq 0} h'(x)!\, k'(x)!}\, \frac{r^{\sum_{x\neq 0} \ell(x)}}{\prod_{x\neq 0} \ell(x)!}}{\sqrt{\frac{n^{\sum_{x\neq 0} (h'(x)+k'(x))}}{\prod_{x\neq 0} h'(x)!\, k'(x)!}\, \frac{n^{\sum_{x\neq 0} (h(x)+k(x))}}{\prod_{x\neq 0} h(x)!\, k(x)!}}} \\[1ex]
&= \prod_{x\neq 0} \sqrt{\binom{h(x)}{\ell(x)}\binom{k(x)}{\ell(x)}} \left(1-\frac{r}{n}\right)^{h'(x) + k'(x)} \left(\frac{r}{n}\right)^{\ell(x)} \\[1ex]
&\tendsn{} \prod_{x\neq 0} \sqrt{\binom{h(x)}{\ell(x)}\binom{k(x)}{\ell(x)}}\ \frac{\delta^{\ell(x)}}{(1+\delta)^{h(x)+k(x)-\ell(x)}}\, .
\label{blurring_convergence_proof_eq6}
\ee
Using Lemma~\ref{uniform_convergence_lemma} repeatedly, it is not difficult to show that the above convergence happens \emph{uniformly} in $r$, provided that $r$ satisfies~\eqref{r_sublinear_range} and that $h,k,\ell$ are kept fixed. We can therefore safely sum~\eqref{blurring_convergence_proof_eq6} over all $r$ appearing in the sum on the rightmost side of~\eqref{blurring_convergence_proof_eq2}, obtaining precisely~\eqref{blurring_entrywise_convergence} once one remembers that
\bb
\sum_{\substack{r\in \{0,\ldots,n\},\\ |\frac{r}{n} - \frac{\delta}{1+\delta}| \,\leq\, v_n}} H(n+\floor{\delta n},\floor{\delta n};\, n,r) \tendsn{} 1
\ee
due, once again, to the concentration bound~\eqref{first_tail_bound_hypergeometric}. 

To obtain~\eqref{blurring_limit}, consider $h,k,h',k'\in \N^{\XX\setminus\{0\}}$. Then
\bb
\lim_{n\to\infty} \bra{h'} \widetilde{B}_{n,\delta}\big(\ketbraa{h}{k}\big) \ket{k'} &= \lim_{n\to\infty} \bra{h'} U_n B_{n,\delta}\big(U_n^\dag \ketbraa{h}{k} U_n^{\vphantom{\dag}}\big) U_n^{\dag} \ket{k'} \\
&= \lim_{n\to\infty} \bra{n, t_{n,h'}} B_{n,\delta}\big( \ketbraa{n, t_{n,h}}{n, t_{n,k}} \big) \ket{n, t_{n,k'}} \\
&= \!\left\{ \begin{array}{ll} \prod\limits_{x\neq 0} \sqrt{\binom{h(x)}{\ell(x)} \binom{k(x)}{\ell(x)}} \frac{\delta^{\ell(x)}}{(1+\delta)^{(h+k-\ell)(x)}} & \text{ if $h\!-\!h'\!\!=\!k\!-\!k'\!=\ell \geq\! 0$,} \\[4ex] 0 & \text{ otherwise.} \end{array} \right.
\ee
where in the second equality we employed~\eqref{action_U_on_type_t_n_p}. Hence, since $\widetilde{B}_{n,\delta}\big(\ketbraa{h}{k}\big)$ has only a finite expansion in the Fock basis, we can take the trace norm limit inside a finite sum and write
\bb
\lim_{n\to\infty} \widetilde{B}_{n,\delta}(\ketbraa{h}{k}) &= \lim_{n\to\infty} \sum_{\ell\in \N^{\XX\setminus\{0\}}} \bra{h\!-\!\ell} \widetilde{B}_{n,\delta}(\ketbraa{h}{k}) \ket{k\!-\!\ell} \ketbraa{h\!-\!\ell}{k\!-\!\ell} \\
&= \sum_{\ell\in \N^{\XX\setminus\{0\}}} \left(\lim_{n\to\infty} \bra{h\!-\!\ell} \widetilde{B}_{n,\delta}(\ketbraa{h}{k}) \ket{k\!-\!\ell} \right) \ketbraa{h\!-\!\ell}{k\!-\!\ell} \\
&= \hspace{0pt}\sum_{\ell\in \N^{\XX\setminus\! \{0\}}} \hspace{-2pt}\left( \prod_{x\neq 0} \sqrt{\binom{h(x)}{\ell(x)} \binom{k(x)}{\ell(x)}}\, \frac{\delta^{\ell(x)}}{(1+\delta)^{h(x)+k(x)-\ell(x)}} \right) \ketbraa{h\!-\!\ell}{k\!-\!\ell}\, ,
\ee
as claimed. This concludes the proof of~\eqref{blurring_limit} and of the lemma.
\end{proof}

Before we continue, we observe that the expression on the right-hand side of~\eqref{blurring_limit} is tantalisingly similar to the one appearing in the equation~\eqref{pure_loss_action_on_Fock} that expresses the action of a pure loss channel $\EE_\lambda : \T(\HH_1)\to \T(\HH_1)$ on Fock states. In fact, the limit of the lifted blurring map is rather easily connected to a tensor product of pure loss channels, each acting on one mode of the $(d-1)$-mode system with Hilbert space $\HH_{d-1} = \big(\ell^2(\N)\big)^{\otimes (d-1)}$. This connection is our second fundamental technical insight, and it will prove crucial in order to establish the quantum blurring lemma. To make it precise, we need to introduce another family of maps $\D_\mu:\TT(\HH_1)\to \TT(\HH_1)$, where $\mu\in [0,1]$. Set
\bb
\D_\mu(X) \coloneqq \mu^{a^\dag a} X \mu^{a^\dag a}\, ,
\label{D_mu_Kraus}
\ee
where $a^\dag a$ is the number operator, so that
\bb
\D_\mu (\ketbraa{h}{k}) = \mu^{h+k} \ketbraa{h}{k}
\label{D_mu}
\ee
for all $h,k\in \N$. We are now ready to prove the following.

\begin{cor} \label{lifted_blurring_strong_convergence_cor}
For some $\delta\in (0,\frac12]$, let $\widetilde{B}_{n,\delta}$ be the lifted blurring map defined by~\eqref{lifted_blurring} and~\eqref{blurring}. Then
\bb
\widetilde{B}_{n,\delta} &\tendsn{s} \big(\EE_{\lambda(\delta)} \circ \D_{\mu(\delta)}\big)^{\otimes (d-1)}
\label{lifted_blurring_strong_convergence}
\ee
with respect to the strong convergence topology (defined in Section~\ref{subsubsec_operator_topologies}), where
\bb
\lambda(\delta) &\coloneqq \frac{1}{1+\delta(1+\delta)}\, ,\quad \mu(\delta) \coloneqq \frac{\sqrt{1+\delta(1+\delta)}}{1+\delta}\, ,
\label{lambda_and_mu}
\ee
and each tensor factor in~\eqref{lifted_blurring_strong_convergence} acts on one copy of $\ell^2(\N)$ inside $\HH_{d-1} = \big(\ell^2(\N)\big)^{\otimes (d-1)}$.
\end{cor}

\begin{proof}
For all $h,k\in \N^{\XX\setminus\{0\}}$, we have
\bb
\big(\EE_{\lambda(\delta)} \circ \D_{\mu(\delta)} \big)^{\otimes (d-1)} (\ketbraa{h}{k}) = \bigotimes_{x\neq 0} \big(\EE_{\lambda(\delta)} \circ \D_{\mu(\delta)} \big)(\ketbraa{h(x)}{k(x)})\, .
\label{lifted_blurring_strong_convergence_proof_eq1}
\ee
Now, for each fixed $x\neq 0$, it holds that
\bb
&\big(\EE_{\lambda(\delta)} \circ \D_{\mu(\delta)} \big)(\ketbraa{h(x)}{k(x)}) \\
&\quad \eqt{(i)} \mu(\delta)^{h(x)+k(x)} \EE_{\lambda(\delta)} (\ketbraa{h(x)}{k(x)}) \\
&\quad \eqt{(ii)} \left(\mu(\delta) \sqrt{\lambda(\delta)}\right)^{h(x)+k(x)} \sum_{\ell=0}^{\infty} \sqrt{\binom{h(x)}{\ell} \binom{k(x)}{\ell}} \left(\frac{1}{\lambda(\delta)} - 1\right)^\ell \ketbraa{h(x)-\ell}{k(x)-\ell} \\
&\quad \eqt{(iii)} \frac{1}{(1+\delta)^{h(x)+k(x)}} \sum_{\ell=0}^{\infty} \sqrt{\binom{h(x)}{\ell}\binom{k(x)}{\ell}}\ \big(\delta(1+\delta)\big)^{\ell(x)} \ketbraa{h(x)-\ell}{k(x)-\ell}
\label{lifted_blurring_strong_convergence_proof_eq2}
\ee
Here, (i)~follows from~\eqref{D_mu}, in~(ii) we employed~\eqref{pure_loss_action_on_Fock}, and (iii)~descends from~\eqref{lambda_and_mu}. By plugging~\eqref{lifted_blurring_strong_convergence_proof_eq2} into~\eqref{lifted_blurring_strong_convergence_proof_eq1} and comparing with~\eqref{blurring_limit}, one sees that
\bb
\lim_{n\to\infty} \widetilde{B}_{n,\delta}(\ketbraa{h}{k}) = \big(\EE_{\lambda(\delta)} \circ \D_{\mu(\delta)}\big)^{\otimes (d-1)}(\ketbraa{h}{k})
\ee
with respect to the trace norm topology. This also implies that
\bb
\lim_{n\to\infty} \widetilde{B}_{n,\delta}(X) = \big(\EE_{\lambda(\delta)} \circ \D_{\mu(\delta)}\big)^{\otimes (d-1)}(X)
\label{lifted_blurring_strong_convergence_proof_eq4}
\ee
holds for any operator $X$ with a finite expansion in Fock basis. Since such operators are trace norm dense in the space of all trace class operators, and both $\widetilde{B}_{n,\delta}$ and $\big(\EE_{\lambda(\delta)} \circ \D_{\mu(\delta)}\big)^{\otimes (d-1)}$ are uniformly (in $n$) continuous with respect to the trace norm topology, as they are all sub-channels and hence contractive with respect to the trace norm, we conclude that~\eqref{lifted_blurring_strong_convergence_proof_eq4} actually holds for all trace class operators $X$. This concludes the proof.
\end{proof}

The choice of the strong convergence topology in Corollary~\ref{lifted_blurring_strong_convergence_cor} is necessary, at least in our proof. This is particularly clear if one looks at~\eqref{lifted_blurring_strong_convergence_proof_eq4}, which we established for all operators $X$ with a finite expansion in Fock basis. These operators contain a constant number of excitations when $n\to\infty$; since the excitations are so diluted in the limit, the blurring map scatters them \emph{independently}, and its action can therefore be approximated by that of a (suitably rescaled) pure loss channel. Mathematically, this is reflected in the convergence of the hypergeometric distribution to a binomial distribution, implicit in~\eqref{blurring_limit}.

\subsection{On the support of the output of certain bosonic channels} \label{subsec_verifying_support_condition}

\begin{lemma} \label{vacuum_in_support_lemma}
For $\lambda,\mu\in (0,1)$, let $\EE_\lambda$ and $\D_\mu$ be defined by~\eqref{pure_loss_Kraus} and~\eqref{D_mu}, respectively. Let $m\in \N^+$ be a positive integer. For all $\Delta \in (0,\frac12]$, it holds that
\bb
\ket{0}^{\otimes m} \in \supp \left( \int_0^\Delta \frac{\dd\delta}{\Delta}\ \big(\EE_{\lambda(\delta)}\circ D_{\mu(\delta)}\big)^{\otimes m}(\rho)\right) ,
\label{vacuum_in_support}
\ee
for al $\rho\in \D(\HH_m)$, where $\lambda(\delta)$ and $\mu(\delta)$ are given by~\eqref{lambda_and_mu}.
\end{lemma}

\begin{rem} \label{Bochner_measurability_rem}
The integral in~\eqref{vacuum_in_support} should be thought of as a Bochner integral in the Banach space of all trace class operators on $m$ bosonic modes. Note that since the functions $\lambda \mapsto \EE_\lambda^{\otimes m}(\rho)$ and $\mu\mapsto \D_\mu(\rho)$ are continuous with respect to the trace norm topology,\footnote{The first claim can be proved rather simply by writing down the action of the pure loss channel on quantum characteristic functions~\cite[Eq.~(4.7)]{Ivan2011} and then using the continuity of quantum characteristic functions~\cite[Theorem~5.4.1]{HOLEVO} together with~\cite[Lemma~4]{G-dilatable}. The second is even easier, and follows by simply expanding the input state in the Fock basis.} they are also Bochner measurable. Then, from the identity $\left\|\big(\EE_\lambda\circ \D_\mu\big)^{\otimes m}(\rho) \right\|_1 \leq 1$ it follows that the above function is Bochner integrable.
\end{rem}

\begin{proof}[Proof of Lemma~\ref{vacuum_in_support_lemma}]
Since any mixed state is lower bounded by a positive multiple of a pure state in the positive semi-definite order, and $A\geq B\geq 0$ implies that $\supp(B) \subseteq \supp(A)$, we can assume without loss of generality that $\rho = \ketbra{\psi} \eqqcolon \psi$ be pure. Proceeding by contradiction, assume that~\eqref{vacuum_in_support} is not satisfied, so that we can find a vector $\ket{\phi} \in \HH_m$ that is orthogonal to $\supp \left( \int_0^\Delta \frac{\dd\delta}{\Delta}\ \big(\EE_{\lambda(\delta)}\circ D_{\mu(\delta)}\big)^{\otimes m}(\rho) \right)$ but not to $\ket{0}^{\otimes n}$. We thus have that
\bb
0 &= \bra{\phi} \left( \int_0^\Delta \frac{\dd\delta}{\Delta}\ \big(\EE_{\lambda(\delta)}\circ D_{\mu(\delta)}\big)^{\otimes m}(\psi) \right) \ket{\phi} \\
&= \int_0^\Delta \frac{\dd\delta}{\Delta}\ \bra{\phi}\big(\EE_{\lambda(\delta)}\circ D_{\mu(\delta)}\big)^{\otimes m}(\psi) \ket{\phi}\, ,
\ee
implying, since the integrand on the right-hand side is non-negative, that
\bb
\bra{\phi}\big(\EE_{\lambda(\delta)}\circ D_{\mu(\delta)}\big)^{\otimes m}(\psi) \ket{\phi} = 0\qquad \forall\ \delta\in (0,\Delta)\, ,
\ee
because the function $\delta \mapsto \bra{\phi}\big(\EE_{\lambda(\delta)}\circ D_{\mu(\delta)}\big)^{\otimes m}(\psi) \ket{\phi}$ is continuous due to~\eqref{lambda_and_mu} and to Remark~\ref{Bochner_measurability_rem}. Using the Kraus representations in~\eqref{pure_loss_Kraus} and~\eqref{D_mu_Kraus} and substituting~\eqref{lambda_and_mu}, we deduce that
\bb
0 &= \bra{\phi} \bigotimes_{i=1}^m \left( a_i^{n_i} \big(\sqrt{\lambda(\delta)}\, \mu(\delta)\big)^{a_i^\dag a_i} \right) \ket{\psi} \\
&= \bra{\phi} \bigotimes_{i=1}^m \left( a_i^{n_i} \left(\tfrac{1}{1+\delta} \right)^{a_i^\dag a_i} \right) \ket{\psi}
\ee
for all $n \in \N^m$, and all $\delta\in (0,\Delta)$ (here, $n = (n_1,\ldots, n_m)^\intercal$). Re-parametrising $z\coloneqq 1/(1+\delta)$, this can be recast as
\bb
\bra{\phi} \bigotimes_{i=1}^m \left( a_i^{n_i} z^{a_i^\dag a_i} \right) \ket{\psi} = 0\qquad \forall\ n \in \N^m\!,\quad \forall\ z\in \left(\tfrac{1}{1+\Delta}, 1\right) .
\label{zeros_F_n}
\ee
Inserting the decomposition $z^{a^\dag a} = \sum_{k=0}^\infty z^k \ketbra{k}$ and using the identity $a^n \ket{k} = \sqrt{\frac{k!}{(k-n)!}}\, \ket{k-n}$, where the right-hand side is understood to be zero when $n > k$, we obtain that
\bb
\bra{\phi} \bigotimes_{i=1}^m \left( a_i^{n_i} z^{a_i^\dag a_i} \right) \ket{\psi} &= \sum_{k\in \N^m\!,\ k\succeq n} z^{\,k_1+\ldots+k_m} \sqrt{\prod_{i=1}^m \frac{k_i!}{(k_i - n_i)!}}\ \phi_{k-n}^* \psi_{k}^{\vphantom{*}} \\
&= \sum_{N= n_1+\ldots + n_m}^\infty z^{\,N} \sum_{\substack{k\in \N^m:\ k\succeq n, \\ k_1 +\ldots + k_m = N}} \sqrt{\prod_{i=1}^m \frac{k_i!}{(k_i - n_i)!}}\ \phi_{k-n}^* \psi_k^{\vphantom{*}} \\
&\eqqcolon F_n(z)\, ,
\label{F_n}
\ee
where $\phi_k \coloneqq \braket{k|\phi}$, similarly for $\ket{\psi}$, the relation $\succeq$ indicates entry-wise comparison, and in the last line we defined a family of holomorphic functions $F_n$ on the open unit disk $\{z\in \C:\, |z| <1\}$. Note that
\bb
\left| \sum_{\substack{k\in \N^m:\ k\succeq n, \\ k_1 +\ldots + k_m = N}} \sqrt{\prod_{i=1}^m \frac{k_i!}{(k_i - n_i)!}}\ \phi_{k-n}^* \psi_k^{\vphantom{*}} \right|\ &\leq \sum_{\substack{k\in \N^m:\ k\succeq n, \\ k_1 +\ldots + k_m = N}} \sqrt{\prod_{i=1}^m \frac{k_i!}{(k_i - n_i)!}}\ |\phi_{k-n}| |\psi_{k}| \\
&\leq N^{(n_1+\ldots+n_m)/2} \sum_{\substack{k\in \N^m:\ k\succeq n, \\ k_1 +\ldots + k_m = N}} |\phi_{k-n}| |\psi_k| \\
&\leq N^{(n_1+\ldots+n_m)/2} \sqrt{\sum_{k\in \N^m} |\phi_k|^2} \sqrt{\sum_{k\in \N^m} |\psi_{n+k}|^2} \\
&\leq N^{(n_1+\ldots+n_m)/2} \\
&= \mathrm{poly}_n\big(\sqrt{N}\big)\, ,
\ee
so that the power series that defines $F_n$ does indeed converge on the open unit disk.

From~\eqref{zeros_F_n} we see that the set of zeros of $F_n$ includes the whole interval $\left(\tfrac{1}{1+\Delta}, 1\right)$. Since $F_n$ is holomorphic on the open unit disk, the only possibility is that $F_n = 0$ identically. In particular, the first term of the power series that defines $F_n$ in~\eqref{F_n}, corresponding to $N=n_1 + \ldots + n_m$, must vanish, so that
\bb
\phi_0^* \psi_n^{\vphantom{*}} = 0 \qquad \forall\ n\in \N^m .
\ee
Since by assumption $\phi_0\neq 0$, it must be that $\psi_n = 0$ for all $n\in \N^m$, implying that $\ket{\psi}$ is the zero vector. We have reached a contradiction, and the proof is complete. 
\end{proof}

The following simple result is the last ingredient we need in order to prove the quantum blurring lemma (Lemma~\ref{quantum_blurring_lemma}).

\begin{lemma} \label{support_lemma}
Let $A\in \T(\HH)$, $A\geq 0$ be a positive semi-definite trace class operator on a separable Hilbert space $\HH$.  For some $\ket{\psi}\in \HH$, the following conditions are equivalent:
\begin{enumerate}[(a)]
\item $\ket{\psi}\in \supp(A)$; and
\item $\lim_{M\to\infty} \Tr\left(\ketbra{\psi} - MA\right)_+ = 0$.
\end{enumerate}
\end{lemma}

\begin{proof}
Clearly~(b) implies~(a), because if $\ket{\psi} \notin \supp(A)$ then we can find some $\ket{\phi}\in \HH$ with $\braket{\phi|A|\phi} = 0$ and $\braket{\psi|\phi}\neq 0$. But then
\bb
\Tr\left(\ketbra{\psi} - MA\right)_+ \geq \bra{\phi}\left(\ketbra{\psi} - MA\right)\ket{\phi} = |\!\braket{\psi|\phi}\!|^2 > 0
\ee
for all $M$, where the first inequality follows by setting $X = \ketbra{\psi} - MA$ and $Q = \ketbra{\phi}$ in~\eqref{second_variational_program_trace_X_plus}.

We will now show that~(a) implies~(b). Let $A = \sum_i a_i \ketbra{i}$ be a spectral decomposition of $A$. Up to restricting the range of $i$, which might be finite or countably infinite, we can assume that $a_i > 0$ for all $i$, so that $\supp(A) = \overline{\Span}\{\ket{i}\}_i$, where the bar denotes (norm) closure. By assumption, $\ket{\psi}$ has an expansion of the form $\ket{\psi} = \sum_i \psi_i \ket{i}$. We can thus find a sequence of (normalised) vectors $\ket{\psi_n} \propto \sum_{i\in I_n} \psi_i \ket{i}$, each of which with a finite expansion (that is, $|I_n|<\infty$) such that $\lim_{n\to\infty} \braket{\psi_n|\psi} = 1$. Since 
\bb
A \geq \sum_{i\in I_n} a_i \ketbra{i} \geq \big(\min\nolimits_{i\in I_n} a_i\big)\, P_n \geq \big(\min\nolimits_{i\in I_n} a_i\big) \ketbra{\psi_n}\, ,
\ee
where $P_n \coloneqq \sum_{i\in I_n} \ketbra{i}$, for all $M \geq \big(\min_{i\in I_n} a_i\big)^{-1}$ we have
\bb
\Tr\left(\ketbra{\psi} - MA\right)_+ &\leq \Tr\left(\ketbra{\psi} - M \big(\min\nolimits_{i\in I_n} a_i\big) \ketbra{\psi_n} \right)_+ \\
&\leq \Tr\left(\ketbra{\psi} - \ketbra{\psi_n} \right)_+ \\
&= \frac12 \left\|\ketbra{\psi} - \ketbra{\psi_n} \right\|_1 \\
&= \sqrt{1 - |\!\braket{\psi_n|\psi}\!|^2}\, ,
\ee
where in the first two inequalities we have used Lemma~\ref{variational_program_trace_X_plus_lemma}(a), and in the last equality~\eqref{trace_distance_pure_states}. Taking the limit $M\to\infty$ first and $n\to\infty$ second proves the claim.
\end{proof}

\subsection{Proof of the quantum blurring lemma} \label{subsec_proof_quantum_blurring_lemma}

This section is devoted to a full proof of the quantum blurring lemma (Lemma~\ref{quantum_blurring_lemma}). This relies heavily on the results established throughout Sections~\ref{subsec_alternative_expressions_blurring}--\ref{subsec_verifying_support_condition}. We will also need the simple lemma below, which builds on a result originally found by Renner~\cite[Lemma~4.2.2]{RennerPhD}.

\begin{lemma} \label{perm_inv_purifications_lemma}
Let $\omega_n = \omega_n^{A^n} = \mathcal{S}_n\big(\omega_n^{A^n}\big)$ and $\tau_n = \tau_n^{A^n} = \mathcal{S}_n\big(\tau_n^{A^n}\big)$ be two permutationally invariant states on a finite-dimensional $n$-copy quantum system $A^n$ (see~\eqref{symmetrisation} for the definition of the symmetrisation operator $\mathcal{S}_n$). Then there exist a Hilbert space $\HH_E \simeq \HH_A$ and two permutationally invariant purifications
\bb
\ket{\psi_n}_{A^n E^n},\, \ket{\phi_n}_{A^n E^n} \in \mathrm{Sym}^n\big(\HH_{AE}\big) 
\ee
of $\omega_n$ and $\tau_n$, respectively, with the property that
\bb
\R \ni \braket{\psi_n|\phi_n} \geq 1 - \frac12 \left\|\omega_n - \tau_n\right\|_1\, .
\ee
Furthermore, if one of the two states is i.i.d., that is, for example, $\omega_n = \omega_1^{\otimes n}$, then the same is true of the corresponding purification, that is, $\ket{\psi_n}_{A^n E^n} = \ket{\psi}_{AE}^{\otimes n}$.
\end{lemma}

\begin{proof}
It suffices to set
\bb
\ket{\psi_n}\coloneqq \sqrt{\omega_n^{A^n}}\! \otimes \id_{E^n}\! \ket{\Phi}_{AE}^{\otimes n} ,\qquad \ket{\phi_n}\coloneqq \sqrt{\tau_n^{A^n}}\! \otimes \id_{E^n}\! \ket{\Phi}_{AE}^{\otimes n} ,
\label{construction_symmetric_purifications}
\ee
where $\ket{\Phi}_{AE} \coloneqq \sum_{i=1}^{|A|} \ket{i}_A \otimes \ket{i}_E$ is the un-normalised maximally entangled state, and $|A| = \dim \HH_A$. Using the Holevo inequality~\cite{Holevo1972}, here reported in~\eqref{Holevo_Fuchs_van_de_Graaf}, we then find that
\bb
\braket{\psi_n|\phi_n} = \Tr \left[ \sqrt{\omega_n} \sqrt{\tau_n} \right] \geq 1 - \frac12 \left\|\omega_n - \tau_n\right\|_1\, ,
\ee
as claimed. Moreover, from the construction in~\eqref{construction_symmetric_purifications} it is clear that if one of the two states is i.i.d.\ then the same is true of the corresponding purification.
\end{proof}

We are now ready to present the full proof of the quantum blurring lemma, which is the linchpin of our proof of the generalised quantum Stein's lemma.

\begin{proof}[Proof of Lemma~\ref{quantum_blurring_lemma} (asymptotic quantum blurring lemma)]
Due to the fact that the trace of the positive part of an operator obeys the data processing inequality (Lemma~\ref{variational_program_trace_X_plus_lemma}(c)), it suffices to prove the claim for the case of pure $\rho$ and pure (and permutationally symmetric) $\rho_n$. Indeed, suppose that the claim has been shown in this latter setting. Then, for an arbitrary state $\rho = \rho_A$ and a sequence of $n$-copy states $\rho_n$ as in the statement, because of Lemma~\ref{perm_inv_purifications_lemma} we can consider two symmetric purifications
\bb
\ket{\psi}^{\otimes n} = \ket{\psi}_{AE}^{\otimes n},\qquad \ket{\phi_n} = \ket{\phi_n}_{A^nE^n} \in \mathrm{Sym}^n\big(\HH_{AE}\big)
\ee
of $\rho^{\otimes n}$ and $\rho_n$, respectively, such that
\bb
\Tr_{E^n} \psi^{\otimes n} &= \Tr_{E^n} \ketbra{\psi}^{\otimes n}_{AE} = \rho^{\otimes n}_A ,\\
\Tr_{E^n} \phi_n &= \Tr_{E^n} \ketbra{\phi_n}_{A^nE^n} = \rho_n^{A^n}\! ,
\label{partial_traces_symm_purifications}
\ee
and
\bb
\R \ni \braket{\psi^{\otimes n}|\phi_n} \geq 1 - \frac12 \left\|\rho^{\otimes n} - \rho_n\right\|_1\, .
\ee
Therefore, using the expression~\eqref{trace_distance_pure_states} for the trace distance between two pure states, we have 
\bb
\limsup_{n\in I} \frac12 \left\| \psi^{\otimes n} - \phi_n \right\|_1 &= \limsup_{n\in I} \sqrt{1- \big|\braket{\psi_n^{\otimes n} | \phi_n}\big|^2} \\
&= \sqrt{1- \liminf_{n\in I} \big|\braket{\psi_n^{\otimes n} | \phi_n}\big|^2} \\
&\leq \sqrt{1- \left(1 - \limsup_{n\in I} \frac12\left\|\rho^{\otimes n} - \rho_n\right\|_1\right)^2} \\
&< \sqrt{1-(1-\e)^2} \\
&\eqqcolon \e' \in (0,1)\, .
\label{epsilon_prime}
\ee
Now, observe that
\bb
\Tr_{E^n} \widebar{B}^{\,\psi}_{n,\delta}(\phi_n) &\eqt{(i)} \Tr_{E^n} \Tr_{A^{\floor{\delta n}} E^{\floor{\delta n}}} \mathcal{S}^{(AE)}_{n+\floor{\delta n}}\left( \phi_n^{A^nE^n} \otimes \psi_{AE}^{\otimes \floor{\delta n}} \right) \\
&\eqt{(ii)} \Tr_{A^{\floor{\delta n}}} \Tr_{E^{n+\floor{\delta n}}} \mathcal{S}^{(AE)}_{n+\floor{\delta n}}\left( \phi_n^{A^nE^n} \otimes \psi_{AE}^{\otimes \floor{\delta n}} \right) \\
&\eqt{(iii)} \Tr_{A^{\floor{\delta n}}} \mathcal{S}^{(A)}_{n+\floor{\delta n}}\left( \Tr_{E^{n+\floor{\delta n}}}  \phi_n^{A^nE^n} \otimes \psi_{AE}^{\otimes \floor{\delta n}} \right) \\
&\eqt{(iv)} \Tr_{A^{\floor{\delta n}}} \mathcal{S}^{(A)}_{n+\floor{\delta n}}\left( \rho_n^{A^n} \otimes \rho_A^{\otimes \floor{\delta n}} \right) \\
&= \widebar{B}^\rho_{n,\delta} (\rho_n)\, .
\label{partial_trace_output_blurring}
\ee
Here, the notation $\mathcal{S}_N^{(AE)}$ after~(i) indicates that the many copies of $A$ and $E$ are understood to be permuted jointly, i.e.\ each pair $AE$ is treated as a single system. Then, (ii)~follows by re-arranging the traces, while in~(iii) we observed that randomly permuting the $N = n+\floor{\delta n}$ pairs of systems $AE$ and subsequently tracing away all $N$ copies of the system $E$ is equivalent to tracing away all the $E$ systems first and then randomly permuting the $N$ copies of $A$ only. Finally, (iv)~follows from~\eqref{partial_traces_symm_purifications}.

Putting all together, and leaving $I'\subseteq I$ arbitrary for the time being, we have 
\bb
&\lim_{M\to\infty} \limsup_{n \in I'} \Tr \left(\rho^{\otimes n} - M \int_0^\Delta \frac{\dd\delta}{\Delta}\, \widebar{B}_{n,\delta}^{\rho}(\rho_n) \right)_+ \\
&\quad \eqt{(v)} \lim_{M\to\infty} \limsup_{n \in I'} \Tr\left( \Tr_{E^n} \left[ \psi^{\otimes n} - M \int_0^\Delta \frac{\dd\delta}{\Delta}\, \widebar{B}_{n,\delta}^{\,\psi}(\phi_n) \right] \right)_{+} \\
&\quad \leqt{(vi)} \lim_{M\to\infty} \limsup_{n \in I'} \Tr \left( \psi^{\otimes n} - M \int_0^\Delta \frac{\dd\delta}{\Delta}\, \widebar{B}_{n,\delta}^{\,\psi}(\phi_n) \right)_{+} \\
&\quad \eqt{(vii)}\ 0\, .
\ee
Here, (v)~follows from~\eqref{partial_traces_symm_purifications} and~\eqref{partial_trace_output_blurring}, (vi)~holds because of Lemma~\ref{variational_program_trace_X_plus_lemma}(c), while~(vii) is the statement of Lemma~\ref{quantum_blurring_lemma} in the case of pure $\rho$ and pure $\rho_n$, applied with $\e\mapsto \e'$, where $\e'$ is defined in~\eqref{epsilon_prime}. In this latter step we fix also the infinite set $I'\subseteq I$.

We can therefore consider, without loss of generality, the case where $\rho = \ketbra{0}$ and $\rho_n = \ketbra{\phi_n} = \phi_n$ are both pure. (Here, the name $\ket{0}$ for the vector that defines $\rho$ is purely conventional; this notation will however prove useful later.) Since $\phi_n$ is permutationally symmetric, we have $\ket{\phi_n}\in \mathrm{Sym}^n\big(\C^d\big)$. By assumption,
\bb
\limsup_{n\in I} \frac12 \left\| \ketbra{0}^{\otimes n} - \phi_n \right\|_1 \leq \e < 1\, ,
\ee
which becomes as before
\bb
\liminf_{n\in I}\, \big|\!\braket{0^n|\phi_n}\!\big| \geq \sqrt{1-\e^2} > 0\, .
\label{limit_overlap}
\ee
We will now show that we can replace the blurring map $\widebar{B}^{\rho}_{n,\delta}$ with its more handy version $B_{n,\delta}$ (cf.~\eqref{quantum_blurring} and~\eqref{blurring}). Indeed, observe that the (mixed) state $\widebar{B}_{n,\delta}^{\,\ket{0}\!\bra{0}}(\phi_n)$ is permutationally symmetric by construction, so that
\bb
\widebar{B}_{n,\delta}^{\,\ket{0}\!\bra{0}}(\phi_n) \geq \Pi_n \widebar{B}_{n,\delta}^{\,\ket{0}\!\bra{0}}(\phi_n) \Pi_n = B_{n,\delta}(\phi_n)
\ee
because the $\widebar{B}_{n,\delta}^{\,\ket{0}\!\bra{0}}(\phi_n)$ and $\Pi_n$ commute. Since due to Lemma~\ref{variational_program_trace_X_plus_lemma}(a)
\bb
\Tr \left( \ketbra{0}^{\otimes n} - M \int_0^{\Delta} \frac{\dd \delta}{\Delta}\ \widebar{B}_{n,\delta}^{\,\ket{0}\!\bra{0}}(\phi_n) \right)_+ \leq \Tr \left( \ketbra{0}^{\otimes n} - M \int_0^{\Delta} \frac{\dd \delta}{\Delta}\ B_{n,\delta}(\phi_n) \right)_+ ,
\label{blurring_lemma_proof_projecting}
\ee
it suffices to show that the right-hand side tends to zero once one takes the limit superior in $n\in I'$ (first) and the limit $M\to\infty$ (second), for some infinite $I'\subseteq I$. In fact, it is not difficult to realise that~\eqref{blurring_lemma_proof_projecting} holds with equality, because $\ket{0}^{\otimes n}$ belongs itself to the symmetric space.

We now lift everything to the Fock space $\HH_{d-1}$, using the canonical embedding $U_n: \mathrm{Sym}^n\big(\C^d\big)\to \HH_{d-1}$ defined by~\eqref{embedding}. Since $U_n$ is an isometry, we have
\bb
\Tr \left( \ketbra{0}^{\otimes n} \!\! - M\! \int_0^{\Delta} \!\frac{\dd \delta}{\Delta}\, B_{n,\delta}(\phi_n) \!\right)_{+} &= \Tr \left( U_n^{\vphantom{\dag}}\! \left(\!\ketbra{0}^{\otimes n} \!\! - M\! \int_0^{\Delta} \!\frac{\dd \delta}{\Delta}\, B_{n,\delta}\big( U_n^\dag U_n^{\vphantom{\dag}}\phi_n U_n^\dag U_n^{\vphantom{\dag}}\big)\!\right) U_n^\dag \right)_{+} \\
&\, = \Tr \left( \ketbra{0}^{\otimes (d-1)}\! - M \!\int_0^{\Delta} \!\frac{\dd \delta}{\Delta}\, \widetilde{B}_{n,\delta}\big(\widetilde{\phi}_n\big) \right)_{+} ,
\label{blurring_lemma_proof_lifting}
\ee
where $\widetilde{B}_{n,\delta}$ is defined by~\eqref{lifted_blurring}, and we set $\ket{\widetilde{\phi}_n} \coloneqq U_n \ket{\phi_n} \in \HH_{d-1}$, with $\widetilde{\phi}_n \coloneqq \ketbra{\widetilde{\phi}_n} \in \D\big(\HH_{d-1}\big)$.

The sequence $\big(\widetilde{\phi}_n\big)_{n\in I}$ is trace-norm-bounded (as $\big\|\widetilde{\phi}_n\big\|_1 = \Tr \widetilde{\phi}_n = 1$ for all $n\in I$). Hence, by Lemma~\ref{weak_star_seq_compactness_lemma} we can find an infinite subset $I'\subseteq I$ --- i.e.\ a sub-sequence of $\big(\widetilde{\phi}_n\big)_{n\in I}$ --- such that
\bb
\widetilde{\phi}_{n} \ctends{w^*}{n\in I'}{0pt} \omega\, ,
\label{blurring_lemma_proof_weak_star_convergence}
\ee
where $\omega\geq 0$ and $\Tr \omega\leq 1$ (i.e.\ $\omega$ is a sub-normalised state). Note that by definition of weak* convergence
\bb
\braket{0|\omega|0} = \lim_{n\in I'}\, \big|\!\braket{0|\widetilde{\phi}_n}\!\big|^2 = \lim_{n\in I'}\, |\!\braket{0^n|\phi_n}\!|^2 \geq 1 - \e^2 > 0\, ,
\ee
entailing that $\omega\neq 0$. Set
\bb
\Lambda_n \coloneqq \int_0^\Delta \frac{\dd\delta}{\Delta}\ \widetilde{B}_{n,\delta}\, ,
\label{Lambda_n}
\ee
and construct the sub-channel on $\HH_{d-1}$ given by
\bb
\Lambda \coloneqq \int_0^\Delta \frac{\dd\delta}{\Delta}\ \big(\EE_{\lambda(\delta)} \circ \D_{\mu(\delta)} \big)^{\otimes (d-1)} ,
\label{Lambda}
\ee
where $\lambda(\delta)$ and $\mu(\delta)$ are defined in~\eqref{lambda_and_mu}. Using Lemma~\ref{variational_program_trace_X_plus_lemma}(b) repeatedly, we now write
\bb
&\Tr \left( \ketbra{0}^{\otimes (d-1)} - M \int_0^{\Delta} \frac{\dd \delta}{\Delta}\ \widetilde{B}_{n,\delta}\big(\widetilde{\phi}_n\big) \right)_+ \\
&\quad = \Tr \left( \ketbra{0}^{\otimes (d-1)} - M \Lambda_n \big(\widetilde{\phi}_n\big) \right)_+ \\
&\quad \leq \Tr \left( \ketbra{0}^{\otimes (d-1)} - M \Lambda (\omega) \right)_+ + M \Tr \left( \big(\Lambda - \Lambda_n\big) (\omega) \right)_+ + M \Tr \left( \Lambda_n\!\left(\omega - \widetilde{\phi}_n\right) \right)_+ \\
&\quad \leq \Tr \left( \ketbra{0}^{\otimes (d-1)} - M \Lambda (\omega) \right)_+ + M \left\| \big(\Lambda - \Lambda_n\big) (\omega) \right\|_1 + M \left\| \Lambda_n\!\left(\omega - \widetilde{\phi}_n\right) \right\|_1\, .
\label{blurring_lemma_proof_decomposition}
\ee
Note that
\bb
\limsup_{n\to\infty} \left\| \big(\Lambda - \Lambda_n\big) (\omega) \right\|_1 &\leq \lim_{n\to\infty} \int_0^\Delta \frac{\dd\delta}{\Delta}\ \left\|\left(\widetilde{B}_{n,\delta} - \big(\EE_{\lambda(\delta)} \circ \D_{\mu(\delta)} \big)^{\otimes (d-1)}\right)(\omega) \right\|_1 = 0\, ,
\label{blurring_lemma_proof_vanishing_second_term}
\ee
where the last equality is obtained by taking the limit inside the integral sign and applying Corollary~\ref{lifted_blurring_strong_convergence_cor}. To perform the first of the above steps, we can use Lebesgue's dominated convergence theorem, since 
\bb
\left\|\left(\widetilde{B}_{n,\delta} - \big(\EE_{\lambda(\delta)} \circ \D_{\mu(\delta)} \big)^{\otimes (d-1)}\right)(\omega) \right\|_1 \leq \left\|\widetilde{B}_{n,\delta}(\omega) \right\|_1 + \left\| \big(\EE_{\lambda(\delta)} \circ \D_{\mu(\delta)} \big)^{\otimes (d-1)}(\omega) \right\|_1 \leq 2
\ee
due to the fact that both $\widetilde{B}_{n,\delta}$ and $\EE_{\lambda(\delta)} \circ \D_{\mu(\delta)}$ are sub-channels and hence contractive with respect to the trace norm (see~\eqref{contractivity_trace_norm}). This takes care of the term $M \left\| \big(\Lambda - \Lambda_n\big) (\omega) \right\|_1$ appearing on the second-to-last line of~\eqref{blurring_lemma_proof_decomposition}.

We now focus our attention on the third term. We have
\bb
\limsup_{n\in I'} \left\| \Lambda_n\!\left(\omega - \widetilde{\phi}_n\right) \right\|_1 \leq \lim_{n\in I'} \int_0^\Delta \frac{\dd\delta}{\Delta}\ \left\| \widetilde{B}_{n,\delta}\!\left(\omega - \widetilde{\phi}_n\right) \right\|_1 = 0\, ,
\label{blurring_lemma_proof_vanishing_third_term}
\ee
where once again the equality is proved by swapping limit and integral via Lebesgue's dominated convergence and then using Corollary~\ref{damping_weakly_vanishing_seq_cor} together with~\eqref{blurring_lemma_proof_weak_star_convergence}. That both steps are possible follows from the trace norm estimates
\bb
\left\| \widetilde{B}_{n,\delta}\!\left(\omega - \widetilde{\phi}_n\right) \right\|_1 \leq \left\| \omega - \widetilde{\phi}_n \right\|_1 \leq 2\, ,
\ee
which follow once again from~\eqref{contractivity_trace_norm}.

Taking the limit superior on $n\in I'$ in~\eqref{blurring_lemma_proof_decomposition} and using~\eqref{blurring_lemma_proof_vanishing_second_term} and~\eqref{blurring_lemma_proof_vanishing_third_term} gives
\bb
\limsup_{n\in I'} \Tr \left( \ketbra{0}^{\otimes (d-1)} - M \int_0^{\Delta} \frac{\dd \delta}{\Delta}\ \widetilde{B}_{n,\delta}\big(\widetilde{\phi}_n\big) \right)_+ \leq \Tr \left( \ketbra{0}^{\otimes (d-1)} - M \Lambda (\omega) \right)_+ .
\ee
Taking the limit $M\to\infty$ on both sides and applying Lemmas~\ref{vacuum_in_support_lemma} and~\ref{support_lemma} yields
\bb
\lim_{M\to\infty} \limsup_{n\in I'} \Tr \left( \ketbra{0}^{\otimes (d-1)} - M \int_0^{\Delta} \!\frac{\dd \delta}{\Delta}\, \widetilde{B}_{n,\delta}\big(\widetilde{\phi}_n\big) \right)_{+} = 0\, .
\ee
(Note that although $\omega$ is sub-normalised rather than normalised, this does not impact the application of Lemma~\ref{vacuum_in_support_lemma} because we have shown before that $\omega \neq 0$, so that $\omega = \kappa \omega'$ for some normalised state $\omega'$ and some $\kappa>0$.) Due to~\eqref{blurring_lemma_proof_projecting} and~\eqref{blurring_lemma_proof_lifting}, this concludes the proof.
\end{proof}

\section{Extension to almost power states} \label{sec_extension_GQSL_almost_iid}

The setting of entanglement testing, or of resource testing more generally, features a null hypothesis that is perfectly i.i.d., that is, of the form $\rho^{\otimes n}$ (see Section~\ref{sec_intro}). This is certainly an idealisation, as it is impossible to certify that a given source produces states that are \emph{exactly} product and \emph{exactly} independent, no matter how many times it is used. What we can realistically certify, instead, is that the correlations between different outcomes of the source are small, especially in the limit in which the source is used many times. Physically, this should not be a problem, as any procedure that allows us to carry out resource testing on states that are exactly i.i.d.\ should realistically also work for states that are only approximately such. But proving this mathematically is highly non-obvious.

The first challenge is to understand how to model an imperfect source that produces approximately i.i.d.\ states. There are several ways to do so, and we will discuss some of them in Section~\ref{subsec_almost_iid}. We will concentrate on a class of sources, called `almost power states', that was identified by Brand\~{a}o and Plenio themselves~\cite[p.~803]{Brandao2010}, based on previous work by Renner~\cite{RennerPhD}, in the course of their original proof attempt. Here, `power' is to be intended as `tensor power', and thus as a synonym of `i.i.d.' After defining the problem precisely, we will extend the generalised quantum Stein's lemma to almost power states in Section~\ref{subsec_extension_GQSL_almost_iid}.

\subsection{Almost power states} \label{subsec_almost_iid}

Sources that produce approximately i.i.d.\ states can be modelled in several different and in principle inequivalent ways. The first basic definitions that are needed to discuss these notions are due to Renner~\cite{RennerPhD}: for a vector $\ket{\psi}$ in some Hilbert space $\HH$ and two non-negative integers $n\geq r\geq 0$, we set~\cite[Eq.~(4.2) and Definition~4.1.4]{RennerPhD}
\bb
\VV(\ket{\psi}; n,r) &\coloneqq \Span\Big\{ U_\pi \!\left(\ket{\psi}^{\otimes (n-r)} \!\otimes \ket{\mu_r} \right)\! : \ \pi\in S_n,\ \ket{\mu_r} \in \HH^{\otimes r} \Big\} , \\
\VVsym(\ket{\psi}; n,r) &\coloneqq \VV(\ket{\psi}; n,r) \cap \mathrm{Sym}^n(\HH)\, .
\ee
Intuitively, $\VV(\ket{\psi}; n,r)$ contains the $n$-copy pure states that are equal to $\ket{\psi}$ on every site except perhaps for at most $r$ `defective' sites, which can be prepared in some arbitrary pure state $\ket{\mu_r}$. If we require permutational symmetry of the global state, which should model the invariance of the source under any exchange of sites as in the setting of de Finetti's theorem~\cite{deFinetti-original}, we obtain the space $\VVsym(\ket{\psi}; n,r)$. We can think of either of these two spaces as modelling a source that produces a global pure state that is approximately equal, in some sense, to $n$ independent copies of $\ket{\psi}$. Renner's work proves that in particular the space $\VVsym(\ket{\psi}; n,r)$ provides a useful generalisation of the notion of i.i.d.\ state: it is precisely these states that appear in his `exponential de Finetti' theorem~\cite[Theorem~4.3.2]{RennerPhD}.

While the above definitions may be deemed satisfactory for pure states, also in light of their applicability and mathematical fruitfulness, it is much less clear how to proceed in the case of mixed states. Fortunately, Brand\~{a}o and Plenio themselves have identified a possible way out: in their original paper on the generalised quantum Stein's lemma~\cite[p.~803]{Brandao2010}, they offer the following definition.

\begin{Def}[{(Almost power states~\cite[p.~803]{Brandao2010})}] \label{almost_power_states_Def}
Let $\rho\in \D(\HH)$ be a state of a quantum system modelled by a Hilbert space $\HH$, and let $\omega_n\in \D(\HH^{\otimes n})$ be another state on $n$ copies of that system. For some integer $0\leq r\leq n$, we say that $\omega_n$ is an \deff{$\boldsymbol{(n,r)}$-almost power state} along $\rho$, and we write $\omega_n \in \aiid{\rho,r}_n$, if there are purifications $\ket{\Gamma}\in (\HH \otimes \HH')^{\otimes n}$ of $\omega_n$ and $\ket{\psi}\in \HH\otimes \HH'$ of $\rho$ such that
\bb
\ket{\Gamma} \in \VVsym(\ket{\psi}; n,r)\, .
\ee 
\end{Def}

\begin{rem} \label{almost_power_perm_symm_rem}
An $(n,r)$-almost power state is by construction permutationally symmetric, in the sense that $\mathcal{S}_n(\omega_n) = \omega_n$, where the symmetrisation operator is given by~\eqref{symmetrisation}.
\end{rem}

The rationale of the above definition is that a source that is in some sense approximately i.i.d.\ should be purifiable to an approximately i.i.d.\ and symmetric pure source. However, it is worth remarking at this point that the above definition is by no means the only possible one. Other alternatives have been recently proposed by Mazzola, Sutter, and Renner~\cite{Mazzola-talk-BIID} and by Renner himself~\cite{Renner-talk-Cambridge}. This renewed interest played a role in convincing the author to advertise the fact, of which he was aware for some time, that the proof of the generalised quantum Stein's lemma presented here can be adapted, \emph{essentially for free}, to cover resource testing of almost power sources. 

Brand\~{a}o and Plenio introduce almost power states to prove a variant of Renner's exponential de Finetti theorem~\cite[Lemma~III.5]{Brandao2010}. The endgame of that part of the argument is their Lemma~III.7, whose proof, interestingly, is precisely the point where the incorrect Lemma~III.9 is employed (see~\cite{gap} for a detailed explanation of the issue). Because of this, Lemma~III.7 is currently not known to be correct. However, physical intuition suggests that it should hold: if the number of defective sites $r$ is sub-linear in $n$, the value of an asymptotically meaningful, extensive quantity such as the regularised relative entropy of entanglement should be unaffected. And indeed, below we show how to recover a weaker variant of Lemma~III.7, in which $r$ is assumed to be bounded rather than merely sub-linear in $n$.

\subsection{Extension of the generalised quantum 
Stein's lemma to almost power states} \label{subsec_extension_GQSL_almost_iid}

The problem we want to tackle here is the calculation of the Stein exponent of the following more complicated variant of the task of resource testing described in Section~\ref{sec_intro}. The parameters of the problem are a state $\rho\in \D(\HH)$, a family $(\FF_n)_n$ of sets $\FF_n\subseteq \D\big(\HH^{\otimes n}\big)$ obeying the Brand\~{a}o--Plenio axioms, and some non-negative integer $r$ (arbitrary but independent of $n$):
\begin{itemize}
\item Null hypothesis: the unknown state is an $(n,r)$-almost power state $\omega_n\in \aiid{\rho,r}_n$.
\item Alternative hypothesis: the unknown state is some arbitrary $\sigma_n\in \FF_n$.
\end{itemize}
The corresponding Stein exponent is given by
\bb
\rel{\stein}{\aiid{\rho,r}}{\FF} &\coloneqq \lim_{\e\to 0^+} \liminf_{n\to\infty} \frac1n\, \rel{D_H^\e}{\aiid{\rho,r}_n}{\FF_n}\, , \\
\rel{D_H^\e}{\aiid{\rho,r}_n}{\FF_n} &\coloneqq \inf_{\omega_n\in \aiid{\rho,r}_n\!\!,\ \sigma_n\in \FF_n} D_H^\e(\omega_n \| \sigma_n)\, .
\ee
Note that for $r=0$ we recover the problem studied before, as in this case $\aiid{\rho,0}_n = \{\rho^{\otimes n}\}$. Since the defective systems are at most $r$, and their number is therefore independent of $n$, they will form a vanishing fraction of the total; hence, it is only natural to speculate that their presence should bear essentially no physical consequence at all. This leads us to conjecture that
\bb
\rel{\stein}{\aiid{\rho,r}}{\FF} \eqt{?} \stein(\rho\|\FF) = D^\infty(\rho\|\FF)\, .
\ee
Below we will show that this is indeed true. Moreover, this conclusion can be reached by running essentially the same argument that allowed us to prove Theorem~\ref{GQSL_thm}, and operating only a minor modification to the proof of the quantum blurring lemma.

\begin{thm} \label{almost_iid_GQSL_thm}
Let $\HH$ be a finite-dimensional Hilbert space, and let $(\FF_n)_n$ be a sequence of sets of states $\FF_n\subseteq \D\big(\HH^{\otimes n}\big)$ that obeys the Brand\~{a}o--Plenio axioms. For some $\rho\in \D(\HH)$ and some fixed integer $r\in \N$, let $\aiid{\rho,r} = \big(\aiid{\rho,r}_n\big)_n$ be the sequence of sets $\aiid{\rho,r}_n$ of $(n,r)$-almost power states\footnote{Strictly speaking, this is only defined for $n\geq r$. Since we are looking at the asymptotic limit anyway, this is by no means a problem.} along $\rho$ (Definition~\ref{almost_power_states_Def}). Then it holds that
\bb
\lim_{n\to\infty} \frac1n\, \rel{D_H^\e}{\aiid{\rho,r}_n}{\FF_n} = D^\infty(\rho\|\FF) \qquad \forall\ \e\in (0,1)\, ,
\label{almost_iid_GQSL_DH}
\ee
implying that
\bb
\rel{\stein}{\aiid{\rho,r}}{\FF} = D^\infty(\rho\|\FF)\, .
\label{almost_iid_GQSL}
\ee
\end{thm}

The proof of the above result is essentially identical to that of Theorem~\ref{GQSL_thm}. The only change is that we will rely on the following modified version of the quantum blurring lemma.

\begin{lemma} \label{almost_iid_blurring_lemma}
Let $\rho\in \D\big(\C^d\big)$ be a finite-dimensional state, and for some infinite set $I\subseteq \N$ let $(\rho_n)_{n\in I}$ be a sequence of permutationally symmetric $n$-copy states $\rho_n = \mathcal{S}_n(\rho_n) \in \D\big( (\C^d)^{\otimes n}\big)$ such that
\bb
\limsup_{n\in I} \frac12\, \big\| \rho_n - \aiid{\rho,r}_n \big\|_1 < 1
\ee
for some $r\in \N$, where $\big\| \rho_n - \aiid{\rho,r}_n \big\|_1\coloneqq \inf_{\omega_n \in \aiid{\rho,r}_n} \|\rho_n - \omega_n\|_1$. Then there exists an infinite subset $I'\subseteq I$ such that, for all $\Delta\in (0,\frac12]$,
\bb
\lim_{M\to\infty} \limsup_{n\in I'} \Tr \left( \rho^{\otimes n} - M \int_0^{\Delta} \frac{\dd \delta}{\Delta}\ \widebar{B}_{n,\delta}^\rho(\rho_n) \right)_+ = 0\, ,
\ee
where $\widebar{B}_{n,\delta}^\rho$ is defined by~\eqref{quantum_blurring}.
\end{lemma}

\begin{proof}
By assumption, there exists a state $\tau_n\in \aiid{\rho,r}_n$ and some $\e\in (0,1)$ such that $\frac12\, \big\| \rho_n - \tau_n \big\|_1 \leq \e$ for all sufficiently large $n\in I'$. Since $\tau_n$ is an $(n,r)$-almost power state along $\rho$, by Definition~\ref{almost_power_states_Def} there exists a purification of $\rho\in \D\big(\C^d\big)$, which we denote by $\ket{0}\in \C^d\otimes \C^d = \C^{d^2}$, and a purification of $\tau_n\in \D\big((\C^d)^{\otimes n}\big)$, denoted by $\ket{\Gamma_n}\in \big(\C^{d^2}\big)^{\otimes n}$, such that
\bb
\ket{\Gamma_n} \in \VVsym(\ket{0}; n,r)\, .
\label{almost_iid_blurring_proof_eq1}
\ee
Using Uhlmann's theorem~\cite{Uhlmann-fidelity} and the Fuchs--van de Graaf inequalities~\eqref{Holevo_Fuchs_van_de_Graaf}, we can also find a purification of $\rho_n$ on the same space, denoted by $\ket{\phi_n}\in \big(\C^{d^2}\big)^{\otimes n}$, such that
\bb
\braket{\phi_n|\Gamma_n} = F(\rho_n,\tau_n) \geq 1 - \frac12\left\|\rho_n - \tau_n\right\|_1 \geq 1 - \e\, ,
\label{almost_iid_blurring_proof_eq2}
\ee
so that
\bb
\frac12 \left\| \phi_n - \Gamma_n \right\|_1 \leq \sqrt{1-(1-\e)^2} = \sqrt{\e(2-\e)}\, .
\label{almost_iid_blurring_proof_eq3}
\ee
Exactly as in the proof of Lemma~\eqref{quantum_blurring_lemma} (cf.~\eqref{blurring_lemma_proof_projecting}), we have
\bb
\Tr \left( \rho^{\otimes n} - M \int_0^{\Delta} \frac{\dd \delta}{\Delta}\ \widebar{B}_{n,\delta}^{\,\rho}(\rho_n) \right)_+ &\leq \Tr \left( \ketbra{0}^{\otimes n} - M \int_0^{\Delta} \frac{\dd \delta}{\Delta}\ \widebar{B}^{\,\ket{0}\!\bra{0}}_{n,\delta}(\phi_n) \right)_+ \\
&\leq \Tr \left( \ketbra{0}^{\otimes n} - M \int_0^{\Delta} \frac{\dd \delta}{\Delta}\ B_{n,\delta}(\phi_n) \right)_+\, .
\ee
By lifting everything to the Fock space $\HH_{d^2-1}$ via the embedding $U_n: \mathrm{Sym}^n\big(\C^{d^2}\big)\to \HH_{d^2-1}$ defined by~\eqref{embedding} (with $d$, the dimension of the underlying space, replaced by $d^2$), we also know that
\bb
\Tr \left( \ketbra{0}^{\otimes n} - M \int_0^{\Delta} \frac{\dd \delta}{\Delta}\ B_{n,\delta}(\phi_n) \right)_+ = \Tr \left( \ketbra{0}^{\otimes (d^2-1)} - M \int_0^{\Delta} \frac{\dd \delta}{\Delta}\ \widetilde{B}_{n,\delta}\big(\widetilde{\phi}_n\big) \right)_+ ,
\ee
as in~\eqref{blurring_lemma_proof_lifting}. Here, $\ket{\widetilde{\phi}_n} \coloneqq U_n \ket{\phi_n} \in \HH_{d^2-1}$, and the lifted blurring map is defined by~\eqref{lifted_blurring}.

As before, since $\big(\widetilde{\phi}_n\big)_{n\in I}$ is trace-norm-bounded, we can extract a weak*-converging subsequence $\big(\widetilde{\phi}_n\big)_{n\in I'}$, for some infinite $I'\subseteq I$. This means that, as in~\eqref{blurring_lemma_proof_weak_star_convergence}, $\phantom{\Big|}\widetilde{\phi}_{n} \ctends{w^*}{n\in I'}{0pt} \omega$ for some trace class, positive semi-definite $\omega\geq 0$. 

The key difference with the i.i.d.\ case lies in how we prove that $\omega \neq 0$. Before, we argued that $\braket{0|\omega|0} > 0$, due to the fact that the overlap between $\ket{\widetilde{\phi}_n}$ and $\ket{0}^{\otimes (d^2-1)}$ in $\HH_{d^2-1}$ --- equivalently, the overlap between $\ket{\phi_n}$ and $\ket{0}^{\otimes n}$ in $\big(\C^{d^2}\big)^{\otimes n}$ --- was uniformly bounded away from zero for all $n$. Now this argument fails, because the only guarantee we have on $\ket{\phi_n}$ is that it has a non-vanishing overlap with a vector in $\VVsym(\ket{0}; n,r)$, that is, $\ket{\Gamma_n}$ (see~\eqref{almost_iid_blurring_proof_eq1}). However, $\VVsym(\ket{0}; n,r)$ contains plenty of vectors that are orthogonal to $\ket{0}^{\otimes n}$, so $\ket{\phi_n}$ itself could also be orthogonal to $\ket{0}^{\otimes n}$ for all $n$. The previous argument, therefore, cannot work.

To find an alternative way to prove that indeed $\omega\neq 0$, we make use of a simple but important observation due to Renner~\cite[Lemma~4.1.5]{RennerPhD}, namely,\footnote{The alphabet underlying the definition of $\mathcal{T}_n$ now has cardinality $d^2$, so we can think of types $t_n\in \mathcal{T}_n$ as functions $t_n:\{0,\ldots, d^2-1\} \to [0,1]$.}
\bb
\VVsym(\ket{0}; n,r) &= \Span\left\{ \ket{n,t_n}:\ t_n \in \mathcal{T}_n\, ,\ \sumno_{x=1}^{d^2-1} nt_n(x) \leq r\right\} \\
&=  \Span\left\{ \ket{n,t_n}:\ t_n \in \mathcal{T}_n\, ,\ nt_n(0) \geq n-r\right\} .
\label{almost_iid_blurring_proof_eq6}
\ee
A brief self-contained justification of~\eqref{almost_iid_blurring_proof_eq6}, essentially identical to Renner's, is as follows: first, if $t_n\in\mathcal{T}_n$ satisfies that $\sum_{x\neq 0} nt_n(x) \leq r$, then clearly $\ket{n,t_n}\in \VVsym(\ket{0};n,r)$, by construction; second, any $\ket{\Psi}\in \VVsym(\ket{0}; n,r)\subseteq \mathrm{Sym}^n\big(\C^{d^2}\big)$ can be written as a linear combination of the vectors $\ket{n,t_n}$, which form a basis of the symmetric space; however, since it must hold that $\braket{n,s_n|\Psi} = 0$ for all $s_n\in\mathcal{T}_n$ such that $\sum_{x\neq 0} ns_n(x) > r$, only the types obeying the condition in~\eqref{almost_iid_blurring_proof_eq6} can actually appear in the decomposition of $\ket{\Psi}$.

An important consequence of~\eqref{almost_iid_blurring_proof_eq6} is that 
\bb
U_n \big(\VVsym(\ket{0};n,r)\big) = \Span\left\{ \ket{k}:\ k\in \N^{d^2-1}\! ,\ \sumno_{x=1}^{d^2-1} k(x) \leq r \right\} \eqqcolon \widetilde{\VV}_r
\label{almost_iid_blurring_proof_eq7}
\ee
for all $n\geq r$. The dimension of $\widetilde{\VV}_r$ is easy to estimate, because if $\sumno_{x=1}^{d^2-1} k(x) \leq r$ then, for all $x\neq 0$, $k(x)$ can take at most $r+1$ distinct values, thus entailing that
\bb
\dim \widetilde{\VV}_r \leq (r+1)^{d^2-1} .
\ee
Call $Q_r$ the orthogonal projector onto $\widetilde{\VV}_r$, and set $\ket{\widetilde{\Gamma}_n}\coloneqq U_n \ket{\Gamma_n}$. We have
\begin{align}
\Tr Q_r\omega\ &\eqt{(i)}\ \lim_{n\in I'} \Tr Q_r \widetilde{\phi}_n \nonumber \\
&\geqt{(ii)}\ \limsup_{n\in I'} \left( \Tr Q_r \widetilde{\Gamma}_n - \frac12 \left\| \widetilde{\phi}_n - \widetilde{\Gamma}_n \right\|_1 \right) \nonumber \\
&\eqt{(iii)}\ \limsup_{n\in I'} \left( 1 - \frac12 \left\| \phi_n - \Gamma_n \right\|_1 \right) \\
&\geqt{(iv)}\ 1 - \sqrt{\e(2-\e)} \nonumber \\
&>\ 0\, . \nonumber
\end{align}
Here, (i)~holds due to weak* convergence, thanks to the fact that $Q_r$ has finite rank; (ii)~is a simple application of Lemma~\ref{variational_program_trace_X_plus_lemma}, and in particular of~\eqref{second_variational_program_trace_X_plus}, once one observes that $\Tr \left(\widetilde{\phi}_n - \widetilde{\Gamma}_n\right)_+ = \frac12 \left\| \widetilde{\phi}_n - \widetilde{\Gamma}_n \right\|_1$ due to normalisation; in~(iii) we noticed that $\ket{\widetilde{\Gamma}_n} \in \widetilde{\VV}_r$ lies in the support of $Q_r$ due to~\eqref{almost_iid_blurring_proof_eq1} and~\eqref{almost_iid_blurring_proof_eq7}, and moreover the isometry $U_n$ cannot change the trace norm; finally, in~(iv) we used the estimate~\eqref{almost_iid_blurring_proof_eq3}.
\end{proof}

We are now ready to present the proof of Theorem~\ref{almost_iid_GQSL_thm}.

\begin{proof}[Proof of Theorem~\ref{almost_iid_GQSL_thm}]
By taking the i.i.d.\ ansatz $\rho^{\otimes n} \in \aiid{\rho,r}_n$ for all $n$, it is not difficult to verify that $\stein(\aiid{\rho,r}\|\FF)\leq \stein(\rho\|\FF) = D^\infty(\rho\|\FF)$. Therefore, we only need to establish the opposite inequality. As before, following Remark~\ref{reformulations_rem} we prove instead that
\bb
\liminf_{n\to\infty} \frac1n\, \rel{D_{\max}^\e}{\aiid{\rho,r}_n}{\FF_n} \geq D^\infty(\rho\|\FF) \qquad \forall\ \e\in (0,1)\, ,
\label{almost_iid_GQSL_proof_eq1}
\ee
where
\bb
\rel{D_{\max}^\e}{\aiid{\rho,r}_n}{\FF_n} \coloneqq \inf_{\omega_n \in \aiid{\rho,r}_n\!\!,\ \sigma_n\in \FF_n} D_{\max}^\e(\omega_n\|\sigma_n)\, . 
\ee
We proceed by contradiction. Assume that there exist some $\e\in (0,1)$ and an infinite set $I\subseteq [n]$ with the property that
\bb
\lim_{n\in I} \frac1n\, \rel{D_{\max}^\e}{\aiid{\rho,r}_n}{\FF_n} < \lambda < D^\infty(\rho\|\FF)\, .
\ee
This means that for all sufficiently large $n\in I$ we can find states $\rho_n$ such that
\bb
\frac12 \left\|\rho_n - \aiid{\rho,r}_n \right\|_1 \leq \e\, ,\qquad D_{\max}(\rho_n \| \FF_n)\leq n\lambda\, .
\ee
Without loss of generality, up to symmetrising we can assume that $\rho_n$ is permutationally symmetric for every $n$ --- remember that $\aiid{\rho,r}_n$ contains only permutationally symmetric states (Remark~\ref{almost_power_perm_symm_rem}). Hence, by Lemma~\ref{almost_iid_blurring_lemma}, there exists an infinite subset $I'\subseteq I$ such that
\bb
\lim_{M\to\infty} \limsup_{n\in I'} \Tr \left( \rho^{\otimes n} - M \widetilde{\rho}_n \right)_+ = 0\, ,
\label{almost_iid_GQSL_proof_eq4}
\ee
where
\bb
\widetilde{\rho}_n \coloneqq \int_0^\Delta \frac{\dd\delta}{\Delta}\ \widebar{B}^{\,\rho}_{n,\delta}(\rho_n)\, .
\ee
We can now proceed exactly as in the above proof of the standard, i.i.d.\ version of the generalised quantum Stein's lemma. There, we have already shown that~\eqref{almost_iid_GQSL_proof_eq4} leads to a contradiction, and the same exact reasoning applies here. Note that the fact that $\rho_n$ was $\e$-away from $\rho^{\otimes n}$ in trace norm, which is precisely the assumption that we are dropping here, was never used in that part of the proof --- it was only employed before, in order to apply the quantum blurring lemma (Lemma~\ref{quantum_blurring_lemma}).
\end{proof}

As a consequence, we immediately deduce a slightly weaker form of~\cite[Lemma~III.7]{Brandao2010}. The main difference lies in the fact that we need to assume that the number of defective sites, $r$, is a fixed (but arbitrary) constant for all $n$, while in the original statement it can grow with $n$, as long as it does so sub-linearly. In any case, we consider the following statement some positive evidence that~\cite[Lemma~III.7]{Brandao2010}, whose original proof is compromised because of the issue found in~\cite{gap}, might be correct after all. If this could be shown to be the case, the whole proof strategy of~\cite{Brandao2010} could be salvaged by by-passing the problematic Lemma~III.9 altogether.

\begin{cor}
Let $\HH$ be a finite-dimensional Hilbert space, and let $(\FF_n)_n$ be a sequence of sets of states $\FF_n\subseteq \D\big(\HH^{\otimes n}\big)$ that obeys the Brand\~{a}o--Plenio axioms. Then, for all $\rho\in \D(\HH)$ and all fixed $r\in \N$, the relative entropy of resource minimised over all $(n,r)$-almost power states along $\rho$ is asymptotically identical to the regularised relative entropy of $\rho$ itself:
\bb
\lim_{n\to\infty} \frac1n \inf_{\omega_n \in \aiid{\rho,r}_n\!\!,\ \sigma_n\in \FF_n} D(\omega_n \| \sigma_n) = D^\infty(\rho \| \FF)\, .
\ee
\end{cor}

\begin{proof}
The fact that 
\bb
\limsup_{n\to\infty} \frac1n \inf_{\omega_n \in \aiid{\rho,r}_n\!\!,\ \sigma_n\in \FF_n} D(\omega_n \| \sigma_n) \leq D^\infty(\rho \| \FF)
\ee
follows by choosing the trivial ansatz $\omega_n = \rho^{\otimes n}$. We will therefore focus on the converse inequality, whose proof follows a standard argument. For any $\omega_n \in \aiid{\rho,r}_n$, $\sigma_n\in \FF_n$, and $\e\in (0,1)$, let $Q_n$ be such that $0\leq Q_n\leq \id$, $\Tr Q_n \omega_n = 1-\e$, and $\Tr Q_n \sigma_n = 2^{-D_H^\e(\omega_n\|\sigma_n)}$. For $p,q\in [0,1]$, define $D_2(p\|q) \coloneqq p \log \frac{p}{q} + (1-p) \log\frac{1-p}{1-q}$, and observe that
\bb
D_2(p\|q) &= - h_2(p) + p \log\frac{1}{q} + (1-p) \log\frac{1}{1-q} \\
&\geq -1 + p \log\frac{1}{q}\, .
\ee
Due to the data processing inequality for the relative entropy, we have 
\bb
D(\omega_n\|\sigma_n) &\geq D_2\big( \Tr Q_n \omega_n\, \big\| \Tr Q_n \sigma_n \big) \\
&\geq -1 + \left(\Tr Q_n \omega_n\right) \log \frac{1}{\Tr Q_n \sigma_n} \\
&= -1 + (1-\e)\, D_H^\e(\omega_n\|\sigma_n)\, .
\ee
The claim follows from Theorem~\ref{almost_iid_GQSL_thm} by minimising over $\omega_n$ and $\sigma_n$, dividing by $n$, and taking first the limit $n\to\infty$ and then the limit $\e\to 0^+$.
\end{proof}

\vspace{2ex}
\noindent \emph{Acknowledgements.} I am grateful to Mario Berta, Fernando G.\ S.\ L.\ Brand\~{a}o, Martin Plenio, Marco Tomamichel, and especially to Bartosz Regula for the many draining but illuminating discussions on the topic of this paper that we have had over the past few years. I also thank Renato Renner for presenting his ideas on almost-i.i.d.\ states at the workshop `Bridging Quantum Information and Mathematical Physics', held at the University of Cambridge in August 2024. I acknowledge financial support from the European Union under the European Research Council (StG ETQO no.~101165230) and from MIUR (Ministero dell'Istruzione, dell'Universit\`a e della Ricerca) through the project `Dipartimenti di Eccellenza 2023--2027' of the `Classe di Scienze' department at the Scuola Normale Superiore.

\bibliography{../../biblio}

\end{document}

%% file: macros.tex
\usepackage{newpxtext,newpxmath}

\usepackage[latin1]{inputenc}
\usepackage{amsthm}
\usepackage{amssymb}
\usepackage{amsmath}
\usepackage{bbold}
\usepackage{bbm}
\usepackage[pdftex, backref=page]{hyperref}
\usepackage{braket}
\usepackage{dsfont}
\usepackage{mathdots}
\usepackage{mathtools}
\usepackage{enumerate}
\usepackage[shortlabels]{enumitem}
\usepackage{csquotes}
\usepackage{stmaryrd}
\usepackage[cal=boondox]{mathalfa}
\usepackage{graphicx}
\usepackage{stackengine}
\usepackage{scalerel}
\usepackage{tensor}       %\tensor[_n]{\braket{j|\psi}}{_{1\ldots n}}
\usepackage{array}
\usepackage{makecell}
\newcolumntype{x}[1]{>{\centering\arraybackslash}p{#1}}
\usepackage{tikz}
\usepackage{pgfplots}
\usetikzlibrary{shapes.geometric, shapes.misc, positioning, arrows, arrows.meta, decorations.pathreplacing, decorations.pathmorphing, patterns, angles, quotes, calc}
\usepackage{booktabs}
\usepackage{xfrac}
\usepackage{siunitx}
\usepackage{centernot}
\usepackage{comment}
\usepackage{chngcntr}
\usepackage[justification=justified,format=plain]{caption}
\usepackage{subcaption}
\usepackage{ragged2e}

\newtheorem{thm}{Theorem}
\newtheorem*{thm*}{Theorem}
\newtheorem{prop}[thm]{Proposition}
\newtheorem*{prop*}{Proposition}
\newtheorem{lemma}[thm]{Lemma}
\newtheorem*{lemma*}{Lemma}
\newtheorem{cor}[thm]{Corollary}
\newtheorem*{cor*}{Corollary}

\newtheorem*{cj*}{Conjecture}
\newtheorem{Def}[thm]{Definition}
\newtheorem*{Def*}{Definition}

\newtheorem*{question*}{Question}

\newtheorem*{problem*}{Problem}

\makeatletter
\def\thmhead@plain#1#2#3{%
  \thmname{#1}\thmnumber{\@ifnotempty{#1}{ }\@upn{#2}}%
  \thmnote{ {\the\thm@notefont#3}}}
\let\thmhead\thmhead@plain
\makeatother

\theoremstyle{definition}
\newtheorem{rem}[thm]{Remark}
\newtheorem*{note}{Note}

\newcommand{\bb}{\begin{equation}\begin{aligned}\hspace{0pt}}
\newcommand{\bbb}{\begin{equation*}\begin{aligned}}
\newcommand{\ee}{\end{aligned}\end{equation}}
\newcommand{\eee}{\end{aligned}\end{equation*}}
\newcommand\floor[1]{\lfloor#1\rfloor}
\newcommand\ceil[1]{\left\lceil#1\right\rceil}
\newcommand{\eqt}[1]{\stackrel{\mathclap{\scriptsize \mbox{#1}}}{=}}
\newcommand{\leqt}[1]{\stackrel{\mathclap{\scriptsize \mbox{#1}}}{\leq}}

\newcommand{\geqt}[1]{\stackrel{\mathclap{\scriptsize \mbox{#1}}}{\geq}}
\newcommand{\ketbra}[1]{\ket{#1}\!\!\bra{#1}}

\newcommand{\ketbraa}[2]{\ket{#1}\!\!\bra{#2}}
\newcommand{\Ketbraa}[2]{\Ket{#1}\!\!\Bra{#2}}
\newcommand{\sumno}{\sum\nolimits}

\newcommand{\e}{\varepsilon}
\renewcommand{\epsilon}{\varepsilon}

\newcommand{\dd}{\mathrm{d}}

\newcommand{\id}{\mathds{1}}
\newcommand{\R}{\mathds{R}}
\newcommand{\N}{\mathds{N}}

\newcommand{\C}{\mathds{C}}

\newcommand{\locc}{\mathrm{LOCC}}
\newcommand{\sep}{\mathrm{SEP}}

\newcommand{\SEP}{\pazocal{S}}

\DeclareMathOperator{\Tr}{Tr}

\DeclareMathOperator{\co}{conv}

\DeclareMathOperator{\Span}{span}
\DeclareMathAlphabet{\pazocal}{OMS}{zplm}{m}{n}

\DeclareMathOperator{\pr}{Pr}
\DeclareMathOperator{\supp}{supp}
\DeclareMathOperator{\spec}{spec}

\DeclareMathOperator{\tr}{tr}

\newcommand{\HH}{\pazocal{H}}
\newcommand{\T}{\pazocal{T}}

\newcommand{\EE}{\pazocal{E}}

\newcommand{\D}{\pazocal{D}}
\newcommand{\K}{\pazocal{K}}

\newcommand{\XX}{\pazocal{X}}
\newcommand{\TT}{\pazocal{T}}
\newcommand{\PP}{\pazocal{P}}
\newcommand{\FF}{\pazocal{F}}

\newcommand{\lsmatrix}{\left(\begin{smallmatrix}}
\newcommand{\rsmatrix}{\end{smallmatrix}\right)}

\newcommand{\deff}[1]{\textbf{\emph{#1}}}

\stackMath
\newcommand\xxrightarrow[2][]{\mathrel{%
  \setbox2=\hbox{\stackon{\scriptstyle#1}{\scriptstyle#2}}%
  \stackunder[5pt]{%
    \xrightarrow{\makebox[\dimexpr\wd2\relax]{$\scriptstyle#2$}}%
  }{%
   \scriptstyle#1\,%
  }%
}}

\newcommand{\tendsn}[1]{\xxrightarrow[\! n\rightarrow \infty\!]{\mathrm{#1}}}
\newcommand{\tendsk}[1]{\xxrightarrow[\! k\rightarrow \infty\!]{\mathrm{#1}}}
\newcommand{\ctends}[3]{\xxrightarrow[\raisebox{#3}{$\scriptstyle #2$}]{\raisebox{-0.7pt}{$\scriptstyle \mathrm{#1}$}}}

\stackMath

\makeatletter
\newcommand*\rel@kern[1]{\kern#1\dimexpr\macc@kerna}
\newcommand*\widebar[1]{%
  \begingroup
  \def\mathaccent##1##2{%
    \rel@kern{0.8}%
    \overline{\rel@kern{-0.8}\macc@nucleus\rel@kern{0.2}}%
    \rel@kern{-0.2}%
  }%
  \macc@depth\@ne
  \let\math@bgroup\@empty \let\math@egroup\macc@set@skewchar
  \mathsurround\z@ \frozen@everymath{\mathgroup\macc@group\relax}%
  \macc@set@skewchar\relax
  \let\mathaccentV\macc@nested@a
  \macc@nested@a\relax111{#1}%
  \endgroup
}

\counterwithin*{equation}{part}
\counterwithin*{thm}{part}
\counterwithin*{figure}{part}

\tikzset{meter/.append style={draw, inner sep=10, rectangle, font=\vphantom{A}, minimum width=30, line width=.8, path picture={\draw[black] ([shift={(.1,.3)}]path picture bounding box.south west) to[bend left=50] ([shift={(-.1,.3)}]path picture bounding box.south east);\draw[black,-latex] ([shift={(0,.1)}]path picture bounding box.south) -- ([shift={(.3,-.1)}]path picture bounding box.north);}}}
\tikzset{roundnode/.append style={circle, draw=black, fill=gray!20, thick, minimum size=10mm}}
\tikzset{squarenode/.style={rectangle, draw=black, fill=none, thick, minimum size=10mm}}

\definecolor{Blues5seq1}{RGB}{239,243,255}
\definecolor{Blues5seq2}{RGB}{189,215,231}
\definecolor{Blues5seq3}{RGB}{107,174,214}
\definecolor{Blues5seq4}{RGB}{49,130,189}
\definecolor{Blues5seq5}{RGB}{8,81,156}

\definecolor{Greens5seq1}{RGB}{237,248,233}
\definecolor{Greens5seq2}{RGB}{186,228,179}
\definecolor{Greens5seq3}{RGB}{116,196,118}
\definecolor{Greens5seq4}{RGB}{49,163,84}
\definecolor{Greens5seq5}{RGB}{0,109,44}

\definecolor{Reds5seq1}{RGB}{254,229,217}
\definecolor{Reds5seq2}{RGB}{252,174,145}
\definecolor{Reds5seq3}{RGB}{251,106,74}
\definecolor{Reds5seq4}{RGB}{222,45,38}
\definecolor{Reds5seq5}{RGB}{165,15,21}

\allowdisplaybreaks